\documentclass{article}

\usepackage{arxiv}

\usepackage[utf8]{inputenc} 
\usepackage[T1]{fontenc}    
\usepackage{hyperref}       
\usepackage{url}            
\usepackage{booktabs}       
\usepackage{amsfonts}       
\usepackage{nicefrac}       
\usepackage{microtype}      
\usepackage{lipsum}		
\usepackage{graphicx}
\usepackage{natbib}
\usepackage{doi}

\usepackage{mathtools,upgreek,eucal,mathrsfs,enumitem,amsthm,amsmath}
\usepackage{times,epsfig,makeidx,amsfonts,amsmath,dcolumn,multirow,amssymb,colortbl,bm,url}
\usepackage{algorithm}
\usepackage[]{algorithmic}

\DeclareMathOperator{\lOp}{\mathcal{O}_p}

\DeclareMathOperator{\lO}{\mathcal{O}}
\newcommand\lo{
  \mathchoice
    {{\scriptstyle\mathcal{O}}}
    {{\scriptstyle\mathcal{O}}}
    {{\scriptscriptstyle\mathcal{O}}}
    {\scalebox{.5}{$\scriptscriptstyle\mathcal{O}$}}
  }
\newtheorem{proposition}{Proposition}

\numberwithin{equation}{section}

\newtheorem{example}{Example}
\newtheorem{scenario}{Scenario}
\newtheorem{lemma}{Lemma}

\title{Additive Bayesian variable selection under censoring and misspecification}

\author{ \href{https://orcid.org/0000-0000-0000-0000}{\includegraphics[scale=0.06]{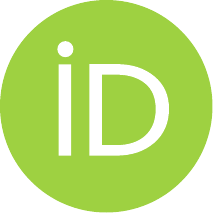}\hspace{1mm}David Rossell}\\
Universitat Pompeu Fabra, \\
Department of Business and Economics, \\
Barcelona, Spain.\\
	\texttt{david.rossell@upf.edu} \\
	\And
	\href{https://orcid.org/0000-0001-7183-8407}{\includegraphics[scale=0.06]{orcid.pdf}\hspace{1mm}Francisco J. Rubio} \\
	University College London \\
	Department of Statistical Science,\\
	London, UK.\\
	\texttt{f.j.rubio@ucl.ac.uk} \\
}

\renewcommand{\shorttitle}{\textit{arXiv} Template}

\hypersetup{
pdftitle={Additive Variable Selection},
pdfsubject={stats},
pdfauthor={Rossell and Rubio},
pdfkeywords={First keyword, Second keyword, More},
}

\begin{document}
\maketitle

\begin{abstract}
We discuss the role of misspecification and censoring on Bayesian model selection in  the contexts of right-censored  survival and concave log-likelihood regression.  Misspecification includes wrongly assuming the censoring mechanism to be non-informative.  Emphasis is placed on additive accelerated failure time, Cox proportional hazards and probit models. We offer a theoretical treatment that includes local and non-local priors, and a general non-linear effect decomposition to improve power-sparsity trade-offs. We discuss a fundamental question: what solution can one hope to obtain when (inevitably) models are misspecified, and  how to interpret it? Asymptotically, covariates that do not have predictive power for neither the outcome nor (for survival data) censoring times, in the sense of reducing a likelihood-associated loss, are discarded.  Misspecification and censoring have an asymptotically negligible effect on false positives, but their impact on power is exponential. We  show that it can be advantageous to consider  simple models that are computationally practical yet attain good power to detect potentially complex effects,  including the use of finite-dimensional basis to detect truly non-parametric effects.  We also discuss algorithms to capitalize on sufficient statistics and fast likelihood approximations  for Gaussian-based survival and binary  models. 
\end{abstract}


\section{Introduction}\label{sec:intro}

Determining what covariates have an effect on a survival (time-to-event) outcome 
is an important task in many fields, including Biomedicine, Economics and Engineering.
For interpretability and computational convenience it is common to use parametric and semi-parametric models
such as Cox proportional hazards \citep{cox:1972} or accelerated failure time (AFT) regression for survival outcomes,  possibly with non-linear additive effects. 
The proportional hazards model assumes that covariates have a multiplicative effect on the baseline hazard, whereas in AFT models covariates drive the mean of the logarithmic (or other monotonic transform) times-to-event.
These models can be combined with Bayesian model selection to provide a powerful mechanism to select variables, enforce sparsity and quantify uncertainty.
However, the precise consequences of model misspecification and censoring are not sufficiently understood.
By misspecification we mean that the data are truly generated by a distribution outside the considered class.
For instance, one may fail to record truly relevant covariates or represent their effects imperfectly, 
 \textit{e.g.} when using a Cox model but the true covariate effects on the hazard are non-proportional. 
This issue can be addressed by enriching the model, \textit{e.g.} via non-parametric or time-dependent effects.
Then, the potential concerns are that the larger number of parameters can adversely affect inference, unless the sample size is large enough, and that computation can be costlier.
Censoring is also important. First, it reduces the effective sample size.  Second, wrongly assuming the censoring mechanism to be non-informative, \textit{i.e.} independent of the outcome conditionally on covariates, may affect the selected model, even asymptotically.

 Our goal is to help understand the consequences of three important issues on model selection: misspecification, censoring, and trade-offs when including non-linear effects.
We first consider that the data analyst assumed a non-linear additive AFT model, or an additive Cox model, but data are truly generated by a different probability distribution $F_0$. We also consider probit regression, which can be formulated as a particular case of the Normal AFT model, and more general concave log-likelihood regression,  which provides a unifying framework for the models we consider here. 

There are many data analysis methods for survival outcomes, along with theory for well-specified models and empirical results suggesting potential issues under misspecification, but  their implications for model selection  have not been described in sufficient detail.
 We first review results on the behavior of misspecified AFT and proportional hazard models, and subsequently discuss some model selection methods for survival data. 
Although both models have similar asymptotic properties, and which model is more appropriate depends on the data at hand,
AFT inference has been argued to be more robust and to better preserve interpretability under misspecification.
More precisely, the maximum likelihood estimators under misspecified Cox and AFT models
have comparable limiting distributions if censoring is absent or independent of covariates \citep{struthers:1986,yingzhiliang:1993}, but not so under covariate-dependent censoring \citep{solomon:1984}.
 Covariate-dependent censoring also affects  frequentist hypothesis tests. In misspecified Cox models it can lead to a substantial type I error inflation \cite{dirienzo:2001}.
In misspecified AFT models 
the power of the tests may be affected, but simple strategies to control the type I error are available \citep{solomon:1984,hutton:2002,hattori:2012}. 

 Another situation where both models behave differently  is when omitting truly active covariates, {\it e.g.} because these were not recorded.
A proportional hazards model with omitted variables tends to underestimate covariate effects, even for a treatment of interest that is uncorrelated with other covariates
\citep{struthers:1986,keiding:1997}. Further, even if the data-generating truth has proportional hazards, the marginal model conditioning only on the observed covariates does not (except in positive stable distributions, \cite{hougaard:1995}).
 In contrast, if the data-generating truth is an AFT model and one omits relevant covariates,  the unaccounted variability is subsumed into the error term, and regression parameters remain interpretable as averaged effects across the population \citep{hutton:2002}. 
Note that omitting covariates is intimately connected to incorporating a covariate but misspecifying its effect:  using a linear or finite-dimensional effect can be seen as omitting a subset of the columns of the basis defining a truly non-parametric effect. Thus our discussion on omitted variables  applies directly to misspecifying covariate effects.
To summarize, censoring and misspecification have non-trivial effects on estimation and hypothesis testing.

 We now review some model selection methods for survival data, discussing the extent to which they considered misspecification. 
\cite{huang:2006,simon:2011,rahaman:2019} proposed likelihood penalties for Cox and semiparametric AFT models, 
and \cite{tong:2013} 
for broader generalized hazards models. 
Most of this work focused on linear covariate effects, computation and proving consistency under covariate-independent censoring.
There are however empirical results on the effect of misspecification, 
\textit{e.g.} \cite{zhangz:2018} 
showed in simulations an increase in false positives of the Cox-LASSO method \citep{tibshirani:1997} when data truly arise from an AFT model.
 There are also many Bayesian variable selection methods for survival data.  \cite{F98} and \cite{S06} proposed shrinkage priors for the Cox and AFT models and assessed performance via simulations where the model was well-specified. \cite{ibrahim:1999} studied Bayesian model selection for the Cox model, 
\cite{dunson:2005} for the so-called additive hazards model, 
and \cite{niko:2017} for the Cox model under non-local priors \citep{JR10,JR12}. 
See \cite{ibrahim:2014} for a review, with a focus on the Cox model.
 While interesting, these Bayesian proposals do not consider misspecification. 
 \cite{rossell:2018} did study misspecified Bayesian linear regression, showing that misspecification often reduces the power to detect active variables, but did not consider censoring. 

 We summarize our main messages.  We show that, under mild assumptions, Bayesian model selection asymptotically 
discards covariates that do not predict neither the outcome nor the censoring times.
 By \textit{predict}, we refer to increasing the expectation of the log-likelihood function. For any fully-specified model, said expectation is a weighted average of a reward for assigning a high probability to the observed censoring time in individuals that are censored, and a reward for predicting survival times accurately in uncensored individuals (\textit{e.g.} mean squared error, for Normal AFT models). For the partially-specified Cox model, the reward is for assigning a high hazard to individuals who experienced the event, relative to other individuals at risk.  
We discuss that both censoring and wrongly specifying covariate effects have an exponential effect in power,
but that asymptotically neither leads to false positive inflation.
We also develop a novel non-linear effect decomposition to ameliorate the power drop, 
and study the consequences of using finite basis to describe covariate effects, a practical strategy to speed up computations when one considers many models.
For concreteness, we outline a formulation based on a novel combination of non-local priors \citep{JR10} and group-Zellner priors that induce group-level sparsity for non-linear effects.
As a technical contribution, we prove the asymptotic validity of Laplace approximations to Bayes factors for concave log-likelihoods under minimal conditions, allowing for misspecification, which provides  a simple basis to study Bayes factors that covers all models considered in this paper.  We also provide software (R package \texttt{mombf}).

The outline is as follows.
Section \ref{sec:model} discusses the likelihood for AFT and Cox models, priors and a  non-linear effect decomposition aimed at improving power.  
Section \ref{sec:theory} discusses known and novel results on asymptotic normality and Bayes factor rates, and how to interpret the Bayesian model selection solution under misspecification.
 Similar results are obtained for general concave log-likelihoods, see Section \ref{sec:theory_logconcave}. 
See also Section \ref{app:Probit} for a known but seemingly unexploited result in the literature, that probit models are a particular case of the Normal AFT model.
Section \ref{sec:computation} discusses the relative computational convenience of AFT vs. Cox models 
related to the use of sufficient statistics.
It also discusses an approximation to the Normal log-distribution function and derivatives that significantly increases speed and accuracy, and may have some independent interest, and simple model exploration strategies.
Section \ref{sec:results} illustrates the effect of misspecification and censoring in simulations and cancer datasets, practical power-sparsity trade-offs, and the use of finite-dimensional non-linear basis.
Section \ref{sec:discussion} concludes.
The supplementary material contains derivations related to the likelihood, priors and their derivatives, and prior elicitation (Sections \ref{suppsec:noninformati_censoring}-\ref{app:prior_elicitation}).
Sections \ref{app:withinmodel_comput}-\ref{app:cda} offer detailed discussions on computational algorithms, including a novel approximation to Normal log-distribution functions that may have some independent interest. Section \ref{sec:proof_asymp} contains all proofs for our main results, Sections \ref{app:AFTLaplace}-\ref{app:Probit} additional propositions for the AFT model with Laplace errors and probit models, and Section \ref{sec:theory_logconcave} on the asymptotic validity of Laplace approximations to integrated likelihoods. Finally, Section \ref{app:figs_tables} contains empirical results that supplement those in the main paper.

\section{Formulation}
\label{sec:model}

 Our discussion focuses on survival data, but see Section \ref{app:Probit} for binary regression and Section \ref{sec:theory_logconcave} for concave log-likelihoods.
Section \ref{ssec:likelihood} sets notation, reviews the AFT and proportional hazards models, their being a particular cases of the generalized hazards structure, and a non-linear effects decomposition to improve interpretability and power. Section \ref{ssec:model_selection} embeds the problem within a Bayesian model selection framework. Section \ref{ssec:elicitation} introduces prior distributions that can accommodate group and hierarchical constraints, and Section \ref{sec:elicitation} suggests default prior parameter values.

\subsection{Likelihood}
\label{ssec:likelihood}

 Let us introduce the notation. Suppose that one is interested in studying the dependence of a  survival (or time-to-event) outcome $o_i \in \mathbb{R}_+$  on a covariate vector  $x_i= (x_{i1},\ldots,x_{ip})^{\top} \in{\mathbb R}^p$, for individuals $i=1,\ldots,n$.
 Suppose that there are right-censoring times $c_i \in \mathbb{R}_+$, such that one only observes the outcome for uncensored individuals, \textit{i.e.} those for which $o_i \leq c_i$.
Denote by $u_i=\mbox{I}(o_i<c_i)$ the indicator that observation $i$ is uncensored, $y_i=\min\{\log(o_i),\log(c_i)\}$ the observed log-times,
$y=(y_1,\ldots,y_n)$, $u=(u_1,\ldots,u_n)$, and the number of uncensored individuals $n_o= \sum_{i=1}^n u_i$.

 We review two popular models for survival data, the AFT and Cox models, and discuss a strategy to decompose non-linear effects. 
The AFT model postulates
\begin{eqnarray}\nonumber
\log(o_i) = \sum_{j=1}^p g_j(x_{ij}) + \epsilon_i,
\end{eqnarray}
where $g_j:\mathbb{R} \rightarrow \mathbb{R}$ and $\epsilon_i$ are
independent across $i=1,\ldots,n$ with mean $E(\epsilon_i)=0$ and variance $V(\epsilon_i)=\sigma^2$ (assumed finite).
Typically, $g_j$ is expressed in terms of an $r$-dimensional basis, \textit{e.g.}~splines or wavelets \citep{wood:2017}.
 For interpretability and to gain power (see Section \ref{ssec:bfrates})  it is convenient to decompose $g_j$ into a linear and a deviation-from-linearity components.
To fix ideas, the cubic splines used in our examples consider
\begin{eqnarray}\label{eq:AFTModel}
\log(o_i) = x_i^{\top}\beta +  s_i^{\top} \delta + \epsilon_i,
\end{eqnarray}
where $\beta=(\beta_1,\ldots,\beta_p)^\top \in \mathbb{R}^p$, $\delta^\top=(\delta_1^\top,\ldots,\delta_p^\top) \in \mathbb{R}^{rp}$
and $s_i^\top=(s_{i1}^\top,\ldots,s_{ip}^\top)$, where $s_{ij} \in \mathbb{R}^r$ is the projection of $x_{ij}$ onto a cubic spline basis orthogonalized to $x_{ij}$ (and the intercept).
 The idea is that $x_i^\top \beta$ captures linear effects, whereas $s_i^{\top} \delta$ captures deviations from linearity. Even if a covariate truly has a non-linear effect, the linear term often captures a fraction of that effect using a single parameter, hence one can gain in power to detect its presence. 
Specifically we built $s_{ij}$, the $i^{th}$ row of the $n \times r$ matrix $S_j$, as follows.
Let $X_j$ and $\widetilde{S}_j$ have row $i$ equal to $(1,x_{ij})$ and
the cubic spline projection of $x_{ij}$ (equi-distant knots), then $S_j= (I- X_j (X_j^\top X_j)^{-1} X_j^\top) \widetilde{S}_j$ is orthogonal to $X_j$.
Denote by $(X,S)$ the design matrix with $(x_i^\top,s_i^\top)$ in its $i^{th}$ row,
and by $(X_o,S_o)$ and $(X_c,S_c)$ the submatrices with the rows for uncensored and censored individuals (respectively).
\color{black} This formulation contains partially linear models as particular cases (\textit{i.e.} when only some of the covariates are assumed to have a non-linear effect). \color{black}
We denote the parameter space by $\Gamma\subset {\mathbb R}^{p(r+1)}\times{\mathbb R}^+$.

In survival analysis it is common to pose a model only for the survival times, such as \eqref{eq:AFTModel}.
This is because the censoring is assumed to be non-informative given the covariates, and then the censoring distribution factors out of the likelihood function (see Section \ref{suppsec:noninformati_censoring} for a brief derivation).
The likelihood and partial likelihood used by the AFT and Cox models, which we review next, embed such a non-informativeness assumption.
See Section \ref{sec:theory} for a discussion on the consequences of this assumption not holding.

Regarding the likelihood associated to \eqref{eq:AFTModel}, consider the particular case where the errors are Gaussian. It is convenient to reparameterize $\alpha = \beta/\sigma$, $\kappa= \delta/\sigma$ and $\tau = 1/\sigma$, as then the log-likelihood is concave, provided the number of uncensored individuals is greater than the number of model parameters ($n_o \geq p+ r p$) and that $(X_{o},S_{o})$ has full column rank \citep{B81,silvapulle:1986}.
The log-likelihood is
\begin{eqnarray}
\ell(\alpha,\kappa,\tau) &=& -\frac{n_o}{2} \log\left(\frac{2\pi}{\tau^2}\right) - \dfrac{1}{2} \sum_{u_i=1}(\tau y_i-x_{i}^{\top}\alpha -s_{i}^{\top}\kappa)^2 \nonumber\\
 &+& \sum_{u_i=0} \log \left\{\Phi\left(x_{i}^{\top}\alpha +s_{i}^{\top}\kappa - \tau y_i\right)\right\},
 \label{logLikeAFT2}
\end{eqnarray}
see Supplementary \eqref{eq:grad_logLikeAFT2} for its gradient and hessian. 

The Cox model instead assumes that the hazard function at time $t$ takes the form
\begin{eqnarray*}
h_{PH}(t\mid x_i) = h_0(t)\exp\left\{x_i^{\top}\beta +  s_i^{\top} \delta \right\},
\end{eqnarray*}
where $h_0(\cdot)$ is a baseline hazard, typically estimated non-parametrically, and $(\beta,\delta)$ are estimated using the log partial likelihood \citep{cox:1972}
\begin{eqnarray}\label{eq:ploglik}
\ell_p(\beta,\delta) =  \sum_{u_i=1} \log\left( \frac{\exp \left\{  x_i^{\top}\beta + s_i^{\top} \delta \right\}}{ \sum_{k \in {\mathcal R}(o_i)} \exp\left\{x_k^{\top}\beta +  s_k^{\top} \delta\right\}} \right),
\end{eqnarray}
where ${\mathcal R}(t) = \{i : o_i \geq t\}$ denotes the set of individuals at risk at time $t$.

To relate both models, \eqref{eq:AFTModel} can be formulated in terms of the hazard function $h_{AFT}(t) = h_0\left(t \exp\left\{x_i^{\top}\beta + s_i^{\top} \delta \right\}\right) \exp\left\{x_i^{\top}\beta + s_i^{\top} \delta \right\}$. 
Both models are special cases of the generalized hazards structure \citep{chen:2001}
\begin{eqnarray}\label{eq:GH}
h_{GH}(t) = h_0\left(t \exp\left\{x_i^{\top}\beta + s_i^{\top}\delta \right\} \right) \exp\left\{x_i^{\top}\theta + s_i^{\top}\xi \right\},
\end{eqnarray}
which we use in our examples to portray the behaviour of misspecified AFT and Cox models.
Clearly, \eqref{eq:GH} contains the AFT model for $(\beta,\delta)=(\theta,\xi)$ 
and the proportional hazards model for $(\beta,\delta)=0$. 

\subsection{Model selection}
\label{ssec:model_selection}

Our goal is model selection, which we formalize as choosing among three possibilities
$$\gamma_j=\begin{cases}
0, \mbox{ if } \beta_j= 0, \delta_j=0, \\
1, \mbox{ if } \beta_j \neq 0, \delta_j=0 ,\\
2, \mbox{ if } \beta_j \neq 0, \delta_j\neq 0, \\
 \end{cases}$$
corresponding to no effect, a linear and a non-linear effect of each covariate $j=1,\ldots,p$.
That is, $\gamma=(\gamma_1,\ldots,\gamma_p)$ determines what covariates enter the model and their effect, and there are $3^p$ total models to consider. 
 We remark that by non-linear effect we refer to the specific effect coded by the chosen basis, \textit{e.g.} B-splines in our examples. One could extend the exercise by considering other types of non-linear effects, for example by adding a fourth possibility $\gamma_j=3$ associated to a wavelet basis. Such basis would be orthogonalized to the linear term, as described after \eqref{eq:AFTModel}. 

This formulation has two key ingredients.
First,  it decomposes effects into linear and deviation from linearity components, enforcing 
the hierarchical desiderata that the latter are only included if the linear terms are present.
This decomposition is similar to the structured additive regression of \cite{scheipl:2012}, the main difference is that they do not test for exact $\beta_j=0$, $\delta_j=0$
and that they rely on a spectral decomposition that is less general than our simpler orthogonalization of $(X_j,S_j)$.
 Our theory and results show that such decompositions improve the power to detect truly active effects. As discussed, this is because the option $\gamma_j=1$ captures part of the effect of a variable with a single parameter. 
The second ingredient is considering the group inclusion of all non-linear coefficients $\delta_j \in \mathbb{R}^r$.
The motivation is that including individual entries in $\delta_j$ increases the probability of false positives,
{\it e.g.}~if $j$ truly had no effect there would be $2^r-1$ subsets of $S_j$ leading to including $j$.

Bayesian model selection proceeds as follows.
Let $p_\gamma=\sum_{j=1}^{p} \mbox{I}(\gamma_j \neq 0)$ be the number of active variables according to model $\gamma$,
$s_\gamma= \sum_{j=1}^{p} \mbox{I}(\gamma_j=2)$ the number of non-linear effects,
and $d_\gamma= p_\gamma + r s_\gamma + 1$ the total number of parameters in $\gamma$ for AFT models,
and $d_\gamma= p_\gamma + r s_\gamma$ for Cox and probit models.
$(X_\gamma, S_\gamma)$ and $(\beta_\gamma, \delta_\gamma)$ are the corresponding submatrices of $(X,S)$ and subvectors of $(\beta, \delta)$,
and $(X_{o,\gamma},S_{o,\gamma})$ and $(X_{c,\gamma},S_{c,\gamma})$ the submatrices of  the observed  $(X_o,S_o)$  and censored  $(X_c,S_c)$ design matrices.
One then obtains posterior model probabilities
\begin{equation}\label{eq:modelposterior}
\pi(\gamma \mid y) = \frac{p(y \mid \gamma)\pi(\gamma)}{\sum_\gamma p(y \mid \gamma)\pi(\gamma)}
= \left( 1 + \sum_{\gamma' \neq \gamma}  B_{\gamma', \gamma} \frac{\pi(\gamma')}{\pi(\gamma)}  \right)^{-1},
\end{equation}
where $\pi(\gamma)$ is the model prior probability,
$B_{\gamma',\gamma}= p(y \mid \gamma')/p(y \mid \gamma)$ the Bayes factor between $(\gamma',\gamma)$ and
$$p(y\mid \gamma)= \int p(y \mid \alpha_{\gamma},\kappa_\gamma,\tau) \pi(\alpha_{\gamma},\kappa_\gamma,\tau \mid \gamma) d\alpha_{\gamma} d\kappa_\gamma d\tau,$$
the integrated likelihood $p(y \mid \alpha_{\gamma},\kappa_\gamma,\tau)$ with respect to a prior density $\pi(\alpha_{\gamma},\kappa_\gamma,\tau \mid \gamma)$.
One may choose the model with highest $\pi(\gamma \mid y)$,
variables with high marginal posterior probabilities $\pi(\gamma_j \neq 0 \mid y)$ and,
when the interest is in prediction, use Bayesian model averaging where models are weighted
according to $\pi(\gamma \mid y)$, or alternatively choosing a sparse model giving similar predictions \citep{hahn:2015}.
Either way $\pi(\gamma \mid y)$ are critical for inference, hence the importance to understand their behavior.

 To conclude, we comment upon a practically-relevant computational issue. 
In additive models, it is common to either let the basis dimension $r$ grow with $n$ and add a regularization term (\textit{e.g.}~P-splines), or to learn $r$ from the data (\textit{e.g.}~knot selection). Letting $r$ grow with $n$ is interesting theoretically and in prediction problems where one fits a single model, but less so when one considers many models. Large $r$ increases the computational cost (\textit{e.g.} matrix determinants require $r^3/3$ operations) and is often unneeded when the goal is just to detect if a covariate has an effect. Instead one may use a moderate $r$, \textit{e.g.}~misspecify the predictive-optimal model. The question is then, what answer can one hope to obtain and what are its properties.
Our theory and software allow learning $r$ among several fixed values, but in our examples a small $r=5$ provided better inference at lower cost (particularly for small $n$, \textit{e.g.} Figure \ref{fig:2vars_margpp}, bottom).

\subsection{Prior distributions}
\label{ssec:elicitation}

Although our discussion applies to a wide class of priors, we present three concrete options. 
\begin{eqnarray*}
{\pi_L(\alpha_{\gamma},\kappa_\gamma,\tau)}
  &=& \left[ \prod_{\gamma_j\geq 1}N\left(\alpha_j; 0, g_L n/ (x_j^\top x_j)\right)
      \prod_{\gamma_j=2} N\left(\kappa_j; 0, g_S n (S_j^\top S_j)^{-1}\right) \right]\pi(\tau)\\
  \pi_M(\alpha_{\gamma},\kappa_\gamma,\tau)
  &=& \left[ \prod_{\gamma_j\geq 1} \frac{\alpha_j^2}{g_M} N\left(\alpha_j; 0, g_M\right)
      \prod_{\gamma_j=2} N\left(\kappa_j; 0, g_S n (S_j^\top S_j)^{-1}\right) \right]\pi(\tau)\\
  \pi_E(\alpha_{\gamma},\kappa_\gamma,\tau)
  &=& \left[ \prod_{\gamma_j\geq 1} e^{\sqrt{2} - g_E/\alpha_j^2} N\left(\alpha_j; 0, g_E\right)
      \prod_{\gamma_j=2} N\left(\kappa_j; 0, g_S n (S_j^\top S_j)^{-1}\right) \right]\pi(\tau),
\end{eqnarray*}
where $\pi(\tau) = 2 \tau^{-3} \mbox{IG}(\tau^{-2}; a_\tau/2, b_\tau/2)$, and $\mbox{IG}$ denotes the inverse gamma density, and $g_L,g_S,g_M,g_E,a_\gamma,b_\tau \in \mathbb{R}_+$ are given dispersion parameters,  for which we propose default values in Section \ref{sec:elicitation}. 

 These choices include a standard Normal prior and two variations of non-local priors. The use of non-local priors can be argued from a foundational viewpoint, where one wishes to assign prior beliefs that are coherent with the parameters assumed non-zero by a given model \cite{JR10}. For our purposes, however, their main role is that they lead to faster Bayes factor rates to discard spurious parameters. See \cite{RT17} for further discussion.   
We refer to $\pi_L$ as group-Zellner prior. It is a product of Zellner priors across groups of linear and non-linear terms for each covariate.
This prior is local, \textit{i.e.}~it assigns non-zero density to $\alpha_\gamma$ having zeroes.
The Zellner structure is chosen for simplicity, our theory can be easily extended to other local priors,  provided they are continuous and positive at the asymptotically-optimal parameter values (Section \ref{sec:theory}, \cite{JR10}).  
The priors $\pi_M$ and $\pi_E$ are non-local with respect to $\alpha_\gamma$, the so-called product MOM and eMOM priors introduced in \cite{JR12,RTJ13}, and a group-Zellner prior on $\kappa_\gamma$.

Regarding  the prior on the models  $\pi(\gamma)$, we consider joint group inclusion of  non-linear coefficients  $\delta_j$ and the hierarchical restriction that their inclusion requires that of  the corresponding linear coefficient  $\beta_j$.
Letting $\pi(\gamma)$ depend only on the number of non-zero parameters in $(\beta_\gamma,\delta_\gamma)$, as customarily done when only linear effects are considered, would ignore such structure and hence be inadequate.
Instead, we let $\pi(\gamma)$ depend on the number of variables having linear and non-linear effects, $(p_\gamma,s_\gamma)$.
By default, we consider independent Beta-Binomial priors \cite{SB10}
\begin{align}
  \pi(\gamma)=
\frac{1}{C}  \mbox{BetaBin}(p_\gamma;p,a_1,b_1) {p \choose p_\gamma}^{-1} \mbox{BetaBin}(s_\gamma;s,a_2,b_2) {s \choose s_\gamma}^{-1},
\label{eq:priormodel}
\end{align}
where $\mbox{BetaBin}(z;p,a,b)$ is the probability of $z$ successes under a Beta-Binomial distribution with $p$ trials and parameters $(a,b)$
and $C$ a normalizing constant that does not need to be computed explicitly.
Any model such that  the number of parameters is  $p_\gamma + r s_\gamma>n$ is assigned $\pi(\gamma)=0$, as it would result in data interpolation.
By default we let $a_1=b_1=a_2=b_2=1$ akin to \cite{SB10},
{\it e.g.}~in the $p=1$ case these give $\pi(\gamma_1=0)=\pi(\gamma_1=1)=\pi(\gamma_1=2)=1/3$.
As alternatives to \eqref{eq:priormodel}, one can also consider Binomial priors where
$\mbox{BetaBin}(z;p,a_j,b_j)$ is replaced by $\mbox{Bin}(z;p,a_j)$ for a given success probability $a_j \in [0,1]$
and Complexity priors \cite{castillo:2015} where it is replaced by $1/p^{a_j z}$ for some constant $a_j>0$.
These two alternatives are implemented in our software and covered by our theory in Section \ref{sec:theory},
but for simplicity our examples focus on \eqref{eq:priormodel}.

%

\subsection{Prior elicitation}
\label{sec:elicitation}

The prior dispersion parameters $(g_L,g_M,g_E,g_S)$ are important for variable selection.
For instance, setting large dispersions helps induces sparsity,  particularly when they are allowed to grow with the sample size $n$ \citep{narisetty:2014}. 
However such large values also reduce power, see \cite{R18} and our Propositions \ref{prop:BFRatesI}, \ref{prop:BFRatesCox} and \ref{prop:BFRatesP},  and are harder to justify from the point of view that the expected effect sizes a priori should not depend on $n$. We briefly discuss default values that do not depend on $n$, and refer the reader to Section \ref{app:prior_elicitation} for details. 

 Specifying prior parameters provides an opportunity  to define what effects are practically relevant.
 Importantly, in what follows we assume that continuous covariates were standardized to unit variance, else the parameter interpretation and default values change. 
Basic considerations give a fairly narrow range of values that we deem reasonable in applications.
 For example, in AFT and Cox models $e^{|\beta_j|}$ define the effect size, when these are say $<15\%$ (\textit{i.e.}~$e^{\vert\beta_j\vert}<1.15$) they are typically practically irrelevant.
Based on these considerations,  our recommended defaults for AFT and Cox models are $g_M=0.192$, $g_E=0.091$, $g_L=1$, $g_S=1/r$ and $a_\tau=b_\tau=3$, whereas for probit regression they are $g_M=0.139$ and $g_E=0.048$.
 One should not take these defaults at their exact value, rather as defining a range of reasonable values. These ranges are discussed in 
Section \ref{app:prior_elicitation}. 
In our examples, results were robust to the prior dispersions, provided they stay within our recommended range. 

We remark that if one were to change the prior dispersion arbitrarily then results would be affected, in a similar manner to how regularization parameters affect penalized likelihood results.
However, in our view the prior beliefs implied by arbitrary prior dispersions would be unreasonable in most applications.
We also note that there is a wide objective Bayes literature on using the data to set the prior parameters, see \cite{consonni:2018} for an excellent review. We do not argue against such strategies,
but we focus on our defaults as a simple strategy that attains a fairly competitive performance in practice.

\section{Theory}
\label{sec:theory}

 This section describes the asymptotic solution returned by Bayesian model selection, when the observed data 
$(o_i,c_i,z_i) \sim F_0$ are independent realizations from some $F_0$, where $z_i \in {\mathbb R}^{p(r+1) + q}$ for $q \geq 0$ contains the observed covariates $(x_i, s_i) \in \mathbb{R}^{p(r+1)}$, and potentially also $q$  additional columns. These columns may contain covariates that were not recorded but are truly relevant for the outcome or the censoring, or non-linear effects and interactions missed by $(x_i,s_i)$. 
We do not assume $F_0$ to be parametric, rather it can be quite general,
and the whole model structure assumed by the analyst (\textit{e.g.} accelerated times, proportional hazards) may be wrong. 

 Section \ref{ssec:consistency} shows that when one assumes the Normal AFT model \eqref{eq:AFTModel} but truly $(o_i,c_i,z_i) \sim F_0$, the maximum likelihood estimator under each model $\gamma$ converges to an optimal $(\alpha_\gamma^*,\kappa_\gamma^*,\tau_\gamma^*)$ and is asymptotically normally-distributed.
See \cite{hjort:1992} and \cite{hjort:2011} for related asymptotic results, and Section \ref{app:AFTLaplace} for analogous results for the Laplace AFT model.
Section \ref{ssec:bfrates} shows that Bayesian model selection in the AFT model asymptotically returns the smallest $\gamma^*$ such that all effects in $(\alpha_\gamma^*,\kappa_\gamma^*)$ are non-zero.
Equivalently, $\gamma^*$ is defined by the zeroes in $(\alpha^*,\kappa^*)$, the optimal value under the full model including all parameters. 
Section \ref{ssec:bfrates_CoxPH} gives analogous results for Cox models.
These results are extended to probit models in Section \ref{app:Probit}, and in Section \ref{sec:theory_logconcave} to more general concave log-likelihood models. It is possible to derive similar results beyond the concave case, however this class encompasses all the models we consider here and allows simplifying the proofs and technical conditions.

Throughout we help interpret the solution and certain Bayes factors properties.
 Of particular relevance,  Section \ref{ssec:consistency} discusses that the asymptotic solution $\gamma^*$ excludes covariates that do not help predict the outcome nor the censoring times, and offers some examples.
Section \ref{ssec:bfrates} comments on potential advantages of using  low-dimensional basis and non-linear decompositions to detect covariate effects.

\subsection{Asymptotic solution in AFT models}
\label{ssec:consistency}

 As the sample size grows, Bayesian model selection recovers a model $\gamma^*$ that excludes parameters that are asymptotically estimated to be zero. Under mild regularity conditions, this limiting parameter is the value maximizing the expected log-likelihood under $F_0$. We start by defining the expected log-likelihood, then state the limiting result, and finally interpret its meaning and implications for model selection.  

Let $\eta_\gamma = (\alpha_{\gamma},\kappa_\gamma,\tau) \in \Gamma_\gamma$ be the \color{black} vector with $p_\gamma + r s_\gamma$ regression parameters under a given model $\gamma$ (Section \ref{ssec:model_selection}) plus the error variance, where $\Gamma_{\gamma}= {\mathbb R}^{p_\gamma + r s_\gamma}\times{\mathbb R}^+$ \color{black} is the corresponding parameter space. Let
\begin{eqnarray*}
 m(\eta_\gamma) &=& 
(1-u_1) \left[ \log\Phi\left(x_1^{\top}\alpha_\gamma + s_1^{\top} \kappa_\gamma - \tau \log(c_1)\right) \right]
 \nonumber \\
& + & u_1\left[ \log(\tau) -\dfrac{1}{2}\log(2\pi) - \dfrac{1}{2} \left(\tau \log(o_1)-x_1^{\top}\alpha_\gamma -s_1^{\top} \kappa_\gamma \right)^2  \right],
\end{eqnarray*}
the contribution of one observation to the log-likelihood \eqref{logLikeAFT2},  and 
\begin{align}
&M(\eta_\gamma)={\mathbb E}_{F_0}(m(\eta_\gamma))=
 P_{F_0}(u_1=0)  {\mathbb E}_{F_0} \left[ \log\Phi\left(x_{1 \gamma}^{\top}\alpha_\gamma + s_{1 \gamma}^{\top} \kappa_\gamma - \tau \log(c_1)\right) \mid u_1=0 \right],
  \nonumber \\
+& P_{F_0}(u_1=1) \left( \log(\tau) -\dfrac{1}{2}\log(2\pi)
- \dfrac{1}{2}  {\mathbb E}_{F_0} \left[ \left(\tau \log(o_1)-x_{1 \gamma}^{\top}\alpha_\gamma -s_{1 \gamma}^{\top} \kappa_\gamma \right)^2 \mid u_1=1 \right]
\right)
\label{eq:expected_logl}
\end{align}
its expectation under the data-generating $F_0$.
 Under minimal conditions, $M(\eta_\gamma)$ has a unique maximizer, denoted by $\eta_\gamma^*=(\alpha_\gamma^*,\kappa_\gamma^*,\tau_\gamma^*)$. Below we focus our interpretation on viewing \eqref{eq:expected_logl} as the expectation of a likelihood-associated reward, and $\eta_\gamma^*$ as the associated minimizer, but $\eta_\gamma^*$ can also be viewed as minimizing the Kullback-Leibler divergence to $F_0(y,u)$ (also called generalized Kullback-Leibler divergence, see \cite{hjort:1992}). 

Proposition \ref{prop:Consistency} proves that the maximum likelihood estimator $\widehat{\eta}_\gamma$ converges to $\eta_\gamma^*$,
and Proposition \ref{prop:AsympNorm} its asymptotic normality
with a sandwich covariance that is standard in misspecified models,
and corresponds to the smallest possible covariance for unbiased estimators under model misspecification.
Such variance alteration does not affect consistency but can alter finite $n$ false positives and asymptotic power (see Section \ref{ssec:bfrates}).
See also Propositions \ref{prop:Consistency_lap}-\ref{prop:AsympNorm_lap} for analogous results on the AFT model with Laplace errors.
Mild technical conditions, denoted A1-A5,  that suffice for the proposition to hold are discussed in Section \ref{sec:proof_asymp}. 
\color{black} We remark that A3 assumes the existence and finiteness of $\eta_\gamma^*$ and $\widehat{\eta}_\gamma$ (the latter for large enough $n$), which implies that these optima cannot occur at the boundary of $\Gamma_\gamma$ and must be unique (by concavity). For example, this rules out situations where $\eta_\gamma^*$ contains infinite regression parameters or variance, or zero variance, which we view as pathological cases that we exclude from consideration.
We thank an anonymous referee for pointing out an inconsistency in our original proof, and providing our current argument leading to A3.
\color{black}

\begin{proposition}\label{prop:Consistency}
Assume A1-A3.
Then, $\eta_\gamma^* = \operatorname{argmax}_{\Gamma_\gamma} M(\eta_\gamma)$ is unique and $\widehat{\eta}_\gamma \stackrel{P}{\rightarrow} \eta_\gamma^*$ as $n\rightarrow \infty$.
\end{proposition}

\begin{proposition}\label{prop:AsympNorm}
Assume A1-A5.
Then 
\begin{align}\sqrt{n}(\widehat{\eta}_\gamma-\eta_\gamma^*) \stackrel{D}{\longrightarrow} N\left(0, V_{\eta_\gamma^*}^{-1} {\mathbb E}_{F_0}[ \nabla m(\eta_\gamma^*) \nabla m(\eta_\gamma^*)^{\top}] V_{\eta_\gamma^*}^{-1}\right),
\nonumber
\end{align}
where $V_{\eta_\gamma^*}$ is the Hessian matrix of $M(\eta_{\gamma})$ evaluated at $\eta_\gamma^*$, and $m(\eta_\gamma^*) = \log p(y_1 \mid \eta_\gamma^*)$.
\end{proposition}

Proposition \ref{prop:Consistency} has important implications for model selection.
 Let $(\alpha^*,\kappa^*)$ be the optimal parameter under the full model that includes all linear and non-linear terms. Asymptotically, one obtains the model $\gamma^*$ of smallest dimension maximizing \eqref{eq:expected_logl} (see Section \ref{ssec:bfrates}), which is defined by zeroes in $(\alpha^*,\kappa^*)$. Specifically, $\gamma_j^*=0$ if both linear and non-linear coefficients $(\alpha_j^*,\kappa_j^*)$ are zero, $\gamma_j^*= 1$ if $\alpha_j^* \neq 0$ and $\kappa_j^*=0$, and $\gamma_j^*=2$ if $\kappa_j^* \neq 0$.  

 To interpret this asymptotic solution, we turn attention to \eqref{eq:expected_logl}.  If a covariate does not contribute to improving neither of the two terms in \eqref{eq:expected_logl}, then its corresponding entry in $(\alpha^*,\kappa^*)$ is zero.
 The first term is the expected log-probability, as predicted by the model, that the individual is censored at the observed $\log(c_1)$ (conditional on being censored). Therefore, any covariate that helps the model predict more accurately the occurrence of censoring events contributes to this first term.
The second term is the mean squared error in predicting the observed time $\log(o_1)$, conditional on the time being uncensored. 
Expression \eqref{eq:expected_logl} is an average of these two components weighted by the true censoring probability $P_{F_0}(u_1=0)$, and averaged across covariate values under $F_0$.
Hence $\gamma^*$ drops covariates that do not predict survival neither censoring times, but may include those that, even if truly unrelated to survival, help explain the censoring.
 This interpretation extends to working models other than the Normal AFT. For any other fully-specified model, the first term in \eqref{eq:expected_logl} features the model log-predicted probability of censoring, and the second term the usual log-likelihood for uncensored data. 
For example, under a AFT model with Laplace errors the asymptotic solution is defined by the mean absolute error and the Laplace survival function (see Section \ref{app:AFTLaplace}). 

 We present some simple examples to illustrate our discussion.

\begin{example}
Suppose that under $F_0$, truly $\log o_i \mid c_i \sim N(x_{i1} + \theta \log c_i, \sigma^2)$ and $\log c_i \sim N(x_{i2}, \sigma^2)$. 
The analyst adopts the model $\log o_i \sim N(\beta_1 x_{i1} + \beta_2 x_{i2}, 1/\tau^2)$, which, as discussed, assumes non-informative censoring.
If $\theta=0$, the censoring under $F_0$ is non-informative, 
and then $\alpha_2^*=\beta_2^*=0$, hence $x_{i2}$ is discarded asymptotically.

However, if $\theta \neq 0$ then truly $\log o_i= x_{i1} + \theta x_{i2} + \epsilon_i$, where $\epsilon_i \sim N(0, (1 + \theta^2) \sigma^2)$. 
Plugging this expression into \eqref{eq:expected_logl}, it is easy to show that then $\alpha_2^* \neq 0$. That is, the presence of informative censoring causes $x_{i2}$ to be asymptotically selected.
\end{example}

\begin{example}
Suppose that there is a fixed administrative censoring at $\log c_i=a$ for all individuals (so it is truly non-informative under $F_0$), a single covariate $x_i \in \mathbb{R}$, and that the analyst adopts the model $\log o_i \sim N(\beta_1 + \beta_2 x_i, 1/\tau^2)$.
Suppose that $x_{i}$ truly has an effect on the outcome under $F_0$, but that said effect only occurs at a time $b>a$. 
Then the effect cannot be detected from the observed data, since all individuals are censored at $a$. The issue is that the covariate has an effect that deviates from the assumed AFT structure. For example suppose that, under $F_0$,
$$
\log o_i= z_i + \theta x_{i} \mbox{I}(z_i>b),
$$
where $x_{i} \in \{0,1\}$ indicates that individual $i$ received a treatment, $z_i \sim N(0,1)$ is the survival time for untreated individuals, and $\theta > 0$ quantifies the treatment effect. 

Here the effect is only present among individuals that live longer than $b$ and, since censoring occurs before $b$, for all uncensored individuals one observes $\log o_i= z_i$. Plugging this expression and $\log c_i= a$ into \eqref{eq:expected_logl}, and noting that the conditioning on $u_1$ can be removed from the expectations, one can show that $\alpha_2^*=\beta_2^*=0$.
This is an extreme example where one cannot detect an effect that strongly deviates from the assumed mean structure, even though the censoring is non-informative. One could conceive related examples where a covariate has a time-varying effect that is first positive and then negative, before administrative censoring occurs, so that the average effect is near-zero.
\end{example}

\begin{example}
Suppose that a potentially informative censoring occurs early, so that $P_{F_0}(u_1=0) \approx 1$.
Then \eqref{eq:expected_logl} under the full model is approximately equal to 
$$
{\mathbb E}_{F_0} \left[ \log\Phi\left(x_{1 }^{\top}\alpha + s_{1 }^{\top} \kappa - \tau \log(c_1)\right) \right].
$$
As discussed, this term is the log-probability that the outcome occurs after the observed censoring time, as predicted by the Normal AFT model.
Hence, $(\alpha^*,\kappa^*)$ are essentially chosen to predict censoring times.
If the censoring is informative and depends on a set of covariates, then $(\alpha^*, \kappa^*)$ will in general assign non-zero coefficients to these covariates, which will be asymptotically selected.
A similar argument can be made for late censoring where $P_{F_0}(u_1=1) \approx 1$, then $(\alpha^*,\kappa^*)$ is approximately the usual (population) least-squares solution. If the outcome depends on the censoring, which in turn depends on a set of covariates, then least-squares will assign a non-zero coefficient to the latter.
\end{example}

\subsection{Bayes factor rates for misspecified AFT models}
\label{ssec:bfrates}

This section proves that the posterior probability of the optimal model $\gamma^*$ converges to 1, under mild conditions.
Recall that the posterior probability of $\gamma^*$ is
\begin{align}
 \pi(\gamma^* \mid y)= \frac{p(y \mid \gamma^*) \pi(\gamma^*)}{\sum_\gamma p(y \mid \gamma) \pi(\gamma)}
= \left( 1 + \sum_{\gamma \neq \gamma^*} B_{\gamma, \gamma^*} \frac{\pi(\gamma)}{\pi(\gamma^*)}  \right)^{-1}.
\nonumber
\end{align}

Proposition \ref{prop:BFRatesI} gives the rate at which each $B_{\gamma, \gamma^*}$ converges to 0 (in probability), when ones assumes a potentially misspecified AFT model.
Provided that each $B_{\gamma, \gamma^*} \pi(\gamma)/\pi(\gamma^*)$ converges to 0 (this follows immediately in the standard case where prior model probabilities are bounded, for example) it follows that $\pi(\gamma^* \mid y) \stackrel{P}{\longrightarrow} 1$. This implies that the highest posterior probability model consistently selects $\gamma^*$, and that including covariates with marginal posterior probability $\pi(\gamma_j^* \mid y) > t$, for any fixed threshold $t$, also leads to consistent selection.

Proposition \ref{prop:BFRatesI} clarifies the role of censoring and misspecification.
The result is stated for Laplace approximations to Bayes factors,
a computationally-convenient alternative to obtaining exact marginal likelihoods,
but in our setting both are asymptotically equivalent (Proposition \ref{prop:valid_laplace_concave}).
Specifically, we consider
\begin{align}
B_{\gamma,\gamma^*}= \frac{\widehat{p}(y \mid \gamma)}{\widehat{p}(y \mid \gamma^*)},
\label{eq:bf_laplace}
\end{align}
where $\widehat{p}(y \mid \gamma)$ is obtained via a Laplace approximation:
\begin{equation*}
\widehat{p}(y \mid \gamma)= \exp\{\ell(\tilde{\eta}_\gamma) + \log \pi(\tilde{\eta}_\gamma) \}
(2\pi)^{d_\gamma/2} \left|H(\tilde{\eta}_\gamma) + \nabla^2 \log\pi(\tilde{\eta}_\gamma) \right|^{-1/2},
\label{eq:laplace_approx}
\end{equation*}
where $\tilde{\eta}_\gamma= \arg\max_{\eta_\gamma} \ell(\eta_\gamma) + \log \pi(\eta_\gamma)$ is the maximum a posteriori under prior $\pi(\eta_\gamma)$. See Section \ref{app:withinmodel_comput} for details on computing this approximation.

Proposition \ref{prop:BFRatesI} treats separately overfitted models (containing $\gamma^*$) and non-overfitted models (not containing $\gamma^*$).
Overfitted models contain all truly relevant plus a few spurious parameters, a situation where the challenge is to enforce sparsity.
Non-overfitted models are missing some truly relevant parameters, there the challenge is also to have high power to detect the missing signal.
By truly relevant we mean improving $M(\eta_\gamma^*)$, \textit{i.e.}~the prediction of either observed or censored times, see Section \ref{ssec:consistency}.
Recall that $d_\gamma= \mbox{dim}(\eta_\gamma)= p_\gamma + r s_\gamma + 1$.
Intuitively the proof of Proposition \ref{prop:BFRatesI} is based on establishing the asymptotic distribution of the likelihood-ratio test statistic $2[\ell(\tilde{\eta}_\gamma) - \ell(\tilde{\eta}_{\gamma^*})]$, which is bounded by central chi-squares in the overfitted case and non-central chi-squares in the non-overfitted case, and then finding an asymptotic approximation to the other quantities featuring in $\widehat{p}(y \mid \gamma)$.

\begin{proposition}\label{prop:BFRatesI}
Let $B_{\gamma,\gamma^*}$ be the Bayes factor in \eqref{eq:bf_laplace} under either $\pi_L$, $\pi_M$ or $\pi_E$,
where $\gamma^*$ is the AFT model with smallest $d_{\gamma^*}$ minimizing \eqref{eq:expected_logl}, and $\gamma \neq \gamma^*$ another AFT model.
Assume that both $\gamma^*$ and $\gamma$ satisfy Conditions A1-A5. 
Suppose that $(g_M,g_E,g_L,g_S)$ are non-decreasing in $n$.

\begin{enumerate}[leftmargin=*,label=(\roman*)]

\item Overfitted models. If $\gamma^* \subset \gamma$, then
$$
\log B_{\gamma \gamma^*}= \log (a_n) + \frac{r}{2} (s_{\gamma^*} - s_\gamma) \log \left( n g_S \right) + \lOp(1),
$$

where $a_n= (n g_L)^{\frac{p_{\gamma^*} - p_\gamma}{2}}$ under $\pi_L$,
{$a_n= (n g_M)^{3(p_{\gamma^*} - p_\gamma)/2}$} under $\pi_M$,
and {$a_n=(n g_E e^{2 g_E \sqrt{n}})^{(p_{\gamma^*} - p_\gamma)/2}$} under $\pi_E$.

\item Non-overfitted models. If $\gamma^* \not\subset \gamma$, then
$$
\log(B_{\gamma \gamma^*})= -n[M(\eta^*_{\gamma^*}) - M(\eta^*_\gamma)]
+ \log(b_n) + \frac{r}{2}(s_{\gamma^*} - s_\gamma) \log(n g_S) + \lOp(1)
$$
where $b_n= (n g_L)^{\frac{p_{\gamma^*} - p_\gamma}{2}}$ under $\pi_L$,
{$b_n= (n g_M^3)^{(p_{\gamma^*} - p_\gamma)/2}$} under $\pi_M$,
and {$b_n=(g_En)^{p_{\gamma^*} - p_\gamma} e^{-g_E c}$} under $\pi_E$, for finite $c \in \mathbb{R}$.
\end{enumerate}
\end{proposition}
By Proposition \ref{prop:BFRatesI}(i) the rates to discard overfitted models are unaffected by misspecification and censoring
(but certain constants can affect finite $n$ behaviour, see the proof).
These sparsity rates are improved by non-local priors and by setting large prior dispersions $(g_L,g_M,g_E,g_S)$,
extending previous results \citep{JR12,narisetty:2014,RT17,rossell:2018} to misspecified survival models.
By Proposition \ref{prop:BFRatesI}(ii) the rate to detect non-spurious effects is
exponential in $n$ with a coefficient $M(\eta_{\gamma^*}^*)-M(\eta_\gamma^*) > 0$ that measures the drop of predictive ability
in $\gamma$ relative to $\gamma^*$,  and is hence affected by misspecification and censoring.
 Recall that predictive ability can be understood as a weighted average of forecasting the outcome to occur after the censoring time (for censored individuals) and the actual outcome time (for uncensored individuals).

When one misspecifies the model family, $M(\eta_{\gamma^*}^*)-M(\eta_\gamma^*)$ is driven by the projection of $F_0$ onto the assumed family.
Interpreting the geometry of such projections is beyond our scope,
but intuitively projections usually reduce distances and hence make $M(\eta_{\gamma^*}^*)-M(\eta_\gamma^*)$ smaller
than if one were to assume the correct model class.  By Part (ii), this would decrease the power to detect non-zero effects in $\eta_\gamma^*$. 

 To facilitate interpretation suppose there is no censoring. 
Then simple algebra shows that $M(\eta^*_{\gamma^*}) - M(\eta^*_\gamma)= {\mathbb E}_{F_0}\left[ \log (\tau^*_{\gamma^*} / \tau^*_\gamma ) \right]$,
which measures the difference in mean squared prediction errors from using model $\gamma$ instead of the optimal $\gamma^*$
(given by $1/(\tau_\gamma^*)^2$ and $1/(\tau_{\gamma^*}^*)^2$, respectively).
For instance,  omitting covariates increases $\tau^*_{\gamma^*} / \tau^*_\gamma$, causing an exponential drop in power,
see our examples in Sections \ref{ssec:res_2vars}-\ref{ssec:res_50vars} for an illustration.

Proposition \ref{prop:BFRatesI} also highlights trade-offs in modeling non-linear covariate effects. 
Including a truly active non-linear term is rewarded by an improved model fit
$M(\eta^*_{\gamma^*}) - M(\eta^*_\gamma)$, but runs into an $r \log(n g_S)$ penalty.
In contrast, including a linear effect leads to a smaller improvement in fit, but also incurs a smaller $\log(n g_S)$ penalty.
 Hence, decomposing effects into a linear and non-linear components can improve power. 

 A similar observation illustrates that for model selection purposes, the advantages of using fully non-parametric effects over a finite-dimensional basis may be small. 
Suppose one replaced the basis dimension $r$ by a larger $r^*$ maximizing $M(\eta^*_{\gamma^*}) - M(\eta^*_\gamma)$.
For $m$-degree splines with equi-spaced knots and sufficiently smooth $M()$ the improvement in $M(\eta^*_{\gamma^*}) - M(\eta^*_\gamma)$ associated to increasing $r$ to $r^*$ is at most of order $1/r^m$ \citep{rosen:1971}.
For said increase to offset the complexity penalty it needs to hold that $r^{m+1}(r^*-r)/2$ is of a smaller order than $n/ \log (n g_S)$.
Hence by letting $r^{m+1} r^*$ grow sub-linearly with $n$ could improve power relative to $r$.
However for even moderate $r$ and cubic splines ($m=3$) the required $n> r^* r^4$ can be impractically large,
\textit{e.g.}~see the examples in Section \ref{ssec:res_2vars} with $r \in \{5,10,15\}$.
Further, the computational cost of using a large $r^*$ for each considered model $\gamma$ is impractical when one wishes to consider many models. 

In summary, using a small basis dimension $r$  (\textit{e.g.} $r=5$, in our examples)  within the non-linear effect decomposition in Section \ref{ssec:likelihood}  may be practically preferable to a non-parametric basis where $r$ grows with $n$, for the purpose of detecting the effect.

\subsection{Bayes factor rates for misspecified additive Cox models}
\label{ssec:bfrates_CoxPH}

 Our Bayes factor results under misspecified Cox models are similar to Section \ref{ssec:bfrates}, but here the optimal model $\gamma^*$ is defined by zeroes in the parameter $\eta^*=(\beta^*,\delta^*)$ maximizing the expected partial likelihood \eqref{eq:ploglik} under $F_0$, see \eqref{eq:eploglik_rearrange} for its expression and some discussion. The interpretation of $\eta^*$ is also analogous, though here  \eqref{eq:ploglik} rewards predicting a higher risk for individuals who experienced the event (uncensored) than for other individuals at risk. 
An alternative interpretation is possible by noting that \eqref{eq:ploglik} can be approximated by a Poisson regression log-likelihood \citep{laird:1981}, where one models the mean number of uncensored events in infinitesimal intervals. Intuitively, any covariate that helps predict this mean, which depends on the distribution of the censoring and survival times, is asymptotically selected. Covariates that are unrelated both to survival and censoring are hence discarded.

%

We consider Bayes factors obtained by a Laplace approximation to the integrated partial likelihood
\begin{align}
  p(y\mid \gamma)= \int \exp\left\{\ell_p(\beta_\gamma,\delta_\gamma)\right\} \pi(\beta_\gamma, \delta_\gamma \mid \gamma) d\beta_\gamma d\delta_\gamma
  \label{eq:integrated_ploglik}
  \end{align}
these can be viewed as the integrated likelihood under a limiting non-informative non-parametric Gamma process prior on $h_0$,
see \cite{ibrahim:2014} and \cite{niko:2017} for a discussion. 
We obtain Bayes factor rates analogous to Section \ref{ssec:bfrates},
the proof builds upon \cite{tsiatis:1981} 
and \cite{lin:1989} who proved
that $\bar{\eta}_\gamma=(\bar{\beta}_\gamma,\bar{\delta}_\gamma)$ maximizing \eqref{eq:ploglik} are consistent and asymptotically normal under misspecification, under Conditions B1-B4 listed in Section \ref{app:CoxProportional Hazards}. 

\begin{proposition}\label{prop:BFRatesCox}
  Let $B_{\gamma,\gamma^*}$ be the Bayes factor based on \eqref{eq:integrated_ploglik}
  under $\pi_L$, $\pi_M$ or $\pi_E$, 
  $\gamma^*$ the Cox model with smallest $d_{\gamma^*}$
minimizing the expected log partial likelihood $M_p$ in \eqref{eq:eploglik}, and $\gamma \neq \gamma^*$ another Cox model.
Assume that $(\gamma^*,\gamma)$ satisfy Conditions B1-B4, and that $(g_M,g_E,g_L,g_S)$ are non-decreasing in $n$.
\begin{enumerate}[leftmargin=*,label=(\roman*)]
\item Let $a_n$ be as in Proposition \ref{prop:BFRatesI}. If $\gamma^* \subset \gamma$, then

$$
\log B_{\gamma \gamma^*}= \log (a_n) + \frac{r}{2} (s_{\gamma^*} - s_\gamma) \log \left( n g_S \right) + \lOp(1),
$$

\item Let $b_n$ be as in Proposition \ref{prop:BFRatesI}. If $\gamma^* \not\subset \gamma$, then
$$
\log(B_{\gamma \gamma^*})= -n[M_p(\eta^*_{\gamma^*}) - M_p(\eta^*_\gamma)]
+ \log(b_n) + \frac{r}{2}(s_{\gamma^*} - s_\gamma) \log(n g_S) + \lOp(1).
$$
\end{enumerate}
\end{proposition}

That is, the Bayes factors under an assumed Cox model have similar asymptotic behavior as under an assumed AFT model, hence the conclusions stated in Section \ref{ssec:bfrates} also apply to the Cox model.

\section{Computation}
\label{sec:computation}

The two main computational challenges are exploring the model space $\gamma \in \{0,1,2\}^p$, and approximating the integrated likelihood $p(y \mid \gamma)$ in \eqref{eq:modelposterior} for each model.
 We first discuss relative advantages of the Normal AFT and Cox models for computing $p(y \mid \gamma)$, and how they relate to the amount of censored data in Section \ref{ssec:within_model}. We also discuss an approximation to the Normal log-distribution function derivatives that dramatically speeds up computation for the AFT and probit models.
Section \ref{ssec:model_search} discusses the model search, when one cannot enumerate all $3^p$ models.

\subsection{Within-model calculations}
\label{ssec:within_model}

 When the log-likelihood is concave (or locally concave around $\eta_\gamma^*$, as in asymptotically Normal models), Laplace approximations to $p(y \mid \gamma)$ are one of the fastest and more accurate methods available.
A practical limitation is that, when one wishes to consider many models or the sample size is large, solving the required optimization problems can still be cumbersome. 
This cost can be significantly ameliorated by combining convex optimization algorithms that use warm initializations, see Section \ref{app:withinmodel_comput}.  See also \cite{rossell:2020} for an approach based on approximate Laplace approximations that bypasses the optimization exercise altogether. 

Within survival analysis, an advantage of exponential-family AFT models is admitting sufficient statistics for the uncensored part of the likelihood, \textit{e.g.}~$(y_o^\top y_o, X_o^\top y, X_o^\top X_o)$ for \eqref{logLikeAFT2}.
These can be computed upfront in $n_o(1 + p + p(p+1)/2)$ operations and re-used whenever a new model $\gamma$ is considered at no extra cost, but for large $p$ such pre-computation has significant cost and memory requirements.
Since one typically visits only a small subset of models, many elements in $X_o^\top X_o$ are never used and it would be wasteful to compute them all upfront.
It is more convenient to compute the entries in $X_o^\top X_o$ when first required by any given $\gamma$ and storing them for later use.
Our software follows this strategy by using sparse matrices in the {\tt C++ Armadillo} library \citep{sanderson:2016}.

Given these sufficient statistics the log-likelihood in \eqref{logLikeAFT2} requires
$\min\{ n d_\gamma, (n_c+1) d_\gamma + d_\gamma(d_\gamma+1)/2\}$ operations, and each entry in its gradient and hessian require $n_c+1$ further operations.
In contrast the Cox model's partial likelihood
has a minimum cost of $n_o d_\gamma + n_o(n_o-1)/2$ operations when censored times precede all observed times ($\max c_i < \min o_i$),
and a maximum cost $n d_\gamma + [n(n+1) - n_c(n_c-1)]/2$ when observed times precede all censored times.
That is, the AFT likelihood has a significantly lower cost than the Cox model when $n_c < n_o$ (moderate censoring) or $n > d_\gamma$ (sparse settings).

 A caveat of the Normal AFT model, however, is requiring  the extensive evaluation of the log-cumulative distribution $\log \Phi$ and its derivatives.
Each likelihood evaluation requires $n_c$ terms featuring $\Phi$ and, although these terms can be re-used when computing $r(z)=\phi(z)/\Phi(z)$ and $D(z)=r(-z)^2 - zr(-z)$ in the gradient and hessian, evaluating $\Phi(z)$ is costly.
Briefly, the problem of approximating the inverse Mill's ratio $r(z)$ has been well-studied \cite{gasull:2014}.
There are many algorithms to approximate $\Phi(z)$, but $r(z)$ is harder, \textit{e.g.}~Expression 26.2.16 in \cite{abramowitz:1965} (page 932) has maximum absolute error $<7.5 \times 10^{-8}$ for $\Phi(z)$ but unbounded absolute error for $r(z)$ as $z \rightarrow -\infty$. 
By combining existing proposals we built a fast approximation that guarantees the small relative errors.
One may combine the Taylor series and asymptotic expansions in \cite{abramowitz:1965} (page 932, Expressions 26.2.16 and 26.2.12) for $\Phi(z)$ with an optimized Laplace continued fraction in \cite{lee_chuin:1992} (Expression (5.3)) for $r(z)$ as $z \rightarrow -\infty$. 
The  resulting $\widehat{r}(z)$ has maximum absolute and relative errors $<0.000185$ and $<0.000102$ respectively,
and for $\widehat{D}(z)=\widehat{r}(-z)^2 - z \widehat{r}(-z)$ they are $<0.000424$ and $<0.000505$.
See Section \ref{app:napp} for further details.
As an empirical check, the posterior model probabilities obtained in Section \ref{ssec:ftbrs} when replacing $(r(z),D(z))$ by $(\widehat{r}(z),\widehat{D}(z))$ remained identical to the third decimal place.

This approximation also facilitates evaluating the log-likelihood and derivatives for probit and other models involving $\log \Phi$, and may have some independent interest.
Using this approximation and the warm initializations in Section \ref{app:withinmodel_comput} is practically meaningful, for the TGFB data (Section \ref{ssec:ftbrs}, 868 parameters) they reduced the cost of 1,000 Gibbs iterations from $>$4 hours to 38 seconds.

\subsection{Model exploration}
\label{ssec:model_search}

 Recent advances in Markov Chain Monte Carlo provide model exploration strategies that perform fairly well in practice, see \cite{zanella:2019} for a tempering approach that is particularly helpful when there are multi-modalities in $p(\gamma \mid y)$, or \cite{griffin:2020} for adaptive methods that reduce the effort in exploring low posterior probability models. Further, as $n$ grows and posterior probabilities concentrate on a single model, it is possible to prove quick convergence \citep{yangyun:2016}. Intuitively, if $p(\gamma^* \mid y) \approx 1$ and the chain converges quickly, there is high probability that $\gamma^*$ will be visited after a few iterations. 
Most iterations are spent on models with high $\pi(\gamma \mid y)$ which, from Proposition \ref{prop:BFRatesI}, are models with dimension close to $d_{\gamma^*}$.
The main burden arises from obtaining $p(y \mid \gamma)$, which only needs to be computed the first time that $\gamma$ is visited and can be stored for subsequent iterations.
Hence, if $d_{\gamma^*}$ is not too large (sparse data-generating truths) or $\pi(\gamma \mid y)$ is concentrated on relatively few models, the cost is manageable.

Here for simplicity we describe Algorithm \ref{alg:augmented_gibbs}, a Gibbs algorithm that builds upon earlier proposals \citep{JR12,rossell:2018}, with the novelty that it adds a latent augmentation to enforce hierarchical restrictions (non-linear terms in $S$ are only added if the corresponding linear term in $X$ is in the model) in a computationally-efficient manner.
The algorithm obtains $B$ samples $\gamma^{(1)},\ldots,\gamma^{(B)}$ from $\pi(\gamma \mid y)$.
It is not a naive Gibbs algorithm that sequentially samples $p$ trinary indicators,
\textit{i.e.}~sets $\gamma_j^{(b)}=k$ with probability $\pi(\gamma_j=k \mid y,\gamma_1,\ldots,\gamma_{j-1},\gamma_{j+1},\ldots,\gamma_p)$ for $k \in \{0,1,2\}$.
Instead, it is more convenient to run an augmented-space Gibbs on $2p$ binary indicators.
Specifically let $\tilde{\gamma}_j=\mbox{I}(\gamma_j=1)$ for $j=1,\ldots,p$ denote that covariate $j$ only has a linear effect,
and $\tilde{\gamma}_j=\mbox{I}(\gamma_{j-p}=2)$ for $j=p+1,\ldots,2p$ a non-linear effect.
Algorithm \ref{alg:augmented_gibbs} samples $\tilde{\gamma}_j$ individually
but prevents $(\tilde{\gamma}_j,\tilde{\gamma}_{j+p})=(0,1)$,
\textit{i.e.} enforces that having a non-linear effect when $\beta_j=0$ has zero posterior probability.
The greedy initialization of $\tilde{\gamma}^{(0)}$ is analogous to that in \cite{JR12} and to the heuristic optimization in \cite{polson:2018}.

We remark that Algorithm \ref{alg:augmented_gibbs} may suffer from worse mixing than naive Gibbs sampling of $\gamma_j \in \{0,1,2\}$, but is advantageous in sparse settings.
If covariate $j$ has a small posterior probability $\pi(\gamma_j \neq 0 \mid y)$ then $\pi(\tilde{\gamma}_j=1 \mid y)$ is small
and in most iterations $\tilde{\gamma}_{j+p}$ is set to zero without the need to perform any calculation.
In contrast when sampling $\gamma_j \in \{0,1,2\}$
one must obtain the integrated likelihood for $\gamma_j=2$, which can be costly due to adding the $r$ extra parameters  needed to capture the non-linear effect. 
As an example, in Section \ref{ssec:ftbrs} sampling $\gamma_j \in \{0,1,2\}$ took over 5 times longer to run than
Algorithm \ref{alg:augmented_gibbs}, but provided the same effective sample size up to 2 decimal places.


\begin{algorithm}
\caption{Augmented-space Gibbs sampling}\label{alg:augmented_gibbs}
\begin{algorithmic}[1]
\STATE Set $b=0$, $\tilde{\gamma}^{(0)}=(0,\ldots,0)$.

\STATE For $j=1,\ldots,2p$, update $\tilde{\gamma}_j^{(0)}= \arg\max_k p\left(\tilde{\gamma}_j=k \mid y, \tilde{\gamma}_{-j}^{(b)}\right)$.
If an update was made across $j=1,\ldots,2p$ go back to Step 2,
else set $\gamma_j^{(0)}= \max \left\{\gamma_j^{(0)},\gamma_{j+p}^{(0)} \right\}$ for $j=1,\ldots,p$ and go to Step 3.

\STATE Set $b=b+1$. For $j=1,\ldots,p$ set $\tilde{\gamma}_j^{(b)}=1$ with probability
\begin{align}
P\left(\tilde{\gamma}_j=1 \mid y, \tilde{\gamma}_{-j}^{(b)}\right)=
\begin{cases}
1 \mbox{, if } \tilde{\gamma}_{j+p}=1, \\
\frac{p\left(y \mid \tilde{\gamma}_j=1, \tilde{\gamma}_{-j}^{(b)}\right) p\left(\tilde{\gamma}_j=1, \tilde{\gamma}_{-j}^{(b)}\right)}
{p\left(y \mid \tilde{\gamma}_j=0, \tilde{\gamma}_{-j}^{(b)}\right) p\left(\tilde{\gamma}_j=0, \tilde{\gamma}_{-j}^{(b)}\right) + p\left(y \mid \tilde{\gamma}_j=1, \tilde{\gamma}_{-j}^{(b)}\right) p\left(\tilde{\gamma}_j=1, \tilde{\gamma}_{-j}^{(b)}\right)} \mbox{, if } \tilde{\gamma}_{j+p}=0,
\end{cases}
\nonumber
\end{align}
and otherwise set $\tilde{\gamma}_j^{(b)}=0$.

\STATE For $j=p+1,\ldots,2p$ set $\tilde{\gamma}_j^{(b)}=1$ with probability
\begin{align}
P\left(\tilde{\gamma}_j=1 \mid y, \tilde{\gamma}_{-j}^{(b)}\right)=
\begin{cases}
0 \mbox{, if } \tilde{\gamma}_{j+p}=0, \\
\frac{p\left(y \mid \tilde{\gamma}_j=1, \tilde{\gamma}_{-j}^{(b)}\right) p\left(\tilde{\gamma}_j=1, \tilde{\gamma}_{-j}^{(b)}\right)}
{p\left(y \mid \tilde{\gamma}_j=0, \tilde{\gamma}_{-j}^{(b)}\right) p\left(\tilde{\gamma}_j=0, \tilde{\gamma}_{-j}^{(b)}\right) + p\left(y \mid \tilde{\gamma}_j=1, \tilde{\gamma}_{-j}^{(b)}\right) p\left(\tilde{\gamma}_j=1, \tilde{\gamma}_{-j}^{(b)}\right)} \mbox{, if } \tilde{\gamma}_{j+p}=1,
\end{cases}
\nonumber
\end{align}
and otherwise set $\tilde{\gamma}_j^{(b)}=0$.
If $b=B$ stop, else go back to Step 3.
\end{algorithmic}
\end{algorithm}

\section{Empirical results}
\label{sec:results}

We illustrate via examples the effect of censoring, misspecification and the use of non-linear effect decompositions on model selection.
Section \ref{ssec:res_2vars} considers a simple simulation study with $p=2$ variables, which Section \ref{ssec:res_50vars} extends to $p=50$.
We consider different data-generating truths where the covariates have a monotone or non-monotone effect, and where the truth follows an AFT, proportional hazards, or generalized hazards structure.
In Section \ref{ssec:ftbrs}, we analyze the effect of gene TGFB on colon cancer.
Given that the data-generating truth is unknown, in Section \ref{ssec:falsepos_tgfb} we study the number of false positives via a permutation exercise.
See also Supplementary Section \ref{app:esr}, where we analyze the effect of the estrogen receptor on breast cancer survival.

We consider five   model selection methods   combining the AFT and Cox models with local and non-local priors and with LASSO.
For all Bayesian methods we took the  highest posterior probability model  $\widehat{\gamma}= \arg\max \pi(\gamma \mid y)$ as the selected model.
We refer to the first three methods as AFT-Zellner, AFT-pMOMZ and AFT-LASSO.
They all assume an AFT model and use either  the block-Zellner prior  $\pi_L$,  the non-local pMOM-Zellner prior $\pi_M$  (Section \ref{ssec:elicitation}), or LASSO penalties as proposed by \cite{rahaman:2019}.
AFT-Zellner and AFT-pMOMZ assume  the Normal AFT model  in \eqref{eq:AFTModel}, whereas AFT-LASSO uses a semi-parametric AFT model.
The remaining two methods combine the Cox model with piMOM priors (Cox-piMOM, \cite{niko:2017}) and LASSO (Cox-LASSO, \cite{simon:2011}).
For AFT-Zellner and AFT-pMOMZ we used the function \texttt{modelSelection} in the R package \texttt{mombf}
with the default prior parameters, the Beta-Binomial prior $\pi(\gamma)$ in \eqref{eq:priormodel}
and $B=10,000$ iterations in Algorithm \ref{alg:augmented_gibbs}.
For Cox-piMOM we the used function \texttt{cov\_bvs} in the R package \texttt{BMSNLP} with default parameters and prior dispersion 0.25 as recommended by \cite{niko:2017}.
For AFT-LASSO and Cox-LASSO we used the functions \texttt{AEnet.aft} and \texttt{glmnet} in the R packages  \texttt{AdapEnetClass} and  \texttt{glmnet},
and we set the penalization parameter via 10-fold cross-validation.

\subsection{Censoring, model complexity and misspecification with $p=2$}
\label{ssec:res_2vars}

\begin{figure}
\begin{center}
\begin{tabular}{cc}
\includegraphics[width=0.4\textwidth,height=0.4\textwidth]{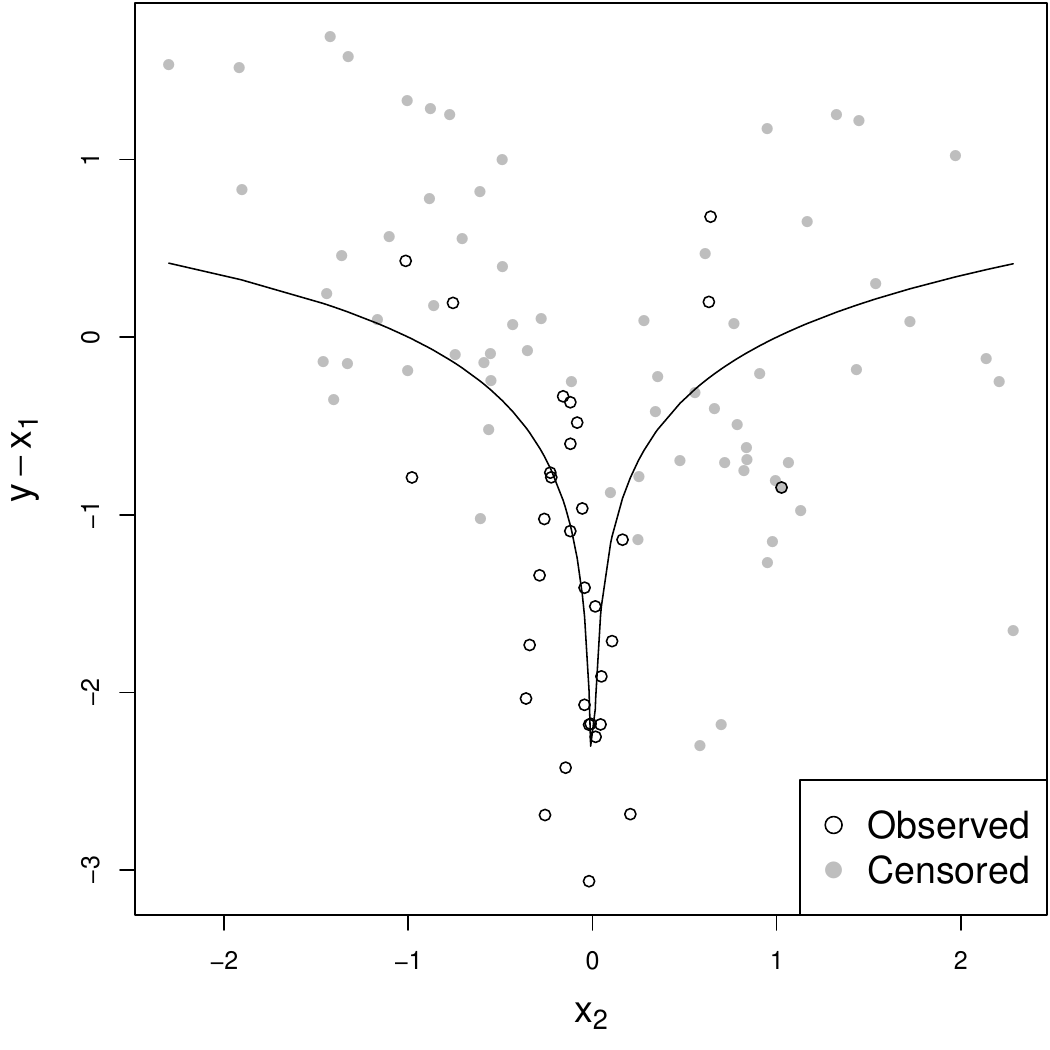} &
\includegraphics[width=0.4\textwidth,height=0.4\textwidth]{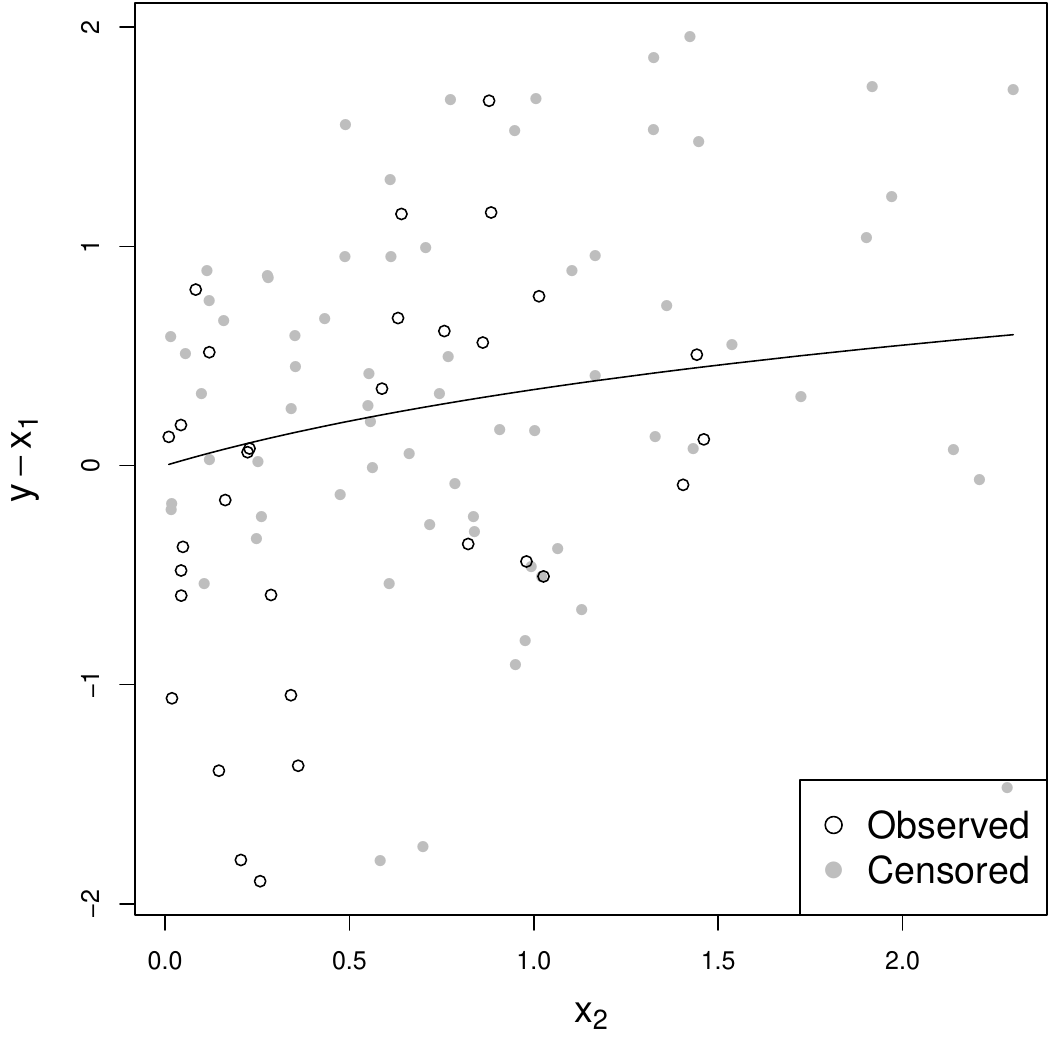} \\
\end{tabular}
\end{center}
\caption{Simulation truth and a simulated dataset for Scenarios 1 (left) and 2 (right).}
\label{fig:2vars}
\end{figure}

We consider sample sizes $n \in \{100, 500 \}$, as well as censored and uncensored data.
We present results for AFT-pMOMZ, as those for AFT-Zellner and Cox-piMOM were largely analogous.
These methods are compared to Cox-LASSO and AFT-LASSO in Section \ref{ssec:res_50vars}.
We consider 6 simulation scenarios. 
Scenarios 1-2 have a data-generating AFT model, Scenarios 3-4 a generalized hazard model and Scenarios 5-6 a proportional hazards model.
 The first covariate has a linear effect in all scenarios, whereas the second covariate has a non-linear effect.  In Scenarios 1, 3 and 5 this effect is strongly non-linear and non-monotone,
whereas in Scenarios 2, 4 and 6 it is monotone and can be roughly approximated by a linear trend, see Figure \ref{fig:2vars}.

\begin{scenario}
AFT structure with $\log o_i= x_{i1} + 0.5 \log(|x_{i2}|) + \epsilon_i$ and $c_i=0.5$,
where $x_i \sim N(0,A)$, $A_{11}=A_{22}=1$, $A_{12}=0.5$, $\epsilon_i \sim N(0,\sigma=0.5)$.
\end{scenario}
\begin{scenario}
AFT structure with $\log o_i= x_{i1} + 0.5 \log(1+x_{i2}) + \epsilon_i$ and $c_i=1$,
where $x_i=(\tilde{x}_{i1},|\tilde{x}_{i2}|)$, $\tilde{x}_i \sim N(0,A)$ and $A$, $\epsilon_i$ as in Scenario 1.
\end{scenario}
\begin{scenario}
Generalized hazards structure with 
\begin{equation*}
h_{GH}(t)= h_0(t \exp \left\{ -x_{i1}/3 + 0.5 \log(|x_{i2}|) \right\})  \exp \left\{ -x_{i1}/3 + 0.75 \log(|x_{i2}|) \right\},
\end{equation*}
$c_i=0.5$,
$h_0$ being the Log-Normal(0,0.5) baseline hazard and $x_i$ as in Scenario 1.
\end{scenario}
\begin{scenario}
 Generalized hazards structure with 
\begin{equation*}
h_{GH}(t)= h_0(t \exp \left\{ -x_{i1}/3 + 0.5 \log(1 +x_{i2}) \right\}) \exp \left\{ -x_{i1}/3 + 0.75 \log(1+x_{i2}) \right\},
\end{equation*}
$c_i=1$, and $h_0$ and $x_i$ as in Scenario 3.
\end{scenario}
\begin{scenario}
Proportional hazards with $h(t)= h_0(t) \exp \left\{ 3x_{i1}/4 - 5\log(|x_{i2}|)/4 \right\}$,
$c_i=0.55$,
$h_0$ being the Log-Normal(0,0.5) baseline hazard and $x_i$ as in Scenario 1.
\end{scenario}
\begin{scenario}
Proportional hazards with $h(t)= h_0(t) \exp \left\{ 3x_{i1}/4 - 5\log(|x_{i2}|) /4\right\}$,
$c_i=0.95$,
and $h_0$ and $x_i$ as in Scenario 5.
\end{scenario}

 In all scenarios, we first consider that there is no censoring, and then a strong administrative censoring,
giving censoring probabilities $P_{F_0}(u_i=0) \approx 0.7$. 

 We first discuss Scenarios 1-2 and illustrate the advantage of using our non-linear effect decomposition. 
We first only considered  the selection of non-linear effects, \textit{i.e.}~$\gamma_j \in \{0,2\}$.
In such case, the power to detect the effects (Figure \ref{fig:2vars_margpp}, top)
was significantly lower than when decomposing them into linear and non-linear parts (Figure \ref{fig:2vars_margpp}, middle).
These findings align with Proposition \ref{prop:BFRatesI}, in the sense that
 the improvement in model fit
needs to overcome the penalty for using a non-linear basis. 
By considering $\gamma_j \in \{0,1,2\}$, one can capture part of the effect with a single linear term.  
Figure \ref{fig:2vars_margpp} also shows that censoring tends to reduce the power for both covariates.

 Second, we illustrate the effect of the non-linear basis dimension $r$.  
We compared the earlier results, where $r$ was part of the model selection, to those obtained under a single fixed $r=$ 5, 10 or 15 (Figure \ref{fig:2vars_margpp}, bottom).
Interestingly, in Scenario 1 the best performance was observed for $r=5$, despite the data-generating truth being strongly non-linear
(Figure \ref{fig:2vars}).
In Scenario 2 the results were highly robust to $r$, as one might expect from the true effect being near-linear.
That is, the smaller $r=5$ gave a good compromise between inference and computation,
we thus used $r=5$ from now on.

The results for Scenarios 3-4 are in Figure \ref{fig:2vars_margpp_sc3}, and for Scenarios 5-6 in Figure \ref{fig:2vars_margpp_sc4}.
The effect of censoring, model complexity and misspecifiying covariate effects were largely analogous to Scenarios 1-2.
To explore further the effects of misspecification, we repeated the simulations in Scenarios 1-2 but now setting $F_0$ to have
asymmetric Laplace errors $\epsilon_i \sim \mbox{ALaplace}(0,s,a)$,
where $a=-0.5$ is the asymmetry and $s$ the scale in the parameterization of \cite{rossell:2018}.
We set $s$ such that the error variance was equal to the Normal simulations, that is $s= \sigma^2 / [2 (1+a^2)]=0.1$.
Figure \ref{fig:2vars_margpp_alapl} shows the results.
These are similar to Figure \ref{fig:2vars_margpp} except for a slight drop in the power to include active covariates.

Finally, we explored the effect of omitting covariates by analyzing the data from Scenarios 1-2
but considering that only $x_{i1}$ was actually observed, \textit{i.e.}~removing $x_{i2}$ from the analysis.
Figure \ref{fig:2vars_margpp_omitX2} shows the results.
Relative to Figure \ref{fig:2vars_margpp}, under Scenario 1 there was a reduction in the posterior evidence for including $x_{i1}$.
Such reduction was not observed in Scenario 2, presumably due to $x_{i1}$ being correlated
with $\log(1+x_{i2})$ and hence picking up part of its predictive power.

\subsection{Censoring, model complexity and misspecification with $p=50$}
\label{ssec:res_50vars}

\begin{figure}
  \begin{center}
    \begin{tabular}{cc}
      \multicolumn{2}{c}{Scenario 1} \\
      \includegraphics[width=0.4\textwidth]{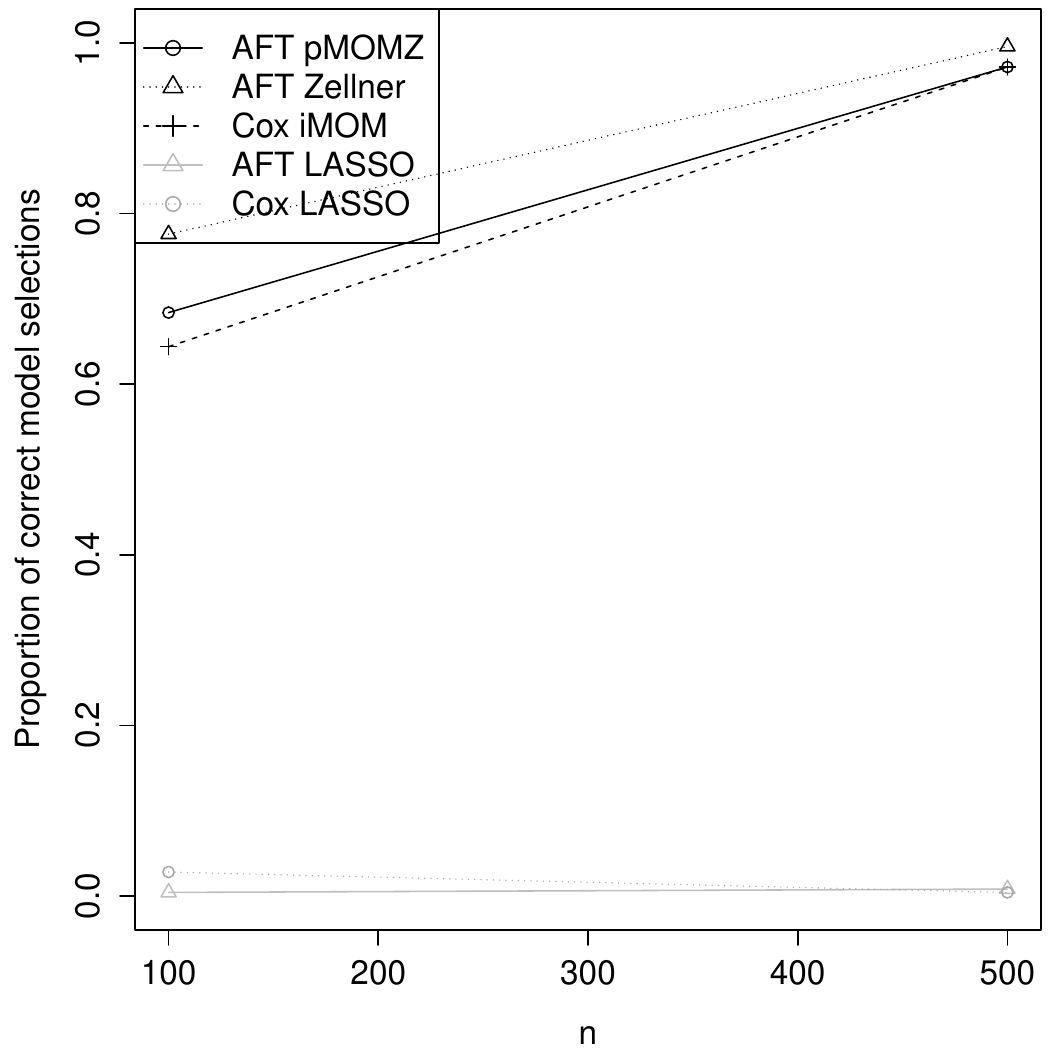} &
      \includegraphics[width=0.4\textwidth]{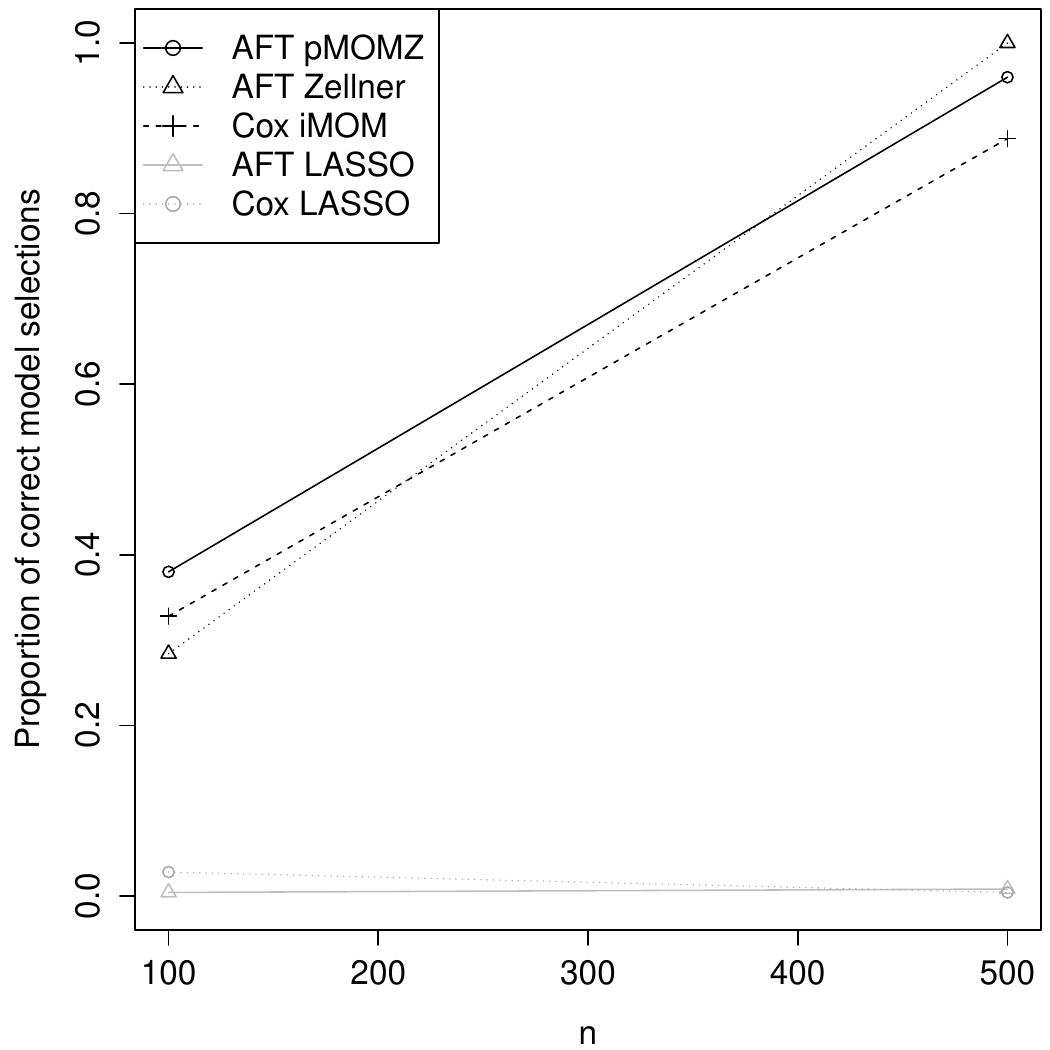} \\
      \multicolumn{2}{c}{Scenario 2} \\
      \includegraphics[width=0.4\textwidth]{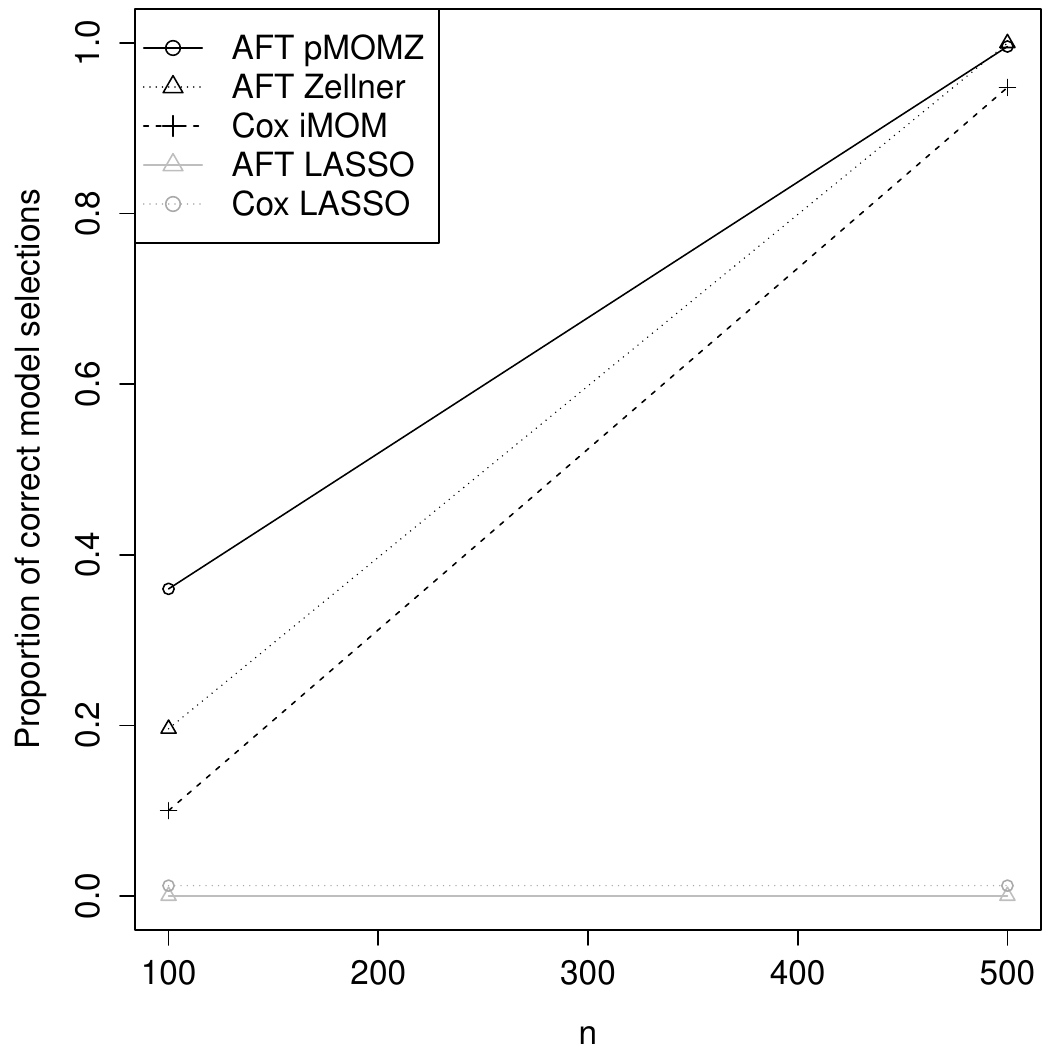} &
      \includegraphics[width=0.4\textwidth]{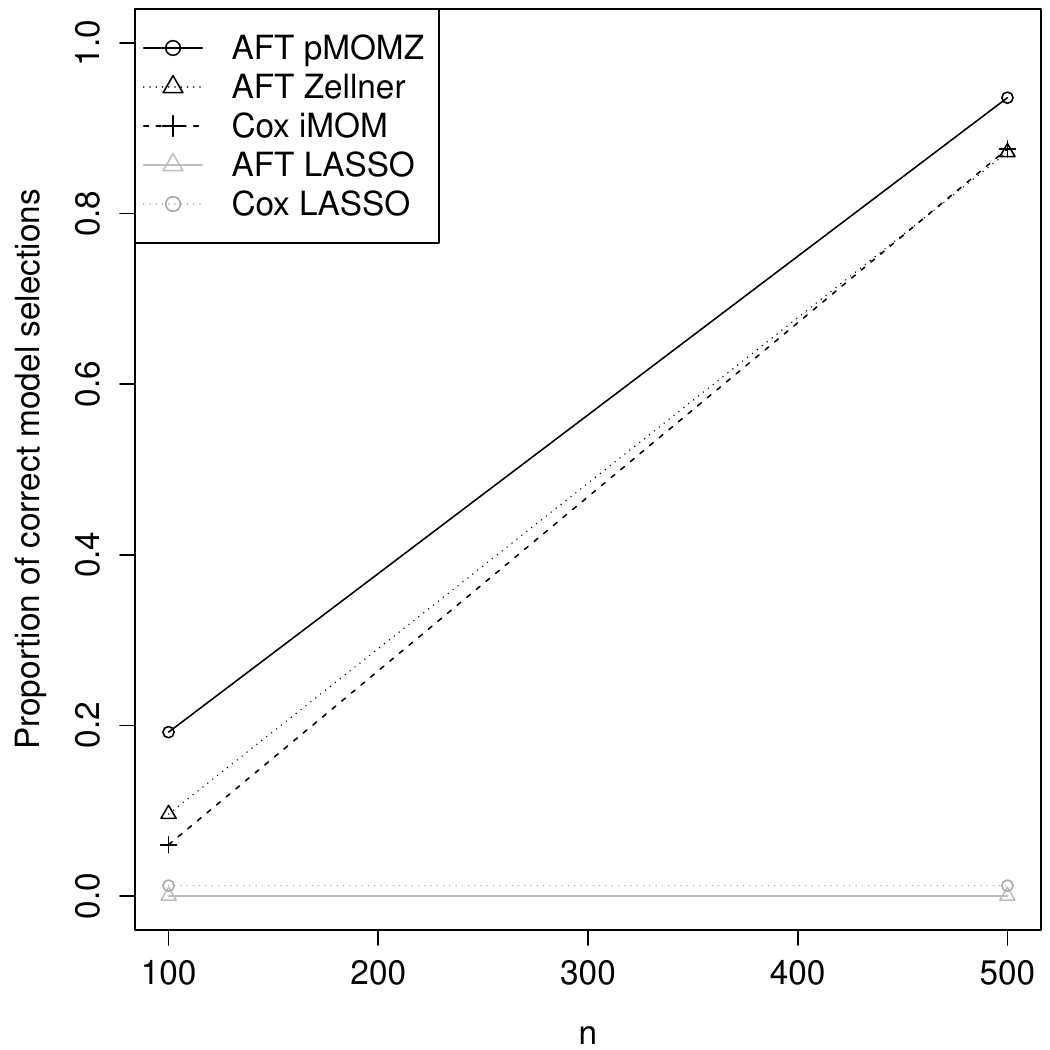}
    \end{tabular}
\end{center}
\caption{Scenarios 1-2, $p=50$. Correct model selection proportion in uncensored (left) and censored (right) data.}
\label{fig:pcorrect_scen12}
\end{figure}

We extended Scenarios 1-6 from Section \ref{ssec:res_2vars} by adding 48 spurious covariates.
We generated  covariates  $x_i \sim N(0,A)$ where $A$ is a $50 \times 50$ matrix with unit diagonal and all off-diagonal $A_{ij}=0.5$, and otherwise simulated data as in Section \ref{ssec:res_2vars}.
Figure \ref{fig:pcorrect_scen12} shows the proportion of correct model selections by each model selection  method
in Scenarios 1-2, across 250 independent simulations.
Figure \ref{fig:pcorrect_scen3} reports these results for Scenarios 3-4, and Figure \ref{fig:pcorrect_scen4} for Scenarios 5-6.
Tables \ref{tab:pcorrect_scen1}-\ref{tab:pcorrect_scen6} also display the posterior probability assigned to  the optimal model  $\gamma^*$
and the average number of truly active and truly inactive selected covariates.
All Bayesian methods exhibited a good ability to select $\gamma^*$ that improved with larger $n$ and uncensored data
(as predicted by Proposition \ref{prop:BFRatesI}),
and they all provided significant improvements over Cox-LASSO and AFT-LASSO,  particularly in reducing the number of false positives. 
As expected AFT-Zellner and AFT-pMOM tended to slightly outperform Cox-piMOM under truly AFT data (Scenarios 1-2),
and conversely under truly proportional hazards data (Scenarios 5-6), though the differences were relatively minor.
Interestingly, under the generalized hazards model (Scenarios 3-4) again AFT-Zellner and AFT-pMOMZ achieved higher correct selection rates,
presumably due to these generalized hazard settings being closer to an AFT than to an proportional hazards model.

\subsection{Effect of TGFB and fibroblasts in colon cancer metastasis}
\label{ssec:ftbrs}

\cite{calon:2012} studied the effect of 172 genes related to fibroblasts (f-TBRS signature), a cell type producing the structural framework in animals, and a growth factor (TGFB) associated with lower colon cancer survival (time until recurrence).
The authors obtained 172 genes responsive to TGFB in mice fibroblasts.
They then used independent gene expression data from human patients,  with tumor stages 1-3,  to show that an overall high mean expression of these 172 genes was strongly associated with metastasis.
We analyzed their data to provide a more detailed description of the role of TGFB and f-TBRS on survival.
 We used the $n=260$ patients with available survival times, and used tumor stage (2 dummy indicators), TGFB and the 172 f-TBRS genes as covariates, for a total of $p=175$.  
We first performed model selection via AFT-pMOMZ only for staging and TGFB.
The top model had 0.976 posterior probability and included stage and a linear effect of TGFB, confirming that TGFB is associated with metastasis.
The posterior marginal inclusion probability for a non-linear effect of TGFB was only 0.009.
As an additional check, the maximum likelihood estimator under the top model gave P-values$<0.001$ for stage and the linear TGFB effect.
The estimated time accelerations associated to TGFB are substantial (Figure \ref{fig:tgfb}, left). 

Next, we extended the exercise to all 175 variables, only considering linear effects.
The top model contained gene FLT1 and the second top model genes ESM1 and GAS1, with respective posterior model probabilities 0.088 and 0.081.
These were also the genes with highest inclusion probabilities (0.208, 0.699 and 0.567 respectively).
There is plausible biology connecting FLT1, ESM1 and GAS1 to metastasis.
From {genecards.org} \citep{stelzer:2016},
FLT1 is a growth and permeability factor in cell proliferation and cancer invasion.
ESM1 is related to endothelium disorders, growth factor receptor binding and gastric cancer networks,
and GAS1 plays a role in growth and tumor suppression.
Interestingly the marginal inclusion probability for TGFB was only 0.107,
that is after accounting for the top 3 genes TGFB did not show a significant effect on survival.
For confirmation, we fitted via maximum likelihood the model with FLT1, ESM1, GAS1, stage and TGFB.
The P-value for TGFB was 0.281 and its estimated effect was substantially reduced (Figure \ref{fig:tgfb}, right).
Finally, we considered both linear and non-linear effects ($p(1+r)=1050$ columns in $(X,S)$).
All non-linear effects had inclusion probabilities below 0.5
and the top 2 models contained FLT1, ESM1 and GAS1, as before.
For comparison we run Cox-piMOM, AFT-LASSO and Cox-LASSO on the linear effects $(p=175)$.
Stage and FLT1 were again selected by the top model under Cox-piMOM and by Cox-LASSO.
Cox-LASSO selected 9 other genes, but only 4 had a significant P-value upon fitting a Cox model via maximum likelihood.
Finally AFT-LASSO selected stage and 6 genes, two of which were also selected by Cox-LASSO. 
See Section \ref{app:esr} for a similar analysis of the estrogen receptor ESR1 effect on breast cancer.

 Since this is a real-data application with an unknown ground truth, it is hard to assess which method performed best. 
As a first check, Table \ref{tab:ci_tgfb} reports the estimated predictive accuracy of each method
via the leave-one-out cross-validated concordance index \citep{harrell:1996}.
Cox-LASSO and AFT-pMOMZ achieved the highest concordance indexes,
with the former selecting more variables than the latter on average across the cross-validation
(13.6 vs. 3.9 for $p=175$ and 11.7 vs. 4.9 for $p(1+r)=1050$).
 We remark that predictive accuracy is not our primary goal, but if a method were to miss truly active covariates then one would expect accuracy to decrease, hence it serves as a rough proxy for statistical power. To complete the exercise, we next evaluate false positive probabilities. 




\subsection{False positive assessment under colon cancer data}
\label{ssec:falsepos_tgfb}

We did a permutation exercise to assess false positive findings in the colon cancer data.
We randomly permuted the recurrence times, and left the covariates unpermuted.
We obtained 100 independent permutations and recorded the model selected by each method. 
We first included only stage, a linear and non-linear term for TGFB as covariates, for a total of $p(r+1)=8$ columns in $(X,S)$.
Next, we repeated the exercise considering linear effects for staging and the 173 genes, for a total of $p=175$ columns.

The results are in Table \ref{tab:falsepos_tgfb} and Figure \ref{fig:falsepos_tgfb}.
AFT-pMOMZ achieved an excellent false positive control, it selected the null model in all permutations
and assigned an average  posterior probability  $\pi(\gamma= 0 \mid y)=0.846$ and $0.844$ to the null model in the exercises with 8 and 175 columns (respectively).
That is, AFT-pMOMZ not only selected the null model but also assigned a high confidence to that selection.
All competing methods selected the null model significantly less frequently.
They also showed inflated false positive percentages for the analysis with 8 columns,
though interestingly these percentages were lower in the analysis with 175 columns.
Figure \ref{fig:falsepos_tgfb} reveals an interesting pattern for Cox-piMOM, in $>97$\% of the permutations only 1 covariate was included.
That is, although the mean false positives percentage for Cox-piMOM was similar to AFT-LASSO and Cox-LASSO,
the selected model was always very close to the null model,
as expected from the strong sparsity-inducing properties of non-local priors.

\begin{table}
\begin{center}
\begin{tabular}{l|rr|rr} \hline
 & \multicolumn{2}{c|}{Stage + TGFB $(p(r+1)= 8)$} & \multicolumn{2}{c}{Stage + all genes $(p= 175)$} \\
 & False positives & $\widehat{\gamma}=0$ & False positives & $\widehat{\gamma}=0$ \\
  \hline
  AFT-pMOMZ & 0.0 & 100.0 & 0.0 & 100.0 \\
 Cox-piMOM  & 12.1 & 3.0 & 0.6 & 1.0 \\
  AFT-LASSO & 35.9 & 31.0 & 2.2 & 45.0 \\
  Cox-LASSO & 12.6 & 68.0 & 1.5 & 61.0 \\
   \hline
\end{tabular}
\end{center}
\caption{Percentage of false positives and correct model selections $(\widehat{\gamma}=0)$ in permuted colon cancer data (100 permutations)
when the design had 8 columns (stage, linear and non-linear effect of TGFB) and 175 columns (stage and linear effect of 173 genes).}
\label{tab:falsepos_tgfb}
\end{table}

\section{Discussion}
\label{sec:discussion}

Our main contributions are describing a generic Bayesian model selection framework
to incorporate non-linear effects in a data-driven fashion to balance power and sparsity
and, perhaps more importantly, helping understand the interplay between censoring, misspecification and model complexity.
In survival models, we showed that one asymptotically discards covariates that do not help predict the outcome neither censoring times (conditionally on other covariates),
whereas in probit regression one keeps those that help reduce the probit loss function,  and similarly for other concave log-likelihoods. 
We showed that censoring and misspecification can reduce power significantly.
 Understanding this phenomenon can be useful in the design of experiments, where one may increase the follow-up length to gain power.  
 Enriching the model class, by considering semi- and non-parametric terms, to alleviate model misspecification requires some care as these additional terms can incur computational and statistical power losses. 
Our recommendation is to use Bayesian model selection to decide their inclusion in a data-adaptive manner, as in the proposed linear plus deviation from linearity decomposition.
Although not discussed here for simplicity, one can also easily incorporate interactions between covariates into the proposed theory and computational methods.

From a technical point of view we used standard asymptotic arguments which, for concave log-likelihoods, lead to simpler proofs and technical conditions. 
It should be possible to extend our results, with some care, to  non-concave and non-asymptotic settings (for example, using the high-dimensional framework in  \citep{panov:2015}),  interval and left censored data, as well as to cure rate, recurrence or excess hazards models.
We focused on fixed $p$ to provide simpler results and intuition, under less restrictive technical conditions. While, in theory, it can be potentially interesting to allow the non-linear basis dimension $r$ to grow with $n$, for actual methodology this often implies an impractical computational cost. This is critical in structural learning, where one wishes to consider many models. For this reason, in applied settings, it is common to use a finite basis.

Regarding high-dimensional settings, from recent results on misspecified penalized non-concave likelihood \citep{loh:2017} Bayesian model selection \citep{yangyun:2017,R18}, we speculate that our main findings should remain valid.
We remark, however, that high-dimensional formulations often incorporate stronger sparsity via the prior distribution, hence the power drop caused by censoring and misspecification could be more problematic than in our fixed $p$ case.

We focused on model selection within additive models, but our results extend directly when one wishes to consider interactions, by adding the corresponding basis to our formulation. 
Our theory is valid for any given basis and also when performing selection on the basis itself, however, admittedly our examples focused on spline basis with fixed knots.
We feel that a detailed study of basis selection would obscure the high-level intuition of our main results, but it represents an interesting aspect for future research.

\section*{Acknowledgments}

David Rossell was supported by Spanish Government grants
RyC-2015-18544, Plan Estatal PGC2018-101643-B-I00, Europa Excelencia EUR2020-112096,
Ayudas Fundaci\'on BBVA a Investigaci\'on en Big Data 2017, and NIH grant R01 CA158113-01.

\clearpage

\section*{Supplementary Material: Additive Bayesian variable selection under censoring and misspecification}

\section{Likelihood factorization under non-informative censoring}
\label{suppsec:noninformati_censoring}

Most popular survival models assume that the censoring times are non-informative, that is that the survival and censoring times $(o_i,c_i)$ are conditionally independent given covariates $x_i$.
Under such an assumption, the contribution of the censoring distribution to the likelihood function factors out, and can hence be ignored for making inference on survival times.
To see this, the likelihood associated to an arbitrary joint model $p_{o,c}(o_i,c_i \mid x_i)$ is
\begin{align}
\prod_{u_i=1} p_{o \mid c}(o_i \mid c_i \geq o_i, x_i) P_c(c_i \geq o_i \mid x_i) \prod_{u_i=0} p_{c \mid o}(c_i \mid  c_i < o_i, x_i) P_o(o_i > c_i \mid x_i)
\nonumber
\end{align}
which, if $(o_i,c_i)$ are assumed conditionally independent given $x_i$, simplifies to
\begin{align}
\prod_{i=1}^n p_o(o_i \mid x_i)^{u_i}  P_o(o_i > c_i \mid x_i)^{1-u_i}
\prod_{i=1}^n P_c(c_i \geq o_i \mid x_i)^{u_i} p_c(c_i \mid  x_i)^{1-u_i}.
\nonumber
\end{align}
The second term can be disregarded, for purposes of making inference on the marginal distribution of the survival times.

\section{Log-likelihood and priors}

\subsection{Marginal prior on $\beta_j$}
\label{app:marginalprior}


Straightforward algebra shows that the marginal MOM prior density on $\beta_j$ associated to $ \pi_M$ is
\begin{align}
  \pi_M(\beta_j)= \int \pi_M(\beta_j | \sigma^2) \mbox{IG}(\sigma^2; a_\tau/2,b_\tau/2)=
  \frac{2 \Gamma \left( (a_\tau+3)/2 \right)}{\Gamma(a_\tau/2) \pi^{1/2} (b_\tau g_M)^{3/2}}
  \frac{\beta_j^2}{(1+ \beta_j^2/(g_M b_\tau))^{(a_\tau+3)/2}}.
\nonumber
\end{align}
Similarly, the marginal eMOM prior density associated to $ \pi_E$ is given by
\begin{align}
  \pi_E(\beta_j)=
  \frac{e^{\sqrt{2}} \Gamma \left( (a_\tau+1)/2 \right)}{\Gamma(a_\tau/2) (\pi g_E b_\tau)^{1/2}}
  \frac{M \left( -g_E/\beta_j^2, (a+1)/2, (b+\beta_j^2/g_M)/2 \right)}{(1+ \beta_j^2/(g_E b_\tau))^{(a_\tau+1)/2}},
\nonumber
\end{align}
where $M(t,c,d)$ is the moment generating function of an $\mbox{IG}(a,b)$ distribution evaluated at $t$.
Let $t(\beta_j; \nu)$ be the density of a t distribution with $\nu$ degrees of freedom.
If $a_\tau=b_\tau$ then $\lim_{\beta_j \to \infty} \pi_M(\beta_j)/t(\beta_j;a_\tau)=C$ where $C$ is a constant not depending on $\beta_j$,
\textit{i.e.}~$\pi_M(\beta_j)$ has the same tail thickness as $t(\beta_j;a_\tau)$.
Similarly for $\pi_E(\beta_j)$ note that $\exp\{-g_E/\beta_j^2\}\leq 1$, then the Dominated Converge Theorem implies that $\lim_{\beta_j \to \infty} M()= 1$
and hence $\lim_{\beta_j \to \infty} \pi_M(\beta_j)/t(\beta_j;a_\tau)=D$ where $D$ is a constant.

To evaluate the cumulative distribution functions associated to $\pi_M(\beta_j)$ and $\pi_E(\beta_J)$
we used the quadrature-based numerical integration implemented in functions \texttt{pmomigmarg} and \texttt{pemomigmarg} from R package \texttt{mombf}.

\begin{figure}[h!][h!]
\begin{center}
\begin{tabular}{cc}
  \includegraphics[width=0.35\textwidth]{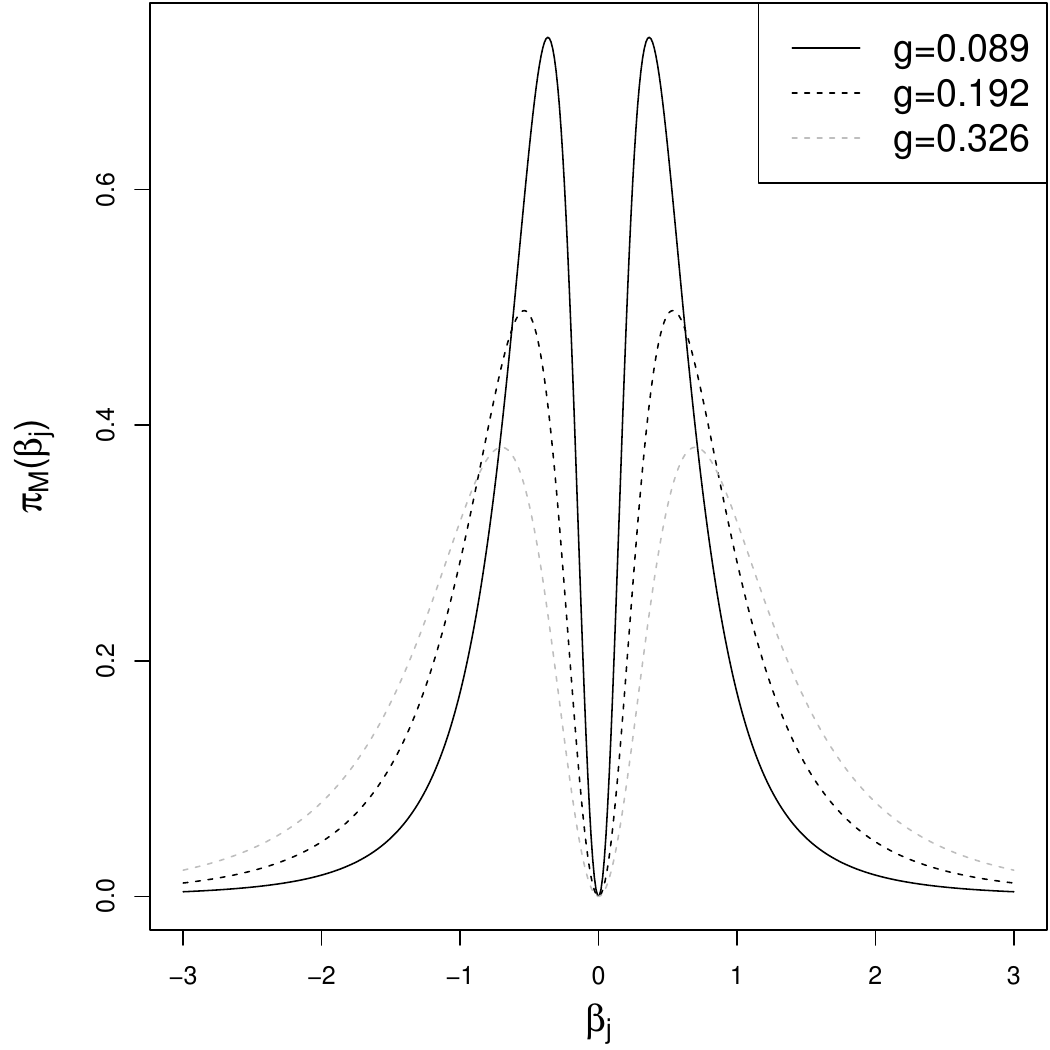}
  \includegraphics[width=0.35\textwidth]{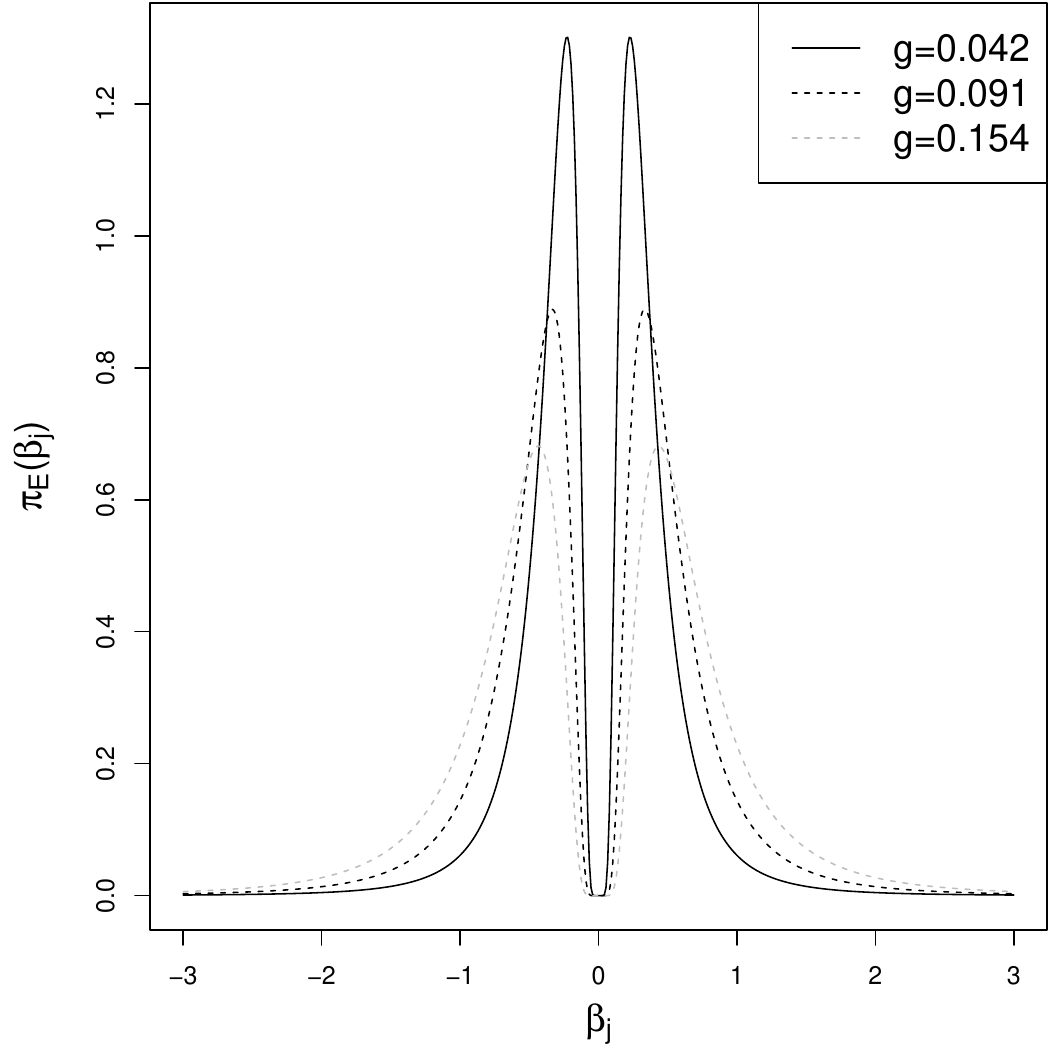}
\end{tabular}
\end{center}
\caption{Marginal pMOMZ prior $\pi_M(\beta_j)$ for $g_M = 0.089,0.192,0.326$ (left)
  and peMOMZ prior $\pi_E(\beta_j)$ for $g_E=0.042,0.091,0.154$.
  For both $a_\tau=b_\tau=3$.}
\label{fig:MOMpriors}
\end{figure}
\subsection{Gradient and Hessian of the priors on $(\alpha,\kappa,\tau)$}
\label{eq:prior_reparam}

The gradient of the logarithm of $\pi_1(\alpha_{\gamma},\kappa_\gamma,\tau)$ is
\begin{eqnarray*}
\nabla_\alpha \log \pi_1(\alpha_{\gamma},\kappa_\gamma,\tau) &=& \dfrac{2}{\alpha_{\gamma}} -\dfrac{\alpha_{\gamma}}{ g_M},\\
\nabla_{\kappa_j} \log \pi_1(\alpha_{\gamma},\kappa_\gamma,\tau) &=&  - \frac{1}{g_S n} S_j^\top S_j \kappa_\gamma,\\
  \nabla_\tau \log \pi_1(\alpha_{\gamma},\kappa_\gamma,\tau) &=&
                                                                 \frac{a_\tau - 1}{\tau} -  b_\tau \tau.
\end{eqnarray*}
where $1/\alpha_{\gamma}$ denotes inverse of each entry of this vector.
Regarding $\pi_2(\alpha_{\gamma},\kappa_\gamma,\tau)$ we obtain
\begin{eqnarray*}
\nabla_\alpha \log \pi_2(\alpha_{\gamma},\kappa_\gamma,\tau) &=& \dfrac{2 g_E}{\alpha_{\gamma}^3} -\dfrac{\alpha_{\gamma}}{ g_E},\\
\nabla_{\kappa_j} \log \pi_2(\alpha_{\gamma},\kappa_\gamma,\tau) &=&  - \frac{1}{g_S n} S_j^\top S_j \kappa_\gamma,\\
  \nabla_\tau \log \pi_2(\alpha_{\gamma},\kappa_\gamma,\tau) &=& \frac{a_\tau - 1}{\tau} -  b_\tau \tau.
\end{eqnarray*}

The Hessian of $\log \pi_1$ is
\begin{eqnarray*}
\nabla^2_\alpha \log \pi_1(\alpha_{\gamma},\kappa_\gamma,\tau) &=& \operatorname{diag}\left\{  -\dfrac{2}{\alpha_{\gamma}^2} -\dfrac{1}{ g_M}\right\},\\
\nabla^2_{\kappa_j} \log \pi_1(\alpha_{\gamma},\kappa_\gamma,\tau) &=&  - \frac{1}{g_S n} S_j^\top S_j,\\
\nabla^2_{\tau} \log \pi_1(\alpha_{\gamma},\tau) &=& -\dfrac{a_\tau -1}{\tau^2} - b_\tau,
\end{eqnarray*}

The Hessian of $\log \pi_2$ is:
\begin{eqnarray*}
\nabla^2_\alpha \log \pi_2(\alpha_{\gamma},\kappa_\gamma,\tau) &=& \operatorname{diag}\left\{ -\dfrac{6 g_E}{\alpha_{\gamma}^4} -\dfrac{1}{ g_E} \right\},\\
\nabla^2_{\kappa_j} \log \pi_2(\alpha_{\gamma},\kappa_\gamma,\tau) &=&  - \frac{1}{g_S n} S_j^\top S_j,\\
\nabla^2_\tau \log \pi_2(\alpha_{\gamma},\kappa_\gamma,\tau) &=& -\dfrac{a_\tau -1}{\tau^2} - b_\tau,
\end{eqnarray*}

Further all other elements in the hessian are zero. That is, for $l=1,2$ we have
$\nabla_\kappa \nabla_{\alpha} \log \pi_l(\alpha_{\gamma},\kappa,\tau)=
\nabla_\tau \nabla_{\alpha} \log \pi_l(\alpha_{\gamma},\kappa,\tau)= \nabla_\tau \nabla_{\kappa} \log \pi_l(\alpha_{\gamma},\kappa_\gamma,\tau)= 0.$

\subsection{Gradient and Hessian of the  log-likelihood and priors}
\label{app:prior_reparam_logtau}

The gradient of \eqref{logLikeAFT2} is $g(\alpha,\kappa,\tau)=$
\begin{equation}
\begin{pmatrix}
\nabla_{(\alpha,\kappa)} \ell(\alpha,\kappa,\tau)  \\
\nabla_\tau \ell(\alpha,\kappa,\tau)
\end{pmatrix} =
\begin{pmatrix}
\sum_{u_i=1} \begin{pmatrix} x_i \\ s_i \end{pmatrix} \left(\tau y_i-x_i^{\top}\alpha -s_i^{\top}\kappa \right) + \sum_{u_i=0} \begin{pmatrix} x_i \\ s_i \end{pmatrix} r\left(x_i^{\top}\alpha + s_i^{\top}\kappa - \tau y_i \right)  \\
\dfrac{n_o}{\tau} -  \sum_{u_i=1} y_i\left(\tau y_i-x_i^{\top}\alpha-s_i^{\top}\kappa \right)  - \sum_{u_i=0} y_i r\left(x_i^{\top}\alpha + s_i^{\top} \kappa - \tau y_i \right)
\end{pmatrix},
\label{eq:grad_logLikeAFT2}
\end{equation}

Let $r(t) = \phi(t)/\Phi(t)$ be the Normal inverse Mills ratio, the Hessian of \eqref{logLikeAFT2} is $H(\alpha,\kappa,\tau)=$
\begin{eqnarray*}
  \nabla^2_{(\alpha,\kappa)} \ell(\alpha,\kappa,\tau) &=&
 - \sum_{u_i=1} \begin{pmatrix} x_i \\ s_i \end{pmatrix} \begin{pmatrix} x_i \\ s_i \end{pmatrix}^{\top}
- \sum_{u_i=0} \begin{pmatrix} x_i \\ s_i \end{pmatrix} \begin{pmatrix} x_i \\ s_i \end{pmatrix}^{\top} D\left(\tau y_i - x_i^{\top}\alpha - s_i^{\top} \kappa \right),\\
  \nabla_{\tau}\nabla_{(\alpha,\kappa)} \ell(\alpha,\kappa,\tau) &=& \sum_{u_i=1} \begin{pmatrix} x_i \\ s_i \end{pmatrix} y_i + \sum_{u_i=0} \begin{pmatrix} x_i \\ s_i \end{pmatrix} y_i D\left(\tau y_i - x_i^{\top}\alpha - s_i^{\top} \kappa \right),\\
  \nabla^2_{\tau}\ell(\alpha,\kappa,\tau)  &=&  -\dfrac{n_o}{\tau^2} - \sum_{u_i=1} y_i^2  - \sum_{u_i=0} y_i^2 D\left(\tau y_i - x_i^{\top}\alpha - s_i^{\top} \kappa \right).
\end{eqnarray*}
where $D\left(z \right) = r\left(-z \right)^2 -z r\left(-z \right) \in (0,1)$ can be interpreted as the proportion of information provided by an observation that was censored $z$ standard deviations after the mean survival.
In fact, as $D\left(z \right)$ is an increasing function of $z$, this implies that increasing $z$, {\it e.g.}~by increasing follow-up, increases the Hessian's curvature and therefore inferential precision.
This gain is largest when $z \in (-2,1)$ and gradually plateaus afterwards.
Observations with small $z$ provide essentially no information and, since the Bayes factor rate to detect signal is exponential in the number of complete observations \citep{dawid:1999}, censoring causes an exponential drop in power. This intuition is made precise in Section \ref{ssec:bfrates}.

\begin{figure}[h!]
\begin{center}
\begin{tabular}{c}
\includegraphics[width=0.5\textwidth]{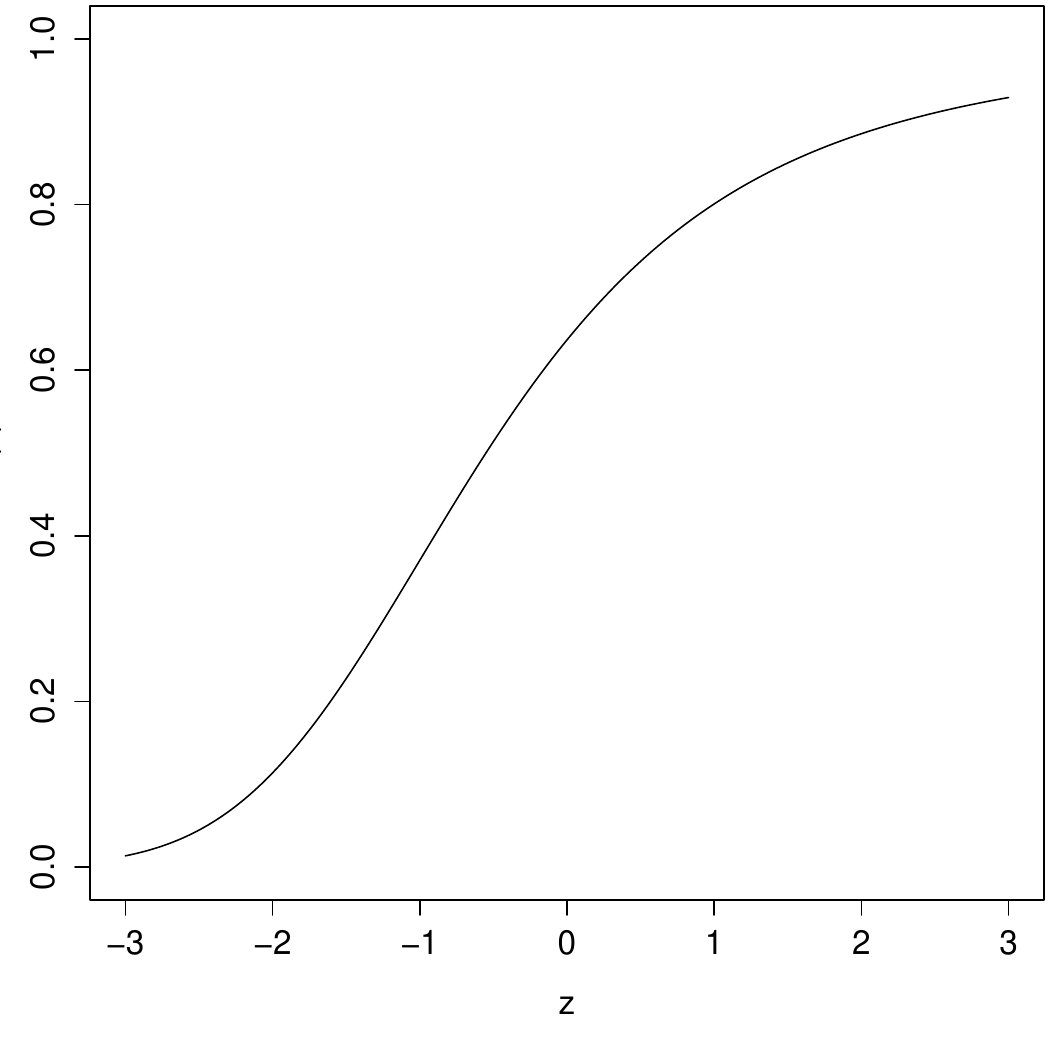}
\end{tabular}
\end{center}
\caption{Proportion of information $D(z)$ when censoring $z$ standard deviations after the mean survival.}
\label{fig:DZ}
\end{figure}

When computing Laplace approximations it is recommendable that the parameter space is unbounded,
to this end we re-parameterize $\rho= \log(\tau)$.
The log-likelihood, its gradient and hessian with respect to $(\alpha,\kappa)$ are as given in \eqref{logLikeAFT2} simply replacing $\tau$ by $e^\rho$.
Its first and second derivatives with respect to $\rho$ are
\begin{eqnarray*}
  \nabla_\rho \ell(\alpha, \kappa, \rho)&=& n_0 - e^\rho \sum_{u_i=1} y_i (e^\rho y_i - x_i^\top \alpha - s_i^\top \kappa) - e^\rho \sum_{u_i=0} y_i r(x_i^\top \alpha + s_i^\top \kappa - e^\rho y_i) \\
      \nabla_\rho \nabla_{(\alpha,\kappa)} \ell(\alpha,\kappa,\rho) &=& e^\rho \sum_{u_i=1} \begin{pmatrix} x_i \\ s_i \end{pmatrix} y_i + e^\rho \sum_{u_i=0} \begin{pmatrix} x_i \\ s_i \end{pmatrix} y_i D\left(e^\rho y_i - x_i^{\top}\alpha - s_i^{\top} \kappa \right) \\
  \nabla^2_\rho \ell(\alpha, \kappa, \rho)&=& \nabla_\rho \ell(\alpha, \kappa, \rho) - n_0
  - e^{2\rho} \left[ \sum_{u_i=1} y_i^2 + \sum_{u_i=0} y_i^2 D(e^\rho y_i - x_i^\top \alpha - s_i^\top \kappa )  \right]
\end{eqnarray*}

The priors on $(\alpha,\kappa,\rho)$ implied by $\pi_M$ and $\pi_E$ are

\begin{eqnarray*}
\tilde{\pi}_1(\alpha_{\gamma},\kappa_\gamma,\rho) &=& \prod_{\gamma_j\geq 1} \frac{\alpha_j^2}{g_M} N(\alpha_j; 0, g_M)
\prod_{\gamma_j=2} N(\kappa_j; 0, g_S n (S_j^\top S_j)^{-1}) \\
&\times &\mbox{IG}(e^{-2\rho}; a_\tau/2, b_\tau/2) 2e^{-2\rho}.\\
\tilde{\pi}_2(\alpha_{\gamma},\kappa_\gamma,\rho) &=& \prod_{\gamma_j\geq 1} e^{\sqrt{2} - g_E/\alpha_j^2} N(\alpha_j; 0, g_E)
\prod_{\gamma_j=2} N(\kappa_j; 0, g_S n (S_j^\top S_j)^{-1}) \\
&\times & \mbox{IG}(e^{-2\rho}; a_\tau/2, b_\tau/2) 2e^{-2\rho}.
\end{eqnarray*}

The gradients of $\log \tilde{\pi}_1$ and $\log \tilde{\pi}_2$ with respect to $(\alpha_\gamma,\kappa_\gamma)$ are the same expressions given for
$\log \pi_1$ and $\log \pi_2$ above.
The gradient and hessian with respect to $\rho$ are
\begin{eqnarray*}
\nabla_\rho \log \tilde{\pi}_l(\alpha_{\gamma},\kappa_\gamma,\rho) &=& a_\tau - b_\tau e^{2\rho} \\
\nabla^2_\rho \log \tilde{\pi}_l(\alpha_{\gamma},\kappa_\gamma,\rho) &=& -2 b_\tau e^{2\rho}
\end{eqnarray*}
for $l=1,2$, and clearly $\nabla_\rho \nabla_\alpha= \nabla_\rho \nabla_\kappa=0$.

\section{Prior elicitation}
\label{app:prior_elicitation}

We discuss a range of default prior parameter values $(g_L,g_M,g_E,g_S)$ that we view as reasonable for most applications.
We focus on AFT models where the relevance of a covariate is based on median survival,
but the discussion applies to Cox models where one measures relevance via log-hazard ratios.
Basic considerations give a fairly narrow range of $(g_L,g_M,g_E,g_S)$ that we deem reasonable in applications.
Without loss of generality assume that continuous covariates have zero mean and unit variance.
Then, $e^{\beta_j}$ is the increase in median survival for a unit standard deviation increase in $x_j$ (for continuous covariates) or between two categories of a discrete covariate.
Suppose that a small change in survival time, say $<15\%$ (\textit{i.e.}~$e^{\vert\beta_j\vert}<1.15$), is practically irrelevant.
We set non-local prior dispersions $(g_M,g_E)$ that assign low prior probability to this range, specifically
\begin{eqnarray}
  P( \vert \beta_j \vert > \log(t) ) = 0.99,
  \label{eq:prior_g}
\end{eqnarray}
where $t$ is a user-defined practical significance threshold.
We consider $t=1.1,1.15,1.2$ on the grounds that smaller effects are rarely viewed as relevant. 
The probability in \eqref{eq:prior_g} is under the marginal priors $\pi_M(\beta_j)$, $\pi_E(\beta_j)$ and $\pi_L(\beta_j)$,
which depend on $g_M$, $g_E$, $g_L$ and on $(a_\tau,b_\tau)$ (Supplementary Material, Section \ref{app:marginalprior}).
By default we set $a_\tau=b_\tau=3$ 
so that the marginal prior variance is finite.
Then, for $t=1.1,1.15,1.2$ one obtains $g_M=0.089,0.192,0.326$ and $g_E=0.042,0.091,0.154$ respectively, see Figure \ref{fig:MOMpriors}, our recommended defaults being $g_M=0.192$ and $g_E=0.091$.
Regarding $\pi_L$, we adopt the classical unit information prior default $g_L=1$ \citep{schwarz:1978}.

In probit regression relevance is measured by a covariate's effect on the success probability, which leads to different defaults.
Suppose that covariates that alter such probability by less than 0.05 are practically irrelevant, then one obtains $g_M=0.139$ and $g_E=0.048$ \citep{RTJ13}.
As our main focus is survival, for details we refer the reader to \cite{RTJ13}.

Finally, consider $g_S$, the prior dispersion for non-linear effects.
Mimicking the unit information prior would lead to $g_S=1$ but we view this choice as inappropriate,
since it implies the belief that the predictive power of each covariate grows unboundedly with the basis dimension $r$.
To see this, consider $\omega=\delta_j^\top S_j^\top S_j^\top \delta_j/(n\sigma^2)$, the ratio of the variance explained by $S_j$ relative to the error variance.
The marginal prior on $\omega/ g_S$ induced by $\pi_L$ , $\pi_M$ and  $\pi_E$ is
a chi-squared distribution with $r$ degrees of freedom, implying that $\mathbb{E}(\omega)= g_S r$.
By default we set $g_S=1/r$, so that $\mathbb{E}(\omega)=1$ and in particular stays bounded as a function of $(n,r)$.

\clearpage
\section{Within-model computations}
\label{app:withinmodel_comput}

We calculate marginal likelihoods via Laplace approximations,
see \cite{kass:1990} for classical results on their asymptotic validity and \cite{RT17} for a recent discussion in a Bayesian model selection setting.
That is, we approximate the numerator in \eqref{eq:modelposterior} via
\begin{align}
\widehat{p}(y \mid \gamma)= \exp\{\ell(\tilde{\eta}_\gamma) + \log \pi(\tilde{\eta}_\gamma) \}
\frac{(2\pi)^{d_\gamma/2}}{\left|H(\tilde{\eta}_\gamma) + \nabla^2 \log\pi(\tilde{\eta}_\gamma) \right|^{1/2}},
\end{align}
where $\tilde{\eta}= \arg\max_{\eta_\gamma} \ell({\eta}_\gamma) + \log \pi({\eta}_\gamma)$ is the maximum a posteriori under prior $\pi(\eta_\gamma)$.
The expressions for the log-likelihood and log-prior hessians $H$ and $\nabla^2 \log \pi$ are in Sections \ref{ssec:likelihood} and \ref{app:prior_reparam_logtau} (respectively).
Standard optimization can be employed to obtain $\tilde{\eta}_\gamma$, \textit{e.g.} 
Newton's algorithm \citep{therneau:2000}.
Newton's algorithm is very efficient when the number of parameters $d_\gamma$ is small, but requires matrix inversions that do not scale well to larger $d_\gamma$.
In contrast the Coordinate Descent Algorithm (CDA) typically requires more iterations but, since its per-iteration cost is linear in $d_\gamma$,
requires a lesser total computation time for large $d_\gamma$ \citep{simon:2011,breheny:2011}. 
For this reason we developed both a Newton algorithm and a CDA (Section \ref{app:cda}),
and use the former for small models ($d_\gamma \leq 15$) and CDA for larger models.

Finally, we discuss initializing the optimization algorithm and computing $\log \Phi$ and Mill's ratio featuring in the log-likelihood and its derivatives.
Addressing these issues can significantly increase speed, \textit{e.g.}~for the TGFB data in Section \ref{ssec:ftbrs} with $d_\gamma=868$
they reduced the cost of 1,000 Gibbs iterations from $>$4 hours to 38 seconds.
These times are under a single-core desktop running Ubuntu 18.04, Intel i7 3.40GHz processor and 32Gb RAM.
Let $\gamma^{(b)}$ be the model visited at iteration $b$ of some model search algorithm.
We consider two possible initial values $\widehat{\eta}_{\gamma}^{(0)}$ for the optimization algorithm under a new model $\gamma$:
the value $\widehat{\eta}_0$ maximizing a quadratic expansion of the log-posterior at $(\alpha,\kappa,\tau)=0$,
and the optimal value in the previously visited model $\tilde{\eta}_{\gamma^{(b-1)}}$.
If the log-posterior at $\widehat{\eta}_0$ is larger than when evaluated at $\tilde{\eta}_{\gamma^{(b-1)}}$,
we set $\widehat{\eta}_{\gamma}^{(0)}=\widehat{\eta}_0$, else we set $\widehat{\eta}_\gamma^{(0)}=\tilde{\eta}_{\gamma^{(b-1)}}$.
Since $\gamma$ and $\gamma^{(b)}$ differ by $\leq r$ variables, the optimization algorithm typically converges in a few iterations.

\section{Approximation to the Normal log distribution function and derivatives}
\label{app:napp}

After comparing several approximations we found that one may combine
the Taylor series and asymptotic expansions in \cite{abramowitz:1965} (page 932, Expressions 26.2.16 and 26.2.12) for $\Phi(z)$
with an optimized Laplace continued fraction in \cite{lee_chuin:1992} (Expression (5.3)) for $r(z)$ as $z \rightarrow -\infty$. 
Specifically this results in approximating $\Phi(z)$ with
\begin{align}
  \widehat{\Phi}(z)= \begin{cases}
    \phi(z) (-1/z + 1/z^3 - 3/z^5)^{-1} & \mbox{ if } z \leq -3.4470887
    \\
    Q(z)   & \mbox{ if } z \in (-3.4470887,0]
    \\
    1-Q(z) & \mbox{ if } z \in (0,3.4470887]
    \\
    1- \phi(z) (1/z - 1/z^3 + 3/z^5)^{-1} & \mbox{ if } z > 3.4470887
    \end{cases}
\label{eq:apnorm}
\end{align}
where
$Q(z)= \phi(z) (a_1 t + a_2 t^2 + a_3 t^3)^{-1}$,
$a_1= 0.4361836$, $a_2= -0.1201676$, $a_3= 0.9372980$
and $t= (1 + \mbox{sign}(z) 0.33267 z)^{-1}$.
We approximate $r(z)$ with the continued fraction
\begin{align}
  \widehat{r}(z)=
  -z + \frac{1}{-z+} \frac{2}{-z+} \frac{3}{-z+} \frac{4}{-z+} \frac{5}{-z+} \frac{11.5}{-z+4.890096}
\label{eq:ainvmillsnorm}
\end{align}
if $z \leq -1.756506$ and by $\widehat{r}(z)=\phi(z) / \Phi(z)$ if $z > -1.756506$.
The cutoffs defining the pieces in $\widehat{\Phi}(z)$ and $\widehat{r}(z)$ were set such that both functions are continuous.
$\widehat{r}(z)$ has maximum absolute and relative errors $<0.000185$ and $<0.000102$ respectively,
and for $\widehat{D}(z)=\widehat{r}(-z)^2 - z \widehat{r}(-z)$ they are $<0.000424$ and $<0.000505$.

\section{Coordinate Descent Algorithm}
\label{app:cda}
In order to calculate the MLE and MAP, we propose the following Coordinate Descent algorithm, which is formulated for a generic function $f:{\mathbb R}^p \times  {\mathbb R}_+ \to {\mathbb R}$, which can be either the log-likelihood or the log-posterior.

\begin{algorithm}
\caption{Coordinate Descent Algorithm}\label{alg:CDA}
\begin{algorithmic}[1]
\STATE Select the initial points $\eta^{(0)} = (\alpha^{(0)},\tau^{(0)})$, and the number of iterations $M$.
\STATE For $j = 1,\dots, p$
\begin{enumerate}[label=\bfseries S\arabic*.]
\item Define
\begin{eqnarray*}
\tilde{f}(a) = f(\alpha_1^{(0)},\dots,\alpha_{j-1}^{(0)},a,\alpha_{j+1}^{(0)},\dots,\alpha_p^{(0)},\tau^{(0)}).
\end{eqnarray*}
\item Define $a_{new} = \alpha_j^{(0)} - \tilde{f}'(\alpha_j^{(0)})/\tilde{f}''(\alpha_j^{(0)})$.
\item If $\tilde{f}(a_{new}) > \tilde{f}(\alpha_j^{(0)})$ then update $\alpha_j^{(1)} = a_{new}$.
\item If $\tilde{f}(a_{new}) \leq \tilde{f}(\alpha_j^{(0)})$ reduce the step size $a_{new}= \alpha_j^{(0)} - c  (\tilde{f}'(\alpha_j^{(0)}) / \tilde{f}''(\alpha_j^{(0)}))$, for some $c\in(0,1)$, until $\tilde{f}(a_{new}) > \tilde{f}(\alpha_j^{(0)})$, and update $\alpha_j^{(1)} = a_{new}$.
\end{enumerate}
\STATE Define
\begin{eqnarray*}
\tilde{f}(t) = f(\alpha^{(1)},t).
\end{eqnarray*}
\begin{enumerate}[label=\bfseries T\arabic*.]
\item Define $\tau_{new} = \tau^{(0)} - \tilde{f}'(\tau^{(0)})/\tilde{f}''(\tau^{(0)})$.
\item If $\tilde{f}(\tau_{new}) > \tilde{f}(\tau^{(0)})$ then update $\tau^{(1)} = \tau_{new}$.
\item If $\tilde{f}(\tau_{new}) \leq \tilde{f}(\tau^{(0)})$ reduce the step size $\tau_{new}= \tau^{(0)} - c  (\tilde{f}'(\tau^{(0)}) / \tilde{f}''(\tau^{(0)}))$, for some $c\in(0,1)$, until $\tilde{f}(\tau_{new}) > \tilde{f}(\tau^{(0)})$, and update $\tau^{(1)} = \tau_{new}$.
\end{enumerate}

\STATE If $m=M$ or $f(\alpha^{(m)}, \tau^{(m)}) - f(\alpha^{(m-1)} \tau^{(m-1)}) < c$, then stop. Otherwise set $m=m+1$ and go back to 2.
\end{algorithmic}
\end{algorithm}
The stopping criterion is either reaching the maximum number of iterations $M$ or (hopefully) converging earlier than that. In practice, this is diagnosed by the increase in the log-likelihood (log-posterior) being smaller than $c$, for some small $c$ (say $c=0.01$). This can save massive time, and is a way to diagnose convergence (in contrast to stopping after $M$ iterations, when one may still not have converged).

\clearpage
\section{Proofs of Asymptotic results}
\label{sec:proof_asymp}
%
%
%

%

We shall assume the following conditions.

\begin{enumerate}[label=\bfseries A\arabic*.]
\item \color{black} The parameter space is $\Gamma_{\gamma}= {\mathbb R}^{p_\gamma + r s_\gamma}\times{\mathbb R}^+$. \color{black}

\item There is some $\tilde{n}$ such that $(X_{o,\gamma},S_{o,\gamma})^{\top} (X_{o,\gamma},S_{o,\gamma})$ is strictly positive definite almost surely for all $n>\tilde{n}$.

\item \color{black} There exists a maximum $\eta_\gamma^* \in \Gamma_\gamma$ of $M(\eta_\gamma)=E_{F_0}[m(\eta_\gamma)]$ such that $M(\eta_\gamma^*) < \infty$,
and $E_{F_0}[|m(\eta_\gamma)|] < \infty$ for all $\eta_\gamma \in \Gamma_\gamma$.  
Further, there exists a maximum $\widehat{\eta}_\gamma \in \Gamma_\gamma$ of $M_n(\eta_\gamma)$ with probability 1, as $n \rightarrow \infty$.
\color{black}

%

\item There exists a neighborhood ${\mathcal B}_{\eta_{\gamma}^*}$ of  $\eta_{\gamma}^*$ such that, for any $\eta_{\gamma} \in {\mathcal B}_{\eta_{\gamma}^*}$,
\begin{eqnarray*}
{\mathbb E}_{F_0} \left[\sup_{\eta_\gamma\in {\mathcal B}_{\eta_\gamma^*}} \vert\vert \nabla   m(\eta_\gamma)\vert\vert^2 \right] &<& \infty.
\end{eqnarray*}

\item The entries of the second derivative matrix $ \nabla^2_{\eta_{\gamma}}M(\eta_\gamma^*) $ are finite. 
 
\end{enumerate}

Briefly, A2 is a minimal condition that the design matrix for uncensored observations has full-rank, which ensures that the log-likelihood is strictly concave with respect to the regression coefficients. 
A3-A5 are standard in asymptotic studies.
\color{black} A3 assumes the existence of the maximum likelihood estimator $\widehat{\eta}_\gamma$ (for large enough $n$) and of the maximizer $\eta_\gamma^*$ of the expectation of the log-likelihood under $F_0$, which implies that neither $\widehat{\eta}_\gamma$ nor $\eta_\gamma^*$ occur at the boundary of $\Gamma_\gamma$ and, by concavity, both $\widehat{\eta}_\gamma$ and $\eta_\gamma^*$ are unique. That is, these maxima do not occur at an infinite value for the regression coefficients nor at the error variance being 0 or infinite. Such values correspond to degenerate cases, e.g. where one attains perfect predictions, and we exclude them from consideration.
Further, the requirement that $E_{F_0}(|m(\eta_\gamma)|)<\infty$ in A3
\color{black}
as well as A4-A5 can be seen as conditions on the tails of the data-generating distribution $F_0$ of the survival and censoring process. For example, notice that the term $\vert \log \Phi(t) \vert$ in the function $m()$ is decreasing in $t$, and that $ \vert \log\Phi(t) \vert = \lO(\vert t\vert^2) $ as $t\to{-\infty}$, hence A3 basically requires a finite squared error when predicting the censoring time $c_1$ with an arbitrary $(\alpha_\gamma,\kappa_\gamma)$. These conditions if survival and censoring times are bounded, as typical in applied research.

\subsection{Proof of Proposition \ref{prop:Consistency}}
\label{app:proof_consistency}


The proof is based on showing that under conditions A1-A3, the assumptions in the consistency Theorem 5.7 from \cite{vandervaart:1998} are satisfied. This requires appealing to the concavity properties of the log likelihood function discussed in the   main paper.

Let $M_n(\eta_\gamma) = n^{-1}\ell(\eta_\gamma)$ be the average log-likelihood evaluated at $\eta_\gamma$. 
The assumptions that data are generated \emph{i.i.d.} from $F_0$ \color{black} and that ${\mathbb E}_{F_0} [|m(\eta_\gamma)|] < \infty$ \color{black} imply that, by law of large numbers, 
$M_n(\eta_\gamma) \stackrel{P}{\rightarrow} M(\eta_\gamma)$ for each $\eta_\gamma\in\Gamma_\gamma$, where
$M(\eta_\gamma) = $ 
\begin{align}
& P_{F_0}(u_1=1) {\mathbb E}_{F_0} \left[ \log(\tau) -\dfrac{1}{2}\log(2\pi) - \dfrac{1}{2} \left(\tau \log(o_1)-x_1^{\top}\alpha_\gamma -s_1^{\top} \kappa_\gamma \right)^2 \mid u_1=1 \right]
 \nonumber \\
& + P_{F_0}(u_1=0)
 {\mathbb E}_{F_0} \left[ \log\Phi\left(x_1^{\top}\alpha_\gamma + s_1^{\top} \kappa_\gamma - \tau \log(c_1)\right) \mid u_1=0 \right]
\label{eq:expected_logl_general}
\end{align}
is the expectation of $M_n(\eta_\gamma)$ under the data-generating $F_0$.

The aim is to first show that $M_n(\eta_\gamma)$ converges to its expected value $M(\eta_\gamma)$ uniformly in $\eta_\gamma$, and then show that this implies that $\widehat{\eta}_\gamma \stackrel{P}{\longrightarrow} \eta_\gamma^*$, where recall that $\widehat{\eta}_\gamma$ is the maximum of $M_n(\eta_\gamma)$ and $\eta_\gamma^*$ that of $M(\eta_\gamma)$.
To see that $M_n(\eta_\gamma)$ converges to $M(\eta_\gamma)$, uniformly in $\eta_\gamma$,
we note that under Conditions A1-A2 and by the results in \cite{B81}, $M_n(\eta_\gamma)$ is a sequence of concave functions in $\eta_\gamma$.
Then, \color{black} recalling that $\Gamma_\gamma$ is an open convex set by assumption, \color{black}
the convexity lemma in \cite{pollard:1991} gives that
\begin{eqnarray}\label{SC1}
\sup_{\eta_\gamma \in K} \left\vert M_n(\eta_\gamma) - M(\eta_\gamma)\right\vert \stackrel{P}{\longrightarrow} 0,
\end{eqnarray}
for each compact set $K\subseteq\Gamma_\gamma$,
\color{black} and that $M(\eta_\gamma)$ is concave in $\Gamma_\gamma$. \color{black}

\color{black} Assumption A3 on the existence of the maxima $\eta_\gamma^*$ and $\widehat{\eta}_\gamma$ implies that neither occurs on the boundary of $\Gamma_\gamma$, hence $\eta_\gamma^* \in K$ and $\widehat{\eta}_\gamma \in K$ (with probability 1, as $n$ grows) for some compact set $K \subseteq \Gamma_\gamma$ and, by concavity, $\eta_\gamma^*$ and $\widehat{\eta}_\gamma$ are unique. \color{black}
That is, for a distance measure $d()$ and every $\varepsilon>0$ we have
\begin{eqnarray}\label{SC2}
\sup_{d(\eta_\gamma^*,\eta_\gamma) \geq \varepsilon } M(\eta_\gamma) <M(\eta_\gamma^*).
\end{eqnarray}
The consistency result $\widehat{\eta}_\gamma \stackrel{P}{\longrightarrow} \eta_\gamma^*$ follows directly from \eqref{SC1} and \eqref{SC2} together with Theorem 5.7 from \cite{vandervaart:1998}.

\subsection{Proof of Proposition \ref{prop:AsympNorm}}
\label{app:proof_asympnorm}

The proof is based on applying Theorem 5.23 in \cite{vandervaart:1998}.
The conditions in that theorem require MLE consistency, which we already proved in Proposition \ref{prop:Consistency},
showing that the expected log-likelihood under $F_0$ has a non-singular hessian at the unique maximizer $\eta_\gamma^*$,
and that log-likelihood increments in a neighbourhood of $\eta_\gamma^*$ have finite expectation under $F_0$.

To see the latter define,
\begin{eqnarray*}
m_{\eta_\gamma}(y_1,x_1,s_1) \equiv m(\eta_\gamma) &=& u_1\left[ \log(\tau) -\dfrac{1}{2}\log(2\pi) - \dfrac{1}{2} \left(\tau \log(o_1)-x_1^{\top}\alpha_\gamma -s_1^{\top} \kappa_\gamma \right)^2  \right]
 \nonumber \\
& + & (1-u_1)
 \left[ \log\Phi\left(x_1^{\top}\alpha_\gamma + s_1^{\top} \kappa_\gamma - \tau \log(c_1)\right) \right],
\end{eqnarray*}
let ${\mathcal B}_{\eta^*_{\gamma}}\subset \Gamma_\gamma$ be a neighbourhood of $\eta_\gamma^*$
and consider $\eta_1,\eta_2 \in {\mathcal B}_{\eta^*_\gamma}$.
We need to show that $\vert m_{\eta_1}(y_1,x_1,s_1) - m_{\eta_2}(y_1,x_1,s_1) \vert$ has finite expectation under $F_0$.
Using the mean value theorem and the Cauchy-Schwarz inequality it follows that, with probability 1,
\begin{eqnarray*}
\vert m_{\eta_1}(y_1,x_1,s_1) - m_{\eta_2}(y_1,x_1,s_1) \vert &=& \vert \nabla m_{\eta_c}(y_1,x_1,s_1)^{\top} (\eta_1-\eta_2)\vert\\
&\leq& \vert \vert \nabla m_{\eta_c}(y_1,x_1,s_1) \vert \vert \cdot \vert \vert \eta_1-\eta_2 \vert \vert
\leq K(y_1,x_1,s_1) \cdot \vert \vert \eta_1-\eta_2 \vert \vert,
\end{eqnarray*}
where $\nabla m_{\eta_\gamma}(y_1,x_1,s_1)$ is the gradient of $m_{\eta_\gamma}(y_1,x_1,s_1)$,
$\eta_c = (1-c)\eta_1+c\eta_2$ for some $c\in(0,1)$
and $K(y_1,x_1,s_1)= \sup_{{\mathcal B}_{\eta_\gamma^*}} \vert\vert \nabla m_{\eta_\gamma}(y_1,x_1,s_1) \vert\vert$.
From assumption A4 it follows that,
\begin{eqnarray*}
\int K(y_1,x_1,s_1)^2 dF_0(o_1,z_1,c_1) < \infty,
\end{eqnarray*}
where $z_1 = (x_1^{\top},s_1^{\top})^{\top}$.

We now show that the Hessian of the expected log-likelihood $M(\eta_\gamma)$ in \eqref{eq:expected_logl_general} is non-singular at $\eta_\gamma^*$.
For ease of notation let $\rho = P_{F_0}(u_1=1)$. Then
\begin{eqnarray*}
M(\eta_{\gamma}) &=& \rho \log(\tau) -\dfrac{\rho}{2}\log(2\pi) - \dfrac{\rho}{2} \int  \left(\tau \log(o_1)-x_1^{\top}\alpha_\gamma - s_1^{\top} \kappa_\gamma\right)^2 dF_0(o_1,z_1\mid u_1=1)  \\
&+&  (1-\rho) \int  \log\Phi\left(x_1^{\top}\alpha_\gamma + s_1^{\top} \kappa_\gamma - \tau \log(c_1)\right) dF_0(c_1,z_1 \mid u_1=0).
\end{eqnarray*}
To obtain the gradient of $M(\eta_{\gamma})$, note that under Conditions A3-A4 we can apply Leibniz's integral rule to differentiate under the integral sign, and hence
\begin{eqnarray*}
  \nabla_{(\alpha_\gamma,\kappa_\gamma)} M(\eta_{\gamma}) &=& \rho \int \left(\tau \log(o_1)-x_1^{\top}\alpha_\gamma - s_1^{\top} \kappa_\gamma\right) \begin{pmatrix} x_1 \\ s_1 \end{pmatrix} dF_0(o_1,z_1\mid u_1=1)  \\
&+&  (1-\rho) \int   r\left(x_1^{\top}\alpha_\gamma + s_1^{\top} \kappa_\gamma - \tau \log(c_1)\right) \begin{pmatrix} x_1 \\ s_1 \end{pmatrix}  dF_0(c_1,z_1 \mid  u_1=0),\\
  \nabla_{\tau} M(\eta_{\gamma}) &=&  \dfrac{\rho}{\tau}  -  \rho \int  \left(\tau \log(o_1)-x_1^{\top}\alpha_\gamma - s_1^{\top} \kappa_\gamma\right)\log(o_1) dF_0(o_1,z_1\mid u_1=1)  \\
&-& (1-\rho) \int   r\left(x_1^{\top}\alpha_\gamma + s_1^{\top} \kappa_\gamma - \tau \log(c_1)\right)\log(c_1)  dF_0(c_1,z_1 \mid u_1=0),
\end{eqnarray*}
Similarly, the entries of the Hessian matrix are
\begin{eqnarray*}
  \nabla^2_{(\alpha_\gamma,\kappa_\gamma)} M(\eta_{\gamma}) &=& -\rho \int  \begin{pmatrix} x_1 \\ s_1 \end{pmatrix}\begin{pmatrix} x_1 \\ s_1 \end{pmatrix}^{\top}    dF_0(z_1 \mid u_1 = 1)  \\
&+&  (1-\rho) \int  D\left(\tau \log(c_1) - x_1^{\top}\alpha_\gamma - s_1^{\top} \kappa_\gamma \right) \begin{pmatrix} x_1 \\ s_1 \end{pmatrix} \begin{pmatrix} x_1 \\ s_1 \end{pmatrix}^{\top}  dF_0(c_1,z_1 \mid u_1=0), \\
    \nabla_{\tau} \nabla_{(\alpha_\gamma,\kappa_\gamma)} M(\eta_{\gamma}) &=& \rho \int \log(o_1) \begin{pmatrix} x_1 \\ s_1 \end{pmatrix} dF_0(o_1,z_1 \mid u_1=1) \\
&+&  (1-\rho) \int   D\left(\tau \log(c_1) - x_1^{\top}\alpha_\gamma - s_1^{\top} \kappa_\gamma\right) \log(c_1) \begin{pmatrix} x_1 \\ s_1 \end{pmatrix}  dF_0(c_1,z_1 \mid u_1=0),\\
  \nabla^2_{\tau} M(\eta_{\gamma}) &=&  -\dfrac{\rho}{\tau^2}  -  \rho \int  \log(o_1)^2 dF_0(o_1\mid z_1, u_1=1) dF_0(z_1 \mid u_1 = 1) \\
&-&  (1-\rho) \int  D\left(\tau \log(c_1) - x_1^{\top}\alpha_\gamma - s_1^{\top} \kappa_\gamma \right)\log(c_1)^2  dF_0(c_1,z_1 \mid u_1=0).
\end{eqnarray*}
The finiteness of $M(\eta_{\gamma})$, as well as the finiteness of the entries of its gradient and Hessian matrix follows by Conditions A3-A5. 
From Proposition \ref{prop:Consistency}, we have that $M(\eta_{\gamma})$ is concave and, consequently, the Hessian
\begin{eqnarray*}
V_{\eta_{\gamma}} =\left( \begin{matrix}
   \nabla^2_{(\alpha_\gamma,\kappa_\gamma)} M(\eta_{\gamma}) & \nabla_{\tau} \nabla_{(\alpha_\gamma,\kappa_\gamma)} M(\eta_{\gamma})\\
   \nabla_{\tau} \nabla_{(\alpha_\gamma,\kappa_\gamma)} M(\eta_{\gamma}) &   \nabla^2_{\tau} M(\eta_{\gamma})
 \end{matrix}\right),
\end{eqnarray*}
is non-singular at $\eta_{\gamma}^*$. Thus, the asymptotic normality follows by Theorem 5.23 from \cite{vandervaart:1998} together with the consistency results in Proposition \ref{prop:Consistency}.

\subsection{Proof of Proposition \ref{prop:BFRatesI}}
\label{app:proof_bfrates}

We aim to characterize the asymptotic behaviour of Laplace-approximated Bayes factors
\begin{align}
B_{\gamma \gamma^*}&=\frac{\widehat{p}(y \mid \gamma)}{\widehat{p}(y \mid \gamma^*)}=
(2 \pi)^{\frac{d_\gamma - d_{\gamma^*}}{2}} e^{T_1} T_2 T_3
\nonumber \\
T_1&= \ell(\tilde{\eta}_\gamma) -  \ell(\tilde{\eta}_{\gamma^*})
\nonumber \\
T_2&= \frac{\pi(\tilde{\eta}_\gamma \mid \gamma)}{\pi(\tilde{\eta}_{\gamma^*} \mid \gamma^*)}
\nonumber \\
T_3&= \frac{ \left| H(\tilde{\eta}_{\gamma^*}) + \nabla^2 \log \pi(\tilde{\eta}_{\gamma^*}) \right|^{\frac{1}{2}}}{\left| H(\tilde{\eta}_{\gamma}) + \nabla^2 \log \pi(\tilde{\eta}_{\gamma^*}) \right|^{\frac{1}{2}}},
\label{eq:BFLap}
\end{align}
where $y_i=\min\{o_i,c_i\}$ and the data-generating truth $(o_i,c_i,z_i) \sim F_0$ satisfies A2-A4.
$H(\tilde{\eta}_{\gamma^*})$ and $H(\tilde{\eta}_\gamma)$ denote the log-likelihood Hessians under models $\gamma^*$ and $\gamma$ (respectively),
evaluated at the posterior modes $\tilde{\eta}_{\gamma^*}$ and $\tilde{\eta}_\gamma$.

The proof strategy is to characterize each term in \eqref{eq:BFLap} individually, then combine the results.
First, $(2 \pi)^{(d_\gamma - d_{\gamma^*})/2}$ is a constant since $d_\gamma$ and $d_{\gamma^*}$ are fixed by assumption. Now, note that 
\begin{eqnarray}\label{eq:ratiodet}
T_3=
  n^{\frac{d_{\gamma^*} - d_\gamma}{2}} \frac{\left\vert n^{-1} \left[ H(\tilde{\eta}_{\gamma^*}) + \nabla^2 \log\pi (\tilde{\eta}_{\gamma^*}) \right]\right\vert^{\frac{1}{2}} }{\left\vert n^{-1} \left[ H(\tilde{\eta}_\gamma) + \nabla^2 \log\pi (\tilde{\eta}_\gamma) \right] \right\vert^{\frac{1}{2}}}.
  \nonumber
\end{eqnarray}
By Proposition \ref{prop:Consistency} together with Proposition 2(i) from \cite{RT17}, we have that the posterior modes
$\tilde{\eta}_\gamma \stackrel{P}{\longrightarrow} \eta_\gamma^*$ and $\tilde{\eta}_{\gamma^*} \stackrel{P}{\longrightarrow} \eta_{\gamma^*}^*$
under the block-Zellner prior $\pi_L$, the MOM-Zellner prior $\pi_M$ and also under the eMOM-Zellner prior $\pi_E$.
Then, appealing to the continuous mapping theorem and the asymptotic Hessians $H(\eta^*_\gamma)$ and $H(\eta^*_{\gamma^*})$ being negative definite (Proposition \ref{prop:AsympNorm} and convexity lemma),
\begin{align}
\frac{\left\vert n^{-1} \left[ H(\tilde{\eta}_{\gamma^*}) + \nabla^2 \log\pi (\tilde{\eta}_{\gamma^*}) \right]\right\vert^{\frac{1}{2}} }{\left\vert n^{-1} \left[ H(\tilde{\eta}_\gamma) + \nabla^2 \log\pi (\tilde{\eta}_\gamma) \right] \right\vert^{\frac{1}{2}}}
\stackrel{P}{\longrightarrow}
\frac{\left| H(\eta^*_{\gamma^*}) \right|^{\frac{1}{2}}}{\left| H(\eta^*_\gamma) \right|^{\frac{1}{2}}}
\nonumber
\end{align}
where the right-hand side is fixed, since $\eta^*_{\gamma^*}$ and $\eta_\gamma^*$ are fixed by assumption.
Therefore
\begin{align}
T_3 n^{\frac{d_\gamma - d_{\gamma^*}}{2}}  \stackrel{P}{\longrightarrow}
\frac{\left| H(\eta^*_{\gamma^*}) \right|^{\frac{1}{2}}}{\left| H(\eta^*_\gamma) \right|^{\frac{1}{2}}} \in (0,\infty)
\label{eq:limit_hess}
\end{align}

It is worth noticing that although the asymptotic expression for the ratio of Hessian determinants is reminiscent of the case without censoring \citep{JR12,rossell:2018}, its behavior is different as here $H()$ is a weighted sum across uncensored and censored observations, and the latter features a discount factor driven by $D()$ displayed in Figure \ref{fig:DZ} (see Section \ref{app:proof_asympnorm}).

{
To characterize $T_1$ and $T_2$ in \eqref{eq:BFLap} we must consider separately the case where $\gamma^*\not\subset \gamma$ and the case where $\gamma^*\subset \gamma$.
It is useful to note that from the continuous mapping theorem,
$\pi(\tilde{\eta}_\gamma \mid \gamma) \stackrel{P}{\longrightarrow} \pi(\eta^*_\gamma \mid \gamma)$ for any $\gamma$ and any continuous prior $\pi$, hence
\begin{align}
T_2
\frac{\pi(\eta^*_{\gamma^*} \mid \gamma^*)}{\pi(\eta^*_\gamma \mid \gamma)} \stackrel{P}{\longrightarrow} 1
\nonumber
\end{align}
for any prior such that $\pi(\eta^*_\gamma \mid \gamma) \in (0,\infty)$, as is satisfied by $\pi_L$.
}

\begin{enumerate}[leftmargin=*,label=(\roman*)]


\item {\bf Case $\gamma^*\subset \gamma$}. In this case, we have $M(\eta_\gamma^*) - M(\eta_{\gamma^*}^*)=0$.
The idea is to first show that $T_1=\ell(\tilde{\eta}_\gamma) -  \ell(\tilde{\eta}_{\gamma^*})= \lOp(1)$ using standard theory on the likelihood ratio statistic.
From \cite{RT17} and the consistency of the maximum likelihood estimators, we get that  $T_1= \ell(\widehat{\eta}_\gamma) -  \ell(\widehat{\eta}_{\gamma^*}) + \lo_p(1)$. Let $U= V_{\eta_\gamma^*}$,  $W_{\gamma^*}= V_{\eta_\gamma^*}^{-1} {\mathbb E}[ \nabla m_{\eta_\gamma^*} \nabla m_{\eta_\gamma^*}^{\top}] V_{\eta_\gamma^*}^{-1}$. From Proposition \ref{prop:AsympNorm} we have that
\begin{align}
  \sqrt{n} (\widehat{\eta}_\gamma - \eta_\gamma^*) \stackrel{D}{\longrightarrow} N \left( 0,  W_{\gamma^*} \right)
  \Rightarrow
 z \stackrel{D}{\longrightarrow} N \left( 0,  I \right),
    \nonumber
\end{align}
where $z= \sqrt{n} W_{\gamma^*}^{-1/2} (\widehat{\eta}_\gamma - \eta_\gamma^*)$. Now, note that the Hessian matrix of the log-likelihood converges to a non-singular matrix by Proposition \ref{prop:AsympNorm}, and that $\sqrt{n} \left( \widehat{\eta}_\gamma -  \widehat{\eta}_{\gamma^*}\right)= \lOp(1)$ (with respect to the Euclidean norm)
by Propositions \ref{prop:Consistency} and \ref{prop:AsympNorm}. Then, we can expand the likelihood ratio as (see Chapter 16 of \citealp{vandervaart:1998})
\begin{eqnarray*}
2 [ \ell(\widehat{\eta}_\gamma) -  \ell(\widehat{\eta}_{\gamma^*}) ]&=&
  n \left( \widehat{\eta}_\gamma -  \widehat{\eta}_{\gamma^*}\right)^{\top} U^{-1} \left( \widehat{\eta}_\gamma -  \widehat{\eta}_{\gamma^*}\right)  +  \lo_p(1)\\
&=&  n \widehat{\eta}_{\gamma \setminus \gamma^*}^{\top} U^{-1}_{\gamma \setminus \gamma^*} \widehat{\eta}_{\gamma \setminus \gamma^*}  +  \lo_p(1)
=  z^\top W_{\gamma \setminus \gamma^*}^{1/2} U^{-1}_{\gamma \setminus \gamma^*} W_{\gamma \setminus \gamma^*}^{1/2} z +  \lo_p(1).
\end{eqnarray*}
If the model is correctly specified, then $W_{\gamma \setminus \gamma^*}^{1/2} U^{-1}_{\gamma \setminus \gamma^*} W_{\gamma \setminus \gamma^*}^{1/2}=I$,
obtaining $z^\top z \stackrel{D}{\longrightarrow} \chi^2_{d_\gamma - d_{\gamma^*}}$.
If the model is misspecified, then the right-hand side converges in distribution to
a linear combination of independent chi-square random variables \citep{baldessari:1967} and can be upper bounded by
$$z^\top W_{\gamma \setminus \gamma^*}^{1/2} U^{-1}_{\gamma \setminus \gamma^*} W_{\gamma \setminus \gamma^*}^{1/2} z \leq \lambda\left( W_{\gamma^*}^{1/2} U^{-1}_{\gamma^*} W_{\gamma^*}^{1/2} \right) z^\top z$$
where $\lambda()$ denotes the largest eigenvalue (and is a bounded constant under our assumptions)
and $z^\top z \stackrel{D}{\longrightarrow} \chi^2_{d_\gamma - d_{\gamma^*}}$. In summary,
\begin{align}
2 T_1 \leq  
\lambda\left( W_{\gamma^*}^{1/2} U^{-1}_{\gamma^*} W_{\gamma^*}^{1/2} \right) z^\top z
\stackrel{D}{\longrightarrow} \lambda\left( W_{\gamma^*}^{1/2} U^{-1}_{\gamma^*} W_{\gamma^*}^{1/2} \right) \chi^2_{d_\gamma - d_{\gamma^*}} = \lOp(1).
\label{eq:limit_lrt_spurious}
\end{align}

Finally, we characterize $T_2$ in \eqref{eq:BFLap}. The local prior density is
\begin{equation}\label{eq:localAFTN}
\pi_L(\alpha_\gamma, \kappa_\gamma \mid \tau^2) =  \prod_{\gamma_j\geq 1}N\left(\alpha_j; 0, g_L n/ (x_j^\top x_j)\right)
      \prod_{\gamma_j=2} N\left(\kappa_j; 0, g_S n (S_j^\top S_j)^{-1}\right),
\end{equation}
Since $\gamma^* \subset \gamma$, the optimal parameter values under $\gamma$ can be written as $\eta^*_{\gamma}= (\eta^*_{\gamma^*}, 0)$.
Hence 
$$
\frac{\pi_L(\eta^*_{\gamma^*}, 0 \mid \gamma)}{\pi_L(\eta^*_{\gamma^*} \mid \gamma^*)}= 
\prod_{j: \gamma_j > 0, \gamma^*_j=0}  \left(\dfrac{x_j^{\top} x_j}{2\pi g_L n}\right)^{1/2} \prod_{j : \gamma_j = 2, \gamma^*_j \neq 2} \left\vert \dfrac{S_j^{\top} S_j}{2 \pi g_S n}\right\vert^{1/2}.
$$
Since $x_j^{\top} x_j/n$ converges in probability to a finite and strictly positive constant and $S_j^{\top} S_j/n$ converges in probability to a positive-definite matrix,
\begin{equation}
T_2 \frac{\pi_L(\eta^*_{\gamma^*} \mid \gamma^*)}{\pi_L(\eta^*_{\gamma^*}, 0 \mid \gamma)} \stackrel{P}{\longrightarrow} 1
\Longrightarrow 
T_2 \left( g_L \right)^{(p_\gamma - p_{\gamma^*})/2} \left( g_S n \right)^{r(s_\gamma - s_{\gamma^*})/2} 
\stackrel{P}{\longrightarrow} c \in (0, \infty)
\label{eq:conv_priorratio_local}
\end{equation}

Combining \eqref{eq:BFLap}, \eqref{eq:limit_hess}, \eqref{eq:limit_lrt_spurious}, \eqref{eq:conv_priorratio_local}, and recalling that $d_\gamma= p_\gamma + r s_\gamma$, we obtain
\begin{align}
B_{\gamma \gamma^*}  = \lOp \left( (g_L n)^{\frac{p_{\gamma^*} - p_\gamma}{2}}  (g_S n)^{\frac{r(s_{\gamma^*} - s_\gamma)}{2}} \right),
\nonumber
\end{align}
as we wished to prove.

A similar reasoning applies to characterize $T_2$ for non-local priors. 
Consider first the pMOMZ prior $\pi_M$.
By using \cite{RT17}, Proposition 2(i), it follows that the posterior mode
$\tilde{\alpha}_{\gamma j}-\widehat{\alpha}_{\gamma j}=\lOp(1/n)$, $\tilde{\kappa}_{\gamma j}-\widehat{\kappa}_{\gamma j}=\lOp(1/n)$.
Hence for any coefficient $j$ in $\gamma \setminus \gamma^*$ we have that $\tilde{\alpha}_{\gamma j}=\lOp(n^{-1/2})$  and $\tilde{\kappa}_{\gamma j}=\lOp(n^{-1/2})$,
whereas for $j$ in $\gamma^*$ we have $\tilde{\alpha}_{\gamma j} \stackrel{P}{\longrightarrow} \alpha^*_{\gamma j}$
and $\tilde{\kappa}_{\gamma j} \stackrel{P}{\longrightarrow} \kappa^*_{\gamma j}$.
Consequently, the ratio of prior densities is
\begin{eqnarray*}
&&T_2=\frac{\pi_M(\tilde{\eta}_{\gamma} \mid \gamma )}{\pi_M(\tilde{\eta}_{\gamma^*} \mid \gamma^*)}
=
    \frac{\pi_L(\tilde{\kappa}_\gamma \mid \gamma)}{\pi_L(\tilde{\kappa}_{\gamma^*} \mid \gamma^*)}
   \prod_{j \in \gamma \setminus \gamma^*} \frac{\tilde{\alpha}_{\gamma j}^2 e^{-\frac{\tilde{\alpha}_{\gamma j}^2}{2g_M}}}{g_M^{3/2} \sqrt{2\pi}} 
\prod_{j \in \gamma^*} \frac{\tilde{\alpha}_{\gamma j}^2 e^{-\frac{\tilde{\alpha}_{\gamma j}^2}{2g_M}}}{\tilde{\alpha}_{\gamma^* j}^2 e^{-\frac{\tilde{\alpha}_{\gamma^* j}^2}{2g_M}}},
\end{eqnarray*}
where $e^{-\tilde{\alpha}_{\gamma j}^2/(2g_M)}$ and $e^{-\tilde{\alpha}_{\gamma^* j}^2/(2g_M)}$ converge in probability to finite constants $\in [1,\infty)$, for all $j$. Hence
\begin{eqnarray*}
T_2 &=& 
\left(
\left[ n g_M^{3/2}  \right]^{(p_{\gamma^*} - p_\gamma)}
 \left[ g_S \right]^{\frac{r(s_{\gamma^*} - s_\gamma)}{2}} 
\right) \times \lOp(1)
\end{eqnarray*}

Combining this expression with \eqref{eq:BFLap}, \eqref{eq:limit_hess} and \eqref{eq:limit_lrt_spurious}
\begin{align}
  B_{\gamma \gamma^*}=
  (g_M n)^{\frac{3}{2}(p_{\gamma^*}-p_\gamma )} (g_S n)^{\frac{r(s_{\gamma^*} - s_\gamma)}{2}} \lOp(1),
\nonumber
\end{align}
as we wished to prove.

Consider now the peMOMZ prior $\pi_E$.
From Proposition 2(i) in \cite{RT17} it follows that for any coefficient $j$ in $\gamma \setminus \gamma^*$
we have that $\tilde{\alpha}_{\gamma j}=\lOp(n^{-1/4})$, $\tilde{\kappa}_{\gamma j} - \widehat{\kappa}_{\gamma j}=\lOp(n^{-1/4})$,
and $e^{-g_E/\tilde{\alpha}_{\gamma j}^2}= \lOp(e^{-g_E \sqrt{n}})$,
whereas for any $j$ in $\gamma^*$ we have $\tilde{\alpha}_{\gamma j} \stackrel{P}{\longrightarrow} \alpha_{\gamma j}^*$
and $\tilde{\kappa}_{\gamma j} \stackrel{P}{\longrightarrow} \kappa_{\gamma j}^*$.
Consequently,
\begin{eqnarray*}
T_2= 
   \frac{\pi_L(\tilde{\kappa}_\gamma \mid \gamma)}{\pi_L(\tilde{\kappa}_{\gamma^*} \mid \gamma^*)}
   \prod_{j \in \gamma \setminus \gamma^*} \frac{e^{-\sqrt{2} - \frac{g_E}{\tilde{\alpha}_{\gamma j}^2}} e^{-\frac{\tilde{\alpha}_{\gamma j}^2}{2 g_E}}}{\sqrt{2 \pi g_E}}
\prod_{j \in \gamma^*} 
\frac{e^{- \frac{g_E}{\tilde{\alpha}_{\gamma j}^2}} e^{-\frac{\tilde{\alpha}_{\gamma j}^2}{2 g_E}}}{e^{- \frac{g_E}{\tilde{\alpha}_{\gamma^* j}^2}} e^{-\frac{\tilde{\alpha}_{\gamma^* j}^2}{2 g_E}}}
\end{eqnarray*}
where $e^{- \tilde{\alpha}_{\gamma j}^2 /(2 g_E)}$ and $e^{- \tilde{\alpha}_{\gamma^* j}^2 /(2 g_E)}$ converge in probability to a finite constant $\geq 1$ for all $j$. Hence
\begin{eqnarray*}
T_2 &=& 
g_E^{\frac{p_{\gamma^*} - p_\gamma}{2}}
e^{g_E \sum_{j \in \gamma^*} \frac{1}{(\alpha^*_{\gamma^* j})^2} - \frac{1}{(\alpha^*_{\gamma j})^2}}
 \left[ g_S \right]^{\frac{r(s_{\gamma^*} - s_\gamma)}{2}} 
\times \lOp( e^{-g_E \sqrt{n} (p_\gamma - p_{\gamma^*})} )
\\
&=& g_E^{\frac{p_{\gamma^*} - p_\gamma}{2}}
 g_S^{\frac{r(s_{\gamma^*} - s_\gamma)}{2}} 
\times \lOp( e^{-g_E \sqrt{n} (p_\gamma - p_{\gamma^*})} )
\end{eqnarray*}

Combining this expression with \eqref{eq:BFLap}, \eqref{eq:limit_hess} and \eqref{eq:limit_lrt_spurious} gives the desired result
\begin{align}
  B_{\gamma \gamma^*}
&=(\sqrt{g_E n} e^{g_E \sqrt{n}})^{p_{\gamma^*} - p_\gamma} (g_S n)^{\frac{r(s_{\gamma^*} - s_\gamma)}{2}} \lOp(1).
\nonumber
\end{align}


\item {\bf Case $\gamma^*\not\subset \gamma$}.
By Proposition \ref{prop:Consistency}, the law of large numbers and uniform convergence of the log-likelihood to its expectation $M(\eta)$ it follows that
\begin{eqnarray*}
\frac{1}{n} T_1
\stackrel{P}{\rightarrow} M({\eta}_\gamma^*) - M({\eta}_{\gamma^*}^*)=
P_{F_0}(u_1=1) a_1(\eta_\gamma^*, \eta_{\gamma^*}^*) + P_{F_0}(u_1=0) a_2(\eta_\gamma^*, \eta_{\gamma^*}^*) <0,
\end{eqnarray*}
where $a_1(\eta_\gamma^*, \eta_{\gamma^*}^*)=$
$$ {\mathbb E}_{F_0 \mid u_1=1}\left[ \log\left(\dfrac{\tau^*_{\gamma}}{\tau^*_{\gamma^*}}\right) - \dfrac{1}{2} \left(\tau^*_{\gamma} \log(o_1)-x_1^{\top}\alpha_\gamma^* - s_1^{\top}\kappa_\gamma^*\right)^2 + \dfrac{1}{2} \left(\tau^*_{\gamma^*} \log(o_1)-x_1^{\top}\alpha^*_{\gamma^*}-s_1^{\top}\kappa^*_{\gamma^*}\right)^2 \right]$$
arises from the contribution of uncensored observations
and $a_2(\eta_\gamma^*, \eta_{\gamma^*}^*)=$
$${\mathbb E}_{F_0 \mid u_1=0} \left\{\log \dfrac{\Phi\left(x_1^{\top}\alpha_\gamma^* + s_1^{\top}\kappa_\gamma^* - \tau_\gamma^* \log(c_1)\right)}{\Phi\left(x_1^{\top}\alpha_{\gamma^*}^* + s_1^{\top}\kappa_{\gamma^*}^* - \tau_{\gamma^*}^* \log(c_1)\right)}  \right\}$$
from censored observations.
If there were no censoring, then $\tau^*_\gamma$ has closed-form that can be plugged into $M(\eta_\gamma^*)$
to show that $a_1(\eta_\gamma^*, \eta_{\gamma^*}^*)= {\mathbb E}_{F_0 \mid u_1=1}\left[ \log\left(\tau^*_{\gamma} / \tau^*_{\gamma^*}\right) \right]$,
\textit{i.e.}~the exponential Bayes factor rate is driven by the ratio of KL-optimal error variances between $\gamma$ and $\gamma^*$.
The role of $a_2$ is interesting in that, even under censoring and model misspecification,
the Bayes factor rate depends on a loss function that depends only on the KL-optimal errors
$\tau^*_\gamma \log(o_1)-x_1^{\top}\alpha^*_\gamma - s_1^{\top}\kappa^*_\gamma$
and $\tau^*_{\gamma^*} \log(c_1)-x_1^{\top}\alpha^*_{\gamma^*}-s_1^{\top}\kappa^*_{\gamma^*}$.

{
The behaviour of $T_2$ is as follows. From Part (i),
\begin{eqnarray*}
\frac{\pi_L(\tilde{\eta}_\gamma \mid \gamma)}{\pi_L(\tilde{\eta}_{\gamma^*} \mid \gamma^*)} &=& \lOp \left( g_L^{\frac{p_{\gamma^*} - p_\gamma}{2}} g_S^{\frac{r(s_{\gamma^*} - s_\gamma)}{2}} \right), \\
\frac{\pi_M(\tilde{\eta}_\gamma \mid \gamma)}{\pi_M(\tilde{\eta}_{\gamma^*} \mid \gamma^*)} &=& 
\lOp \left( g_S^{\frac{r(s_{\gamma^*} - s_\gamma)}{2}} \right) \times
g_M^{\frac{3}{2}(p_{\gamma^*} - p_\gamma)}
\frac{\prod_{j \in \gamma} \tilde{\alpha}_{\gamma j}^2}{\prod_{j \in \gamma^*}  \tilde{\alpha}_{\gamma^* j}^2} \\
&=& \lOp \left( g_S^{\frac{r(s_{\gamma^*} - s_\gamma)}{2}} g_M^{\frac{3(p_{\gamma^*}-p_\gamma)}{2}} \right),
\end{eqnarray*}
since $\tilde{\alpha}_{\gamma j}$ converges in probability to a finite constant $\geq 0$ and $\tilde{\alpha}_{\gamma^* j}$ to a strictly positive finite constant, for all $j$. Further,
\begin{align}
\frac{\pi_E(\tilde{\eta}_\gamma \mid \gamma)}{\pi_E(\tilde{\eta}_{\gamma^*} \mid \gamma^*)} &= 
\lOp \left( g_S^{\frac{r(s_{\gamma^*} - s_\gamma)}{2}} \right) \times
g_E^{\frac{p_{\gamma^*}- p_\gamma}{2}}
\exp \left\{ \sum_{j \in \gamma^*} \frac{g_E}{(\tilde{\alpha}_{\gamma^* j})^2}
- \sum_{j \in \gamma} \frac{g_E}{(\tilde{\alpha}_{\gamma j})^2} \right\}
\nonumber \\
&=\lOp \left( g_S^{\frac{r(s_{\gamma^*} - s_\gamma)}{2}}
g_E^{\frac{p_{\gamma^*}- p_\gamma}{2}} e^{g_E c} \right)
\nonumber
\end{align}
since $\sum_{j \in \gamma^*} 1/(\tilde{\alpha}_{\gamma^* j})^2$ converges in probability to a finite constant and
$\sum_{j \in \gamma} 1/(\tilde{\alpha}_{\gamma j})^2$ either converge to infinity 
(if $\alpha^*_{\gamma j}=0$ for any $j$, and then the term is $e^{-g_E \sqrt{n}}$, as in Part (i))
or to a finite non-zero constant (if $\alpha^*_{\gamma j} \neq 0$ for all $j$).
In the latter case one obtains the worst-case rate $e^{g_E c}$ where
$$
c= \sum_{j \in \gamma^*} \frac{1}{(\alpha^*_{\gamma^* j})^2} - \sum_{j \in \gamma} \frac{1}{(\alpha^*_{\gamma j})^2} \in \mathbb{R}.
$$
}

To summarize
\begin{align}
\log(B_{\gamma \gamma^*}) &=
T_1 + \log(T_2) + \log(T_3) + \mathcal{O}(1)
\nonumber \\
&= n [M(\eta^*_\gamma) - M(\eta^*_{\gamma^*})] {[1 + \lo_p(1)]} + \log(T_2) + \frac{d_{\gamma^*} - d_\gamma}{2} \log(n) =
\nonumber \\
&= n[M(\eta^*_\gamma) - M(\eta^*_{\gamma^*})] {[1 + \lo_p(1)]}
 + \log(b_n) + \frac{r}{2}(s_{\gamma^*} - s_\gamma) \log(n g_S) + \lOp(1),
\end{align}
{
where 
\begin{align}
b_n=
 \begin{cases}
(n g_L)^{\frac{(p_{\gamma^*} - p_\gamma)}{2}} \mbox{, for } \pi_L \\
(n g_M^3)^{\frac{(p_{\gamma^*} - p_\gamma)}{2}} \mbox{, for } \pi_M \\
(n g_E)^{\frac{(p_{\gamma^*} - p_\gamma)}{2}} e^{-g_E c} \mbox{, for } \pi_E  \mbox{ and  finite } c \in \mathbb{R}
\end{cases}
\nonumber
\end{align}
}
as we wished to prove.

\end{enumerate}

\clearpage
\subsection{Proof of Proposition \ref{prop:BFRatesCox}} \label{app:CoxProportional Hazards}
Denote by $\eta_\gamma^{\top} = (\beta_{\gamma}^\top,\delta_{\gamma}^{\top})$ the whole parameter vector and by $v_{\gamma i}^{\top} = (x_{\gamma i}^{\top},s_{\gamma i}^{\top}) $ the covariates selected by model $\gamma$. 
Recall that the data-generating truth $F_0$ is allowed to depend on a larger covariate set $z_i$, that is $(x_i,s_i) \subseteq z_i \in {\mathbb R}^{p(r+1) + q}$, and that \color{black} $\eta_\gamma\in {\mathbb R}^{p_{\gamma} + r s_{\gamma} }$ \color{black} is estimated using the log partial likelihood \citep{cox:1972}
\begin{equation*}
\ell_p(\eta_\gamma) = \sum_{u_i=1} \left(v_{\gamma i}^{\top} \eta_\gamma \right) - \sum_{u_i=1} \log\left( \sum_{k \in {\mathcal R}(o_i)} \exp\left\{ v_{\gamma k}^{\top} \eta_\gamma \right\} \right),
\end{equation*}
where ${\mathcal R}(t) = \{i : o_i \geq t\}$ denotes the risk set at time $t$.





%

The proof follows analogously to that of Proposition \ref{prop:BFRatesI}.
The outline of the proof is as follows. 
Bayes factors $B_{\gamma \gamma^*}$ are obtained by Laplace approximations to the integrated partial likelihoods \eqref{eq:integrated_ploglik} under models $\gamma$ and $\gamma^*$, that is by replacing the AFT log-likelihood by the Cox partial likelihood in \eqref{eq:BFLap}.
The rates of the determinants in \eqref{eq:ratiodet} remain as in Proposition \ref{prop:BFRatesI},
given that  $\tilde{\eta}_\gamma \stackrel{P}{\longrightarrow} \eta_\gamma^*$
under Conditions B1--B4 below, for an optimal parameter value \color{black} $\eta_\gamma^* \in {\mathbb R}^{p_{\gamma} + r s_{\gamma} }$ \color{black} also defined below, and the asymptotic Hessians at $\eta_\gamma^*$ being strictly negative definite \citep{wong:1986,lin:1989}.

We first discuss $\eta_\gamma^*$, subsequently the sufficient technical conditions B1-B4 for the result to hold, and finally we present the proof for the cases where $\gamma$ is an overfitted and a non-overfitted model.
When one considers model misspecification, the optimal parameter value $\eta_{\gamma^*}$ is that minimizing $M_p(\eta_\gamma)$, the expected log partial likelihood under the data-generating $F_0$, given by \citep{struthers:1986}
\begin{eqnarray}\label{eq:eploglik}
M_p(\eta_\gamma) &=&  \operatorname{plim}_{n\to\infty} \dfrac{1}{n}\ell_p(\eta_{\gamma})\nonumber\\
&=& \eta_\gamma^{\top} \int s^{(1)}(t)dt - \int \log\left[ s^{(0)}(\eta_\gamma,t) \right]  s^{(0)}(t) dt,
\end{eqnarray}
where $\operatorname{plim}$ denotes the limit in probability, and we define the vector functions
\begin{eqnarray*}
s^{(k)}(t) &=& E_{F_0}\left[S^{(k)}(t) \right], \\
S^{(k)}(t) &=& \dfrac{1}{n} \sum_{i=1}^{n}  \text{I}(t_i \geq t) \lambda_i(t) v_{\gamma i}^{\otimes k},\\
s^{(k)}(\eta_\gamma,t) &=& E_{F_0}\left[S^{(k)}(\eta_\gamma,t) \right], \\
S^{(k)}(\eta_\gamma,t) &=& \dfrac{1}{n} \sum_{i=1}^{n}  \text{I}(t_i \geq t) \exp\left[v_{\gamma i}^{\top}\eta_{\gamma} \right] v_{\gamma i}^{\otimes k},
\end{eqnarray*}
for $k=0,1,2$, where $\lambda_i(t) =   f_0(t\mid z_{i})/[1-F_0(t \mid z_{i})]$ is the data-generating hazard function for individual $i$ at time $t$, $f_0$ is the (Radon-Nikodym) density associated to $F_0$, 
$\text{I}()$ is the indicator function,
$a^{\otimes 2}$ refers to the matrix $aa^{\top}$, $a^{\otimes 1}$ refers to the vector $a$, and $a^{\otimes 0}$ refers to the scalar $1$, and the expectations are taken with respect to $F_0$, the data-generating distribution of $(t_i,u_i,z_{ i})$. 

We re-write the limit function in \eqref{eq:eploglik} to facilitate its interpretation, see \cite{wong:1986} for further discussion. 
Since $(t_i,u_i,z_{i}) \sim F_0$ are identically distributed, and using that $s^{(0)}(t)= E_{F_0}[\text{I}(t_1 > t) \lambda_1(t)]$, $s^{(1)}(t)= E_{F_0}[I(t_1>t) \lambda_1(t) v_{\gamma 1}]$ and $s^{(0)}(\eta_\gamma,t)= E_{F_0}[\text{I}(t_1>t) e^{v_{\gamma 1}^\top \eta_{\gamma}}]$ we obtain
\begin{align}\label{eq:eploglik_rearrange}
&M_p(\eta_\gamma)= \int \log \left( \frac{e^{ E_{F_0}\left[\lambda_1(t)\eta_\gamma^\top v_{\gamma 1} \text{I}(t_1>t) \right]} }{\left[ E_{F_0} \left(e^{\eta_\gamma^\top v_{\gamma 1}} \text{I}(t_1>t) \right)  \right]^{E_{F_0}[\text{I}(t_1 > t) \lambda_1(t)]}} \right) dt.
\end{align}
Now, using that the marginal data-generating density for the survival times $t_i$ can be written as $f_0(t) = \lambda_0(t) P_{F_0}(t_1>t)$, where $\lambda_0()$ is the associated marginal hazard, we obtain
\begin{eqnarray*}
M_p(\eta_\gamma) &=& \int \dfrac{f_0(t)}{f_0(t)} \log \left( \frac{e^{ E_{F_0}\left[\lambda_1(t)\eta_\gamma^\top v_{\gamma 1} \text{I}(t_1>t) \right]} }{\left[ E_{F_0} \left(e^{\eta_\gamma^\top v_{\gamma 1}} \text{I}(t_1>t) \right)  \right]^{E_{F_0}[\text{I}(t_1 > t) \lambda_1(t)]}} \right) dt\\
&=&\int \log \left( \frac{e^{ E_{F_0}\left[\eta_\gamma^\top v_{\gamma 1} \lambda_1(t) \mid t_1>t \right]/\lambda_0(t)} }{\left[ E_{F_0} \left(e^{\eta_\gamma^\top v_{\gamma 1}} \text{I}(t_1>t) \right)  \right]^{E_{F_0}[ \lambda_1(t) \mid t_1 > t]/\lambda_0(t)}} \right) dF_0(t),
\end{eqnarray*}
where in the denominator we used that $E_{F_0}( \lambda_1(t) \text{I}(t_1 > t))= E_{F_0}(\lambda_1(t) \mid t_1>t) P_{F_0}(t_1 >t)$, and similarly for the numerator.

Therefore, $M_p(\eta_\gamma)$ can be seen as an average log-reward for correct predictions at time $t$, averaged with respect to the data-generating $F_0(t)$. The reward is for assigning a large expected hazard $\eta_\gamma^{\top} v_{\gamma 1}$ to the individual, conditional on being at risk at time $t$, relative to the hazard assigned to individuals at risk (with $t_1 > t$).


We next outline the technical conditions, consistency and asymptotic normality of the maximum partial likelihood estimator $\bar{\eta}_\gamma$.
Define the matrices
\begin{eqnarray*}
\bar{A}(\eta_\gamma) &=& -\dfrac{1}{n}\nabla_{\eta_\gamma}^2 \ell_p(\eta_{\gamma}),\\
A(\eta_\gamma) &=& \int_{0}^{\infty} \left\{  \dfrac{s^{(2)}(\eta_\gamma,t)}{s^{(0)}(\eta_\gamma,t)} - \dfrac{s^{(1)}(\eta_\gamma,t)^{\otimes 2}}{s^{(0)}(\eta_\gamma,t)^{2}}\right\} s^{(0)}(t)dt,\\
\bar{V}(\eta_\gamma) &=& \bar{A}^{-1}(\eta_\gamma) \left( \dfrac{1}{n}\sum_{i=1}^{n}W_i(\eta_\gamma)^{\otimes 2} \right) \bar{A}^{-1}(\eta_\gamma) \\
W_i(\eta_\gamma) &=&  u_i \left\{ v_{\gamma i} - \dfrac{S^{(1)}(\eta_\gamma,t_i)}{S^{(0)}(\eta_\gamma,t_i)} \right\} \\
&-& \sum_{j=1}^{n} \dfrac{u_j \text{I}(t_i \geq t_j) \exp(v_{\gamma i}^{\top}\eta_{\gamma})}{n S^{(0)}(\eta_\gamma,t_j)} \times \left\{ v_{\gamma i} - \dfrac{S^{(1)}(\eta_\gamma,t_j)}{S^{(0)}(\eta_\gamma,t_j)} \right\}.\\
\end{eqnarray*}
Here $\bar{A}(\eta_\gamma)$ is the observed hessian matrix, $A(\eta_\gamma)$ the expected hessian under $F_0$, 
and we denote by $V_{\eta_{\gamma}^*}$ the limit in probability of $\bar{V}(\bar{\eta}_\gamma)$ as $n\to \infty$, which gives the asymptotic covariance of $\bar{\eta}_\gamma$ \citep{lin:1989}.

%

Consider now the following technical conditions
\begin{enumerate}[label=\bfseries B\arabic*.]


\item $(t_i, u_i, z_i)$ for $i=1,\dots,n$ are independent draws from $F_0$, and $z_i$ has bounded support.

\item There exists a neighbourhood ${\mathcal B}_{\gamma}$ of $\eta_{\gamma}^*$ such that for each $t<\infty$
\begin{eqnarray*}
\sup_{r\in[0,t], \eta_{\gamma}\in{\mathcal B}_{\gamma}} \vert \vert S^{(0)}(\eta_\gamma,r) - s^{(0)}(\eta_\gamma,r) \vert \vert \stackrel{P}{\to} 0,
\end{eqnarray*}
as $n\to \infty$, the entries of $s^{(0)}(\eta_\gamma,r)$ are bounded away from zero on ${\mathcal B}_{\gamma}\times[0,t]$, and the entries of $s^{(0)}(\eta_\gamma,r)$ and $s^{(1)}(\eta_\gamma,r)$ are bounded on ${\mathcal B}_{\gamma}\times[0,t]$.
\item For each $t<\infty$,
\begin{eqnarray*}
\int_0^t s^{(2)}(r)dr <\infty,
\end{eqnarray*}
entrywise.

\item The matrix $A(\eta_{\gamma}^*)$ is positive definite.
\end{enumerate}
Under Conditions B1--B4, \cite{lin:1989} showed that the maximum partial likelihood estimators (PMLEs), $\bar{\eta}_{\gamma}$, are consistent,
$$\bar{\eta}_{\gamma} \stackrel{P}{\to} \eta_{\gamma}^*,$$
and asymptotically normal.
$$ \sqrt{n} \left( \bar{\eta}_{\gamma}- \eta_{\gamma}^* \right) \stackrel{D}{\to} N\left(0, V_{\eta_{\gamma}^*}\right).$$

The prior structure adopted for the Cox model is the same as that used in AFT models (except for the scale parameter, which does not appear in the Cox model). Thus, the first part of the proof of Proposition \ref{prop:BFRatesI} also applies to the Cox model. Now, to characterize the two remaining terms in \eqref{eq:BFLap}, for the Cox model, we again consider separately the case where $\gamma^*\not\subset \gamma$ and the case where $\gamma^*\subset \gamma$.
\begin{enumerate}[leftmargin=*,label=(\roman*)]


\item {\bf Case $\gamma^*\subset \gamma$}. In this case, we have $M_p(\eta_\gamma^*) - M_p(\eta_{\gamma^*}^*)=0$.

From \cite{RT17} and the consistency of the partial maximum likelihood estimators, we get that  $\ell_p(\tilde{\eta}_\gamma) -  \ell_p(\tilde{\eta}_{\gamma^*}) = \ell_p(\bar{\eta}_\gamma) -  \ell_p(\bar{\eta}_{\gamma^*}) + \lo_p(1)$. Let $U= A({\eta_\gamma^*})$,  $W_{\gamma^*}= V_{\eta_\gamma^*}$. From the previous arguments we have that
\begin{align}
  \sqrt{n} (\bar{\eta}_\gamma - \eta_\gamma^*) \stackrel{D}{\longrightarrow} N \left( 0,  V_{\gamma^*} \right)
  \Rightarrow
 z \stackrel{D}{\longrightarrow} N \left( 0,  I \right),
    \nonumber
\end{align}
where $z= \sqrt{n} V_{\gamma^*}^{-1/2} (\bar{\eta}_\gamma - \eta_\gamma^*)$. Now, note that the Hessian matrix of the log partial likelihood converges to a non-singular matrix under Conditions B1--B4 \citep{lin:1989}, and that $\sqrt{n} \left( \bar{\eta}_\gamma -  \bar{\eta}_{\gamma^*}\right)= \lOp(1)$ (with respect to the Euclidean norm). Then, we can expand the log partial likelihood ratio as
(see Chapter 25 of \citealp{vandervaart:1998} and \citealp{fan:2009})
\begin{eqnarray*}
2 [ \ell_p(\bar{\eta}_\gamma) -  \ell_p(\bar{\eta}_{\gamma^*}) ]
&=&  z^\top W_{\gamma \setminus \gamma^*}^{1/2} U^{-1}_{\gamma \setminus \gamma^*} W_{\gamma \setminus \gamma^*}^{1/2} z +  \lo_p (1),
\end{eqnarray*}
If the model is correctly specified, then $W_{\gamma \setminus \gamma^*}^{1/2} U^{-1}_{\gamma \setminus \gamma^*} W_{\gamma \setminus \gamma^*}^{1/2}=I$, obtaining $z^\top z \stackrel{D}{\longrightarrow} \chi^2_{d_\gamma - d_{\gamma^*}}$. Thus, we have that the partial likelihood ratio test has the same asymptotic distribution as the Wald test statistic.
If the model is misspecified, then the right-hand side converges in distribution to
a linear combination of independent chi-square random variables (see also \citealp{bickel:2007}) and can again be upper bounded by
$$z^\top W_{\gamma \setminus \gamma^*}^{1/2} U^{-1}_{\gamma \setminus \gamma^*} W_{\gamma \setminus \gamma^*}^{1/2} z \leq \lambda\left( W_{\gamma^*}^{1/2} U^{-1}_{\gamma^*} W_{\gamma^*}^{1/2} \right) z^\top z$$
where $\lambda()$ denotes the largest eigenvalue (and is a bounded constant under our assumptions)
and $z^\top z \stackrel{D}{\longrightarrow} \chi^2_{d_\gamma - d_{\gamma^*}}$.
That is, 
$$2[\ell_p(\tilde{\eta}_\gamma) -  \ell_p(\tilde{\eta}_{\gamma^*})] \stackrel{D}{\longrightarrow} \lambda\left( W_{\gamma^*}^{1/2} U^{-1}_{\gamma^*} W_{\gamma^*}^{1/2} \right) \chi^2_{d_\gamma - d_{\gamma^*}} = \lOp(1).$$
Consequently, the characterization of the second term in \eqref{eq:BFLap} remains valid for local and non-local priors, and implies the same Bayes factors rates.


\item {\bf Case $\gamma^*\not\subset \gamma$}.
By consistency of the PMLE, the law of large numbers and uniform convergence of the log partial likelihood to its expectation $M_p(\eta)$ it follows that
\begin{eqnarray*}
\dfrac{1}{n}\left[ \ell_p(\tilde{\eta}_\gamma) -  \ell_p(\tilde{\eta}_{\gamma^*}) \right] \stackrel{P}{\rightarrow} M_p({\eta}_\gamma^*) - M_p({\eta}_{\gamma^*}^*)<0,
\end{eqnarray*}

The behaviour of the ratio of prior densities remains the same as in the proof of Proposition \ref{prop:BFRatesI}. Thus,
\begin{eqnarray*}
\log(B_{\gamma \gamma^*})= n[M_p(\eta^*_\gamma) - M_p(\eta^*_{\gamma^*})] + \log(b_n) + \frac{r}{2}(s_{\gamma^*} - s_\gamma) \log(n g_S) + \lOp(1),
\nonumber
\end{eqnarray*}
as we wished to prove.

\end{enumerate}

\clearpage
\section{AFT model with Laplace errors} \label{app:AFTLaplace}

In this section we derive asymptotic properties of the MLE under an assumed AFT model with Laplace errors.
Recall that the Laplace pdf and cdf are given, respectively, by
\begin{eqnarray*}
f(x\mid \mu,\sigma) &=& \dfrac{1}{2\sigma}\exp\left(-\dfrac{\vert x-\mu\vert}{\sigma}\right),\\
F(x\mid \mu,\sigma) &=& \dfrac{1}{2}+ \dfrac{1}{2}\operatorname{sign}(x-\mu)\left(1-\exp\left(-\dfrac{\vert x-\mu\vert}{\sigma}\right)\right).
\end{eqnarray*}
Let us denote $f(t) = f(t \mid 0,1)$ and $F(t) = F(t\mid 0,1)$ for simplicity.
Consider the AFT model with Laplace distributed errors,
the log-likelihood function is $\log p_L(y \mid \beta_\gamma,\delta_\gamma,\sigma)=$
\begin{equation}\label{eq:logLikeAFTLap}
\ell(\beta_\gamma,\delta_\gamma,\sigma) =
-{n_o} \log(2\sigma) - \dfrac{1}{\sigma} \sum_{u_i=1} \left\vert y_i-x_i^{\top}\beta_\gamma-s_i^{\top} \delta_\gamma\right\vert
 + \sum_{u_i=0} \log \left\{F\left(\dfrac{x_i^{\top}\beta_\gamma+s_i^{\top} \delta_\gamma-y_i}{\sigma}\right)\right\}.
\end{equation}
This model was used for instance in \cite{bottai:2010}.
Consider the reparametrization $\alpha_\gamma = \beta_\gamma/\sigma$, $\kappa_\gamma= \delta_\gamma/\sigma$, and $\tau = 1/\sigma$, the corresponding log-likelihood is
\begin{equation}\label{eq:logLikeAFTLap2}
\ell(\alpha_\gamma,\kappa_\gamma,\tau) = {n_o}\log\left(\frac{\tau}{2}\right) -  \sum_{u_i=1}\left \vert \tau y_i-x_i^{\top}\alpha_\gamma -s_i^{\top}\kappa_\gamma\right\vert + \sum_{u_i=0} \log \left\{F\left(x_i^{\top}\alpha_\gamma +s_i^{\top}\kappa_\gamma - \tau y_i\right)\right\}.
\end{equation}
The term of the log-likelihood function associated to the observed times is strictly concave, provided that the corresponding design matrix $X_o$ has full column rank \citep{rossell:2018}. Moreover, given that $\log F$ is log-concave \citep{bagnoli:2005} it follows that the the term associated to the censored observations is a sum of concave functions and hence is also concave (see also \citealp{bottai:2010}). Consequently, the log-likelihood function is concave. This property facilitates deriving asymptotic results, as presented below.

Simple algebra shows that the gradient of \eqref{eq:logLikeAFTLap2} is
\begin{eqnarray*}
\nabla_{(\alpha_{\gamma},\kappa_{\gamma}) } \ell(\alpha_\gamma,\kappa_\gamma,\tau) &=&  \sum_{u_i=1}  \begin{pmatrix} x_i \\ s_i \end{pmatrix}\omega_i +  \sum_{u_i=0}\begin{pmatrix} x_i \\ s_i \end{pmatrix} \tilde{r}\left(x_i^{\top}\alpha_\gamma +s_i^{\top}\kappa_\gamma - \tau y_i\right)\\
\nabla_{\tau } \ell(\alpha_\gamma,\kappa_\gamma,\tau) &=& \dfrac{n_o}{\tau} - \sum_{u_i=1} \omega_i y_i - \sum_{u_i=0}
 \tilde{r}\left(x_i^{\top}\alpha_\gamma +s_i^{\top}\kappa_\gamma - \tau y_i\right)  y_i,
\end{eqnarray*}
where
\begin{eqnarray*}
\omega_i = \begin{cases}
-1 & \mbox{if   } \,\,\, \tau y_i < x_i^{\top}\alpha_\gamma +s_i^{\top}\kappa_\gamma,  \\
1 & \mbox{if   } \,\,\, \tau y_i > x_i^{\top}\alpha_\gamma +s_i^{\top}\kappa_\gamma,
\end{cases}
\end{eqnarray*}
and
\begin{eqnarray*}
\tilde{r}\left(z \right) =
\begin{cases}
1, & \mbox{if } z \leq 0,\\
\dfrac{f(z)}{F(z)}, & \mbox{if } z> 0,
\end{cases}
\end{eqnarray*}
except on the zero-Lebesgue measure set $\Delta = \bigcup_{i=1}^n\{(\alpha_\gamma,\kappa_\gamma,\tau):\tau y_i = x_i^{\top}\alpha_\gamma +s_i^{\top}\kappa_\gamma \}$, where the gradient is not defined.

The Hessian matrix $H(\alpha_\gamma,\kappa_\gamma,\tau)$ only depends on the censored observations, and  is given by
\begin{eqnarray*}
  \nabla^2_{(\alpha_\gamma,\kappa_\gamma)} \ell(\alpha_\gamma,\kappa_\gamma,\tau) &=&
- \sum_{u_i=0} \begin{pmatrix} x_i \\ s_i \end{pmatrix} \begin{pmatrix} x_i \\ s_i \end{pmatrix}^{\top} \tilde{D}\left(\tau y_i - x_i^{\top}\alpha_\gamma - s_i^{\top} \kappa_\gamma \right),\\
  \nabla_{\tau}\nabla_{(\alpha_\gamma,\kappa_\gamma)} \ell(\alpha_\gamma,\kappa_\gamma,\tau) &=&  \sum_{u_i=0} \begin{pmatrix} x_i \\ s_i \end{pmatrix} y_i \tilde{D}\left(\tau y_i - x_i^{\top}\alpha_\gamma - s_i^{\top} \kappa_\gamma \right),\\
  \nabla^2_{\tau}\ell(\alpha_\gamma,\kappa_\gamma,\tau)  &=&  -\dfrac{n_o}{\tau^2}  - \sum_{u_i=0} y_i^2 \tilde{D}\left(\tau y_i - x_i^{\top}\alpha_\gamma - s_i^{\top} \kappa_\gamma \right).
\end{eqnarray*}
where
\begin{eqnarray*}
\tilde{D}\left(z \right) = - r'\left(-z \right) =
\begin{cases}
\tilde{r}\left(-z \right)^2 + \tilde{r}\left(-z \right), & \mbox{if } z<0,\\
0, & \mbox{if } z > 0,\\
\text{Undefined}, & \mbox{if } z = 0,\\
\end{cases}
\end{eqnarray*}
 except on the zero-Lebesgue measure set $\Delta$, where the Hessian is not defined.

To prove consistency and asymptotic MLE normality we consider Conditions A1-A2 and A4-A5 in the main paper, and we replace condition A3 by condition A3* below. Throughout, we denote $\eta_\gamma = (\alpha_\gamma,\kappa_\gamma,\tau)$.
\begin{enumerate}[label=\bfseries A\arabic**.]
\setcounter{enumi}{2}

\item Let
\begin{eqnarray*}
 m(\eta_\gamma) &=& u_1\left[ \log(\tau/2) -   \left\vert\tau \log(o_1)-x_1^{\top}\alpha_\gamma -s_1^{\top} \kappa_\gamma \right\vert  \right]
 \nonumber \\
& + & (1-u_1)
 \left[ \log F\left(x_1^{\top}\alpha_\gamma + s_1^{\top} \kappa_\gamma - \tau \log(c_1)\right) \right].
\end{eqnarray*}
 \color{black} There exists a maximum $\eta_\gamma^* \in \Gamma_\gamma$ of $M(\eta_\gamma)=E_{F_0}[m(\eta_\gamma)]$ such that $M(\eta_\gamma^*) < \infty$,
and $E_{F_0}[|m(\eta_\gamma)|] < \infty$ for all $\eta_\gamma \in \Gamma_\gamma$, where the expectation is taken with respect to the data-generating process $F_0$.
Further, there exists a maximum $\widehat{\eta}_\gamma \in \Gamma_\gamma$ of $M_n(\eta_\gamma)$ with probability 1, as $n \rightarrow \infty$. \color{black}
\end{enumerate}

\begin{proposition}\label{prop:Consistency_lap}
Assume Conditions A1-A2 and A3*.
Then, $M(\eta_\gamma)$ has a unique maximizer $\eta_\gamma^* = \operatorname{argmax}_{\Gamma_\gamma} M(\eta_\gamma)$.
Moreover, $\widehat{\eta}_\gamma \stackrel{P}{\rightarrow} \eta_\gamma^*$ as $n\rightarrow \infty$.

\proof

The proof is based on showing that under Conditions A1-A2 and A3*, the assumptions in the consistency Theorem 5.7 from \cite{vandervaart:1998} are satisfied. This requires appealing to the concavity properties of the log likelihood function discussed above.

Let $M_n(\eta_\gamma) = n^{-1}\ell(\eta_\gamma)$ be the average log-likelihood evaluated at $\eta_\gamma$. 
The assumption that data are generated \emph{i.i.d.} from $F_0$ \color{black} and that ${\mathbb E}_{F_0}(|m(\eta_\gamma)|)<\infty$ give that, \color{black}
by the law of large numbers, $M_n(\eta_\gamma) \stackrel{P}{\rightarrow} M(\eta_\gamma)$, for each $\eta_\gamma\in\Gamma_\gamma$. The expectation of $M_n(\eta_\gamma)$ under the data-generating $F_0$ is $M(\eta_\gamma) = $
\begin{eqnarray*}
&&  P_{F_0}(u_1=1) {\mathbb E}_{F_0} \left[ \log(\tau) -\log(2) -  \left\vert \tau \log(o_1)-x_1^{\top}\alpha_\gamma -s_1^{\top} \kappa_\gamma \right\vert  \mid u_1=1 \right]
 \nonumber \\
& + & P_{F_0}(u_1=0)
 {\mathbb E}_{F_0} \left[ \log F\left(x_1^{\top}\alpha_\gamma + s_1^{\top} \kappa_\gamma - \tau \log(c_1)\right) \mid u_1=0 \right].
\end{eqnarray*}
The proof then follows analogously to that of Proposition \ref{prop:Consistency}.
\color{black} Briefly, the convexity lemma in \cite{pollard:1991} gives uniform convergence of $M_n(\eta_\gamma)$ to $M(\eta_\gamma)$ and concavity of $M(\eta_\gamma)$ which, together with A3*, guarantee that $\eta_\gamma^*$ and $\eta_\gamma$ are unique and lie in a compact set $K \subset \Gamma_\gamma$. \color{black}
This provides the conditions to apply Theorem 5.7 in \cite{vandervaart:1998}, which gives that $\widehat{\eta}_\gamma \stackrel{P}{\longrightarrow} \eta_\gamma^*$.
\end{proposition}

\begin{proposition}\label{prop:AsympNorm_lap}
Assume Conditions A1--A2, A3*, and A4-A5.
Then $\sqrt{n}(\widehat{\eta}_\gamma-\eta_\gamma^*) \stackrel{D}{\longrightarrow} N\left(0, V_{\eta_\gamma^*}^{-1} {\mathbb E}_{F_0}[ \nabla m(\eta_\gamma^*) \nabla m(\eta_\gamma^*)^{\top}] V_{\eta_\gamma^*}^{-1}\right)$,
where $V_{\eta_\gamma^*}$ is the Hessian of $M(\eta_{\gamma})$ evaluated at $\eta_\gamma^*$, and $m(\eta_\gamma^*) = \log p_L(y_1 \mid \eta_\gamma^*)$.

\proof

The proof is analogous to the proofs of Proposition 4 in \cite{rossell:2018} and Proposition \ref{prop:AsympNorm} in the main paper.
Recall that Proposition \ref{prop:Consistency_lap} implies the existence and uniqueness of $\eta_\gamma^*$. Let us define
\begin{eqnarray*}
m_{\eta_\gamma}(y_1,z_1) \equiv m(\eta_\gamma) &=& u_1\left[  \log(\tau) -\log(2) -  \left\vert \tau \log(o_1)-x_1^{\top}\alpha_\gamma -s_1^{\top} \kappa_\gamma \right\vert  \right]
 \nonumber \\
& + & (1-u_1)
 \left[ \log F\left(x_1^{\top}\alpha_\gamma + s_1^{\top} \kappa_\gamma - \tau \log(c_1)\right) \right]
\end{eqnarray*}
The gradient of $m_{\eta}(y_1,z_1)$ is bounded for almost all $\eta_{\gamma}\in\Gamma_\gamma$ and $(y_1,z_1)$, due to the compactness of $\Gamma_\gamma$. Following  the proof of Proposition \ref{prop:AsympNorm}, we obtain that $\vert \vert \nabla m_{\eta}(y_1,z_1) \vert \vert$ is upper bounded almost surely by sum of the norm of its entries. Conditions A1--A2, A3*, and A4 guarantee that
\begin{eqnarray*}
\int K(y_1,z_1)^2 dF_0(o_1,z_1,c_1)   < \infty,
\end{eqnarray*}
where $K(y_1,z_1)=\sup_{\eta\in{\mathcal B}_{\eta_\gamma^*}} \vert \vert \nabla m_{\eta}(y_1,z_1) \vert \vert$, ${\mathcal B}_{\eta_\gamma^*}\subset \Gamma$ is any neighborhood of $\eta_\gamma^*$, whose projection over $\theta_\gamma$ coincides with ${\mathcal B}_{(\alpha_\gamma^*,\kappa_{\gamma^*)}}$. This result, together with the mean value theorem and the Cauchy-Schwarz inequality, implies that for $\eta_1, \eta_2\in{\mathcal B}_{\eta_\gamma^*}$, with probability 1,
\begin{eqnarray*}
\vert m_{\eta_1}(y_1,z_1) - m_{\eta_2}(y_1,z_1) \vert &\leq& K(y_1,z_1) \cdot \vert \vert \eta_1-\eta_2 \vert \vert.
\end{eqnarray*}
Consider now the conditional expectation
\begin{eqnarray*}
{\mathbb E}_{F_0}[m_{\eta_\gamma}\mid z_1] &=& - \log(2)P_{F_0}(u_1=1 \mid z_1) + \log(\tau)P_{F_0}(u_1=1 \mid z_1)\\
 &-&P_{F_0}(u_1=1 \mid z_1) {\mathbb E}_{F_0}\left[  \left\vert\tau \log(o_1)-x_1^{\top}\alpha_\gamma - s_1^{\top} \kappa_\gamma\right\vert  \mid z_1, u_1=1 \right] \\
&+& P_{F_0}(u_1=0 \mid z_1){\mathbb E}_{F_0}\left\{\log F\left(x_1^{\top}\alpha_\gamma + s_1^{\top} \kappa_\gamma - \tau \log(c_1)\right) \mid z_1, u_1=0\right\}.
\end{eqnarray*}
\cite{rossell:2018} showed that the term in this expected log-likelihood associated to the observed times is twice differentiable with respect to the parameters. The term associated to the censored observations is positive definite by A5. Thus, we have that ${\mathbb E}_{F_0}[m_{\eta_\gamma}] $ is concave and, consequently, its Hessian is non-singular at ${\eta^*_{\gamma}}$. Thus, the asymptotic normality follows by Theorem 5.23 from \cite{vandervaart:1998} together with the consistency results in Proposition \ref{prop:Consistency_lap}.
\end{proposition}

\clearpage

\section{Results for the probit model}\label{app:Probit}

Let $\omega_i \in \{0,1\}$ be a binary outcome, and $\omega=(\omega_1,\ldots,\omega_n)$.
Analogously to Section \ref{sec:theory}, we assume that truly $(\omega_i,z_i) \sim F_0$ arise \textit{i.i.d.} from a data-generating $F_0$,
where $z_i \in {\mathbb R}^{p(r+1) + q}$ contains observed and potentially unobserved covariates.

Probit regression assumes that $P(\omega_i=1 \mid x_i,s_i)= \Phi(x_i^\top \alpha + s_i^\top \kappa)$ and has log-likelihood
\begin{align}
\tilde{\ell}(\alpha, \kappa)=
\sum_{i=1}^n \log p(\omega_i \mid \alpha,\kappa)=
\sum_{\omega_i=1} \log \Phi(x_i^\top \alpha + s_i^\top \kappa)
+ \sum_{\omega_i=0} \log \left( 1 - \Phi(x_i^\top \alpha + s_i^\top \kappa) \right).
\label{eq:probit}
\end{align}
Before stating our results, we remark that all our algorithms in Section \ref{sec:computation} apply directly to \eqref{eq:probit}
by exploiting that it is a particular case of the AFT model \eqref{logLikeAFT2}.
Specifically, let $\tilde{y}=(0,\ldots,0)$, $\tilde{d}=(0,\ldots,0)$ and
$(\tilde{x}_i,\tilde{s}_i)= (x_i,s_i) [\mbox{I}(\omega_i=1) - \mbox{I}(\omega_i=0)]$.
Then, an AFT model regressing $(\tilde{y},\tilde{d})$ on $(\widetilde{X},\widetilde{S})$ has log-likelihood
\begin{eqnarray*}
\ell( \alpha, \kappa, 1) =
\sum_{\omega_i=1} \log \left\{\Phi\left(x_i^{\top}\alpha + s_i^{\top} \kappa \right)\right\}
 + \sum_{\omega_i=0} \log \left\{\Phi\left(-x_i^{\top}\alpha - s_i^{\top} \kappa \right)\right\}
= \tilde{\ell}(\alpha, \kappa).
\end{eqnarray*}
This connection was first noted (to our knowledge) by \cite{doksum:1990} and,
besides its theoretical interest, it opens the potential to speed up probit regression via the approximations to $\log \Phi$ and its derivatives described in Section \ref{app:withinmodel_comput}.

Let $\eta_\gamma = (\alpha_\gamma^{\top}, \kappa_\gamma^{\top})^{\top}$ be the whole parameter vector,
$\widetilde{m}(\eta_\gamma)= \log p(\omega_1 \mid \eta_\gamma)$ the contribution of the first individual to the log-likelihood,
and
\begin{eqnarray*}
\widetilde{M}(\eta_\gamma)&=& {\mathbb E}_{F_0} \left( \widetilde{m}(\eta_\gamma)  \right)=
P_{F_0}(\omega_1=1) {\mathbb E}_{F_0} \left[ \log \Phi(x_{i\gamma}^\top \alpha_\gamma + s_{i\gamma}^\top \kappa_\gamma) \mid \omega_1=1 \right]\\
&+& P_{F_0}(\omega_1=0) {\mathbb E}_{F_0} \left[ \log \left( 1 - \Phi(x_{i\gamma}^\top \alpha_\gamma + s_{i\gamma}^\top \kappa_\gamma) \right) \mid \omega_1=0 \right]
\end{eqnarray*}
the expected average log-likelihood. 
The gradient and hessian of the log-likelihood are
\begin{eqnarray*}
\nabla_{(\beta_\gamma, \delta_\gamma)}\widetilde{\ell}( \beta_\gamma, \delta_\gamma)=
\sum_{\omega_i=1} r\left(x_i^{\top}\beta_\gamma + s_i^{\top} \delta_\gamma \right)  \begin{pmatrix} x_i \\ s_i \end{pmatrix}
 - \sum_{\omega_i=0} r\left(-x_i^{\top}\beta_\gamma - s_i^{\top} \delta_\gamma \right)\begin{pmatrix} x_i \\ s_i \end{pmatrix}.
\end{eqnarray*}
\begin{eqnarray*}
\nabla_{(\beta_\gamma, \delta_\gamma)}^2 \widetilde{\ell}(\beta_\gamma, \delta_\gamma) &=&
-\sum_{\omega_i=1} D\left(-x_i^{\top}\beta_\gamma - s_i^{\top} \delta_\gamma \right)  \begin{pmatrix} x_i \\ s_i \end{pmatrix} \begin{pmatrix} x_i \\ s_i \end{pmatrix}^{\top}\\
& +& \sum_{\omega_i=0} D\left(x_i^{\top}\beta_\gamma + s_i^{\top} \delta_\gamma \right)\begin{pmatrix} x_i \\ s_i \end{pmatrix} \begin{pmatrix} x_i \\ s_i \end{pmatrix}^{\top}.
\end{eqnarray*}

Our technical conditions for the probit model are:
\begin{enumerate}[label=\bfseries C\arabic*.]
\item  The parameter space is $\Gamma_{\gamma}= {\mathbb R}^{p_\gamma + r s_\gamma}$.

\item The marginal probability $P_{F_0}(\omega_1 = 0) \in (0,1)$.

\item For given $\gamma$, it holds that $\lim_{n \rightarrow \infty} P_{F_0}\left({\mathcal C}_{n,\gamma} \cup {\mathcal Q}_{n,\gamma}\right)= 0$.
${\mathcal C}_{n,\gamma}$ is the set of completely-separated datasets $(\omega,X_\gamma,S_\gamma)$, that is for which there exists $\eta_{\gamma} \in \Gamma_{\gamma}$ such that
\begin{eqnarray*}
x_{i \gamma}^{\top}\alpha_\gamma + s_{i \gamma}^{\top} \kappa_\gamma > 0 \mbox{ for all } \omega_i = 1
\mbox{ and } x_{i \gamma}^{\top}\alpha_\gamma + s_{i \gamma}^{\top} \kappa_\gamma < 0 \mbox{ for all } \omega_i = 0,
\end{eqnarray*}
and ${\mathcal Q}_{n,\gamma}$ is the set of quasi-completely-separated $(\omega,X_\gamma,S_\gamma)$, for which there exists $\eta_{\gamma} \in \Gamma_{\gamma}$ such that
\begin{eqnarray*}
x_{i \gamma}^{\top}\alpha_\gamma + s_{i \gamma}^{\top} \kappa_\gamma \geq 0 \mbox{ for all } \omega_i = 1
\mbox{ and } x_{i \gamma}^{\top}\alpha_\gamma + s_{i \gamma}^{\top} \kappa_\gamma \leq 0 \mbox{ for all } \omega_i = 0,
\end{eqnarray*}
with equality for at least one $i=1,\ldots,n$.

\item \color{black} There exists a maximum $\eta_\gamma^* \in \Gamma_\gamma$ of $\widetilde{M}(\eta_\gamma)=E_{F_0}[\widetilde{m}(\eta_\gamma)]$ such that $\widetilde{M}(\eta_\gamma^*) < \infty$,
and $E_{F_0}[|\widetilde{m}(\eta_\gamma)|] < \infty$ for all $\eta_\gamma \in \Gamma_\gamma$.  
Further, the MLE $\widehat{\eta}_\gamma$ exists with probability converging to 1, as $n \rightarrow \infty$.

%
%

\item  There exists a neighborhood ${\mathcal B}_{\eta_{\gamma}^*}$ of  $\eta_{\gamma}^*$ such that, for any $\eta_{\gamma} \in {\mathcal B}_{\eta_{\gamma}^*}$,
\begin{eqnarray*}
{\mathbb E}_{F_0} \left[\sup_{\eta_\gamma\in {\mathcal B}_{\eta_\gamma^*}} \vert\vert \nabla   \widetilde{m}(\eta_\gamma)\vert\vert^2 \right] &<& \infty.
\end{eqnarray*}

\item The entries of the second derivative matrix $ \nabla^2_{\eta_{\gamma}}\widetilde{M}(\eta_\gamma^*) $ are finite. 
\color{black}
\end{enumerate}

Briefly, C3 guarantees that there is no perfect separation in the data 
in a sequence of sets whose probability converges to 1, which implies log-likelihood concavity \citep{albert:1984}.
A detailed study of which situations lead to perfect separation is beyond our scope, we refer the reader to \cite{lesaffre:1989} and \cite{heinze:2002}.
The remaining conditions are analogous to Conditions A1-A4  for the Gaussian AFT model, discussed in Section \ref{sec:proof_asymp}.

\begin{proposition}\label{prop:ConsistencyP}
Assume Conditions C1-C4.
Then, $\widetilde{M}(\eta_\gamma)$ has a unique maximizer $\eta_\gamma^* = \operatorname{argmax}_{\eta_\gamma} \widetilde{M}(\eta_\gamma)$.
Moreover, $\widehat{\eta}_\gamma \stackrel{P}{\rightarrow} \eta_\gamma^*$ as $n\rightarrow \infty$.
\end{proposition}

\begin{proposition}\label{prop:AsympNormP}
Assume Conditions C1-C6.
Let $V_{\eta_\gamma^*}$ be the Hessian of $\widetilde{M}(\eta_{\gamma})$ at $\eta_\gamma^*$.
Then
$\sqrt{n}(\widehat{\eta}_\gamma-\eta_\gamma^*) \stackrel{D}{\longrightarrow}
N \left( 0, V_{\eta_\gamma^*}^{-1} {\mathbb E}_{F_0}[ \nabla {\widetilde{m}}(\eta_\gamma^*) \nabla {\widetilde{m}}(\eta_\gamma^*)^{\top}] V_{\eta_\gamma^*}^{-1}\right)$,
where $V_{\eta_\gamma^*}$ is the Hessian of $\widetilde{M}(\eta_\gamma)$ evaluated at $\eta_\gamma^*$
and $\widetilde{m}(\eta_\gamma^*) = \log p(\omega_1 \mid \eta_\gamma^*)$.
\end{proposition}

\begin{proposition}\label{prop:BFRatesP}
Let $B_{\gamma,\gamma^*}$ be the Bayes factor in \eqref{eq:bf_laplace} under either $\pi_L$, $\pi_M$ or $\pi_E$,
where $\gamma^*$ is the probit model with smallest $d_{\gamma^*}$ minimizing \eqref{eq:probit} for $F_0$ as given in Condition B1, and $\gamma \neq \gamma^*$ another probit model.
Assume that both $\gamma^*$ and $\gamma$ satisfy Conditions C1-C6.
Suppose that $(g_M,g_E,g_L,g_S)$ are non-decreasing in $n$.

\begin{enumerate}[leftmargin=*,label=(\roman*)]

\item Let $a_n$ be as in Proposition \ref{prop:BFRatesI}. If $\gamma^* \subset \gamma$, then
$$
\log B_{\gamma \gamma^*}= \log (a_n) + \frac{r}{2} (s_{\gamma^*} - s_\gamma) \log \left( n g_S \right) + \lOp(1),
$$

\item Let $b_n$ be as in Proposition \ref{prop:BFRatesI}.
If $\gamma^* \not\subset \gamma$, then under $\pi_L$, $\pi_M$ and $\pi_E$ it holds that
$$
\log(B_{\gamma \gamma^*})= -n[\widetilde{M}(\eta^*_{\gamma^*}) - \widetilde{M}(\eta^*_\gamma)]
+ \log(b_n) + \frac{r}{2}(s_{\gamma^*} - s_\gamma) \log(n g_S) + \lOp(1).
$$
\end{enumerate}
\end{proposition}

Propositions \ref{prop:ConsistencyP}-\ref{prop:BFRatesP} are akin to Propositions \ref{prop:Consistency}-\ref{prop:BFRatesI},
see Section \ref{sec:theory} for a discussion of their implications, although here the number of parameters $d_\gamma= \mbox{dim}(\eta_\gamma)= p_\gamma + r s_\gamma $ (i.e. the probit model has no dispersion parameter).
First, we prove an extension of the Convexity Lemma \citep{pollard:1991} in which we allow for the probability of convexity to converge to 1, rather than requiring convexity for almost all realizations of the sequence of random functions.
\begin{lemma}\label{le:ExtConvexityLemma}
Let $\{\lambda_n(\theta): \theta \in \Theta\}$ be a sequence of random  functions defined on a convex, open subset $\Theta$ of ${\mathbb R}^d$. Suppose that $p_n=P(\lambda_n \text{ is convex}) \to 1$ as $n\to \infty$. Suppose $\lambda(\cdot)$ is a real valued function on $\Theta$ for which $\lambda_n(\theta) \to \lambda(\theta)$ in probability, for each $\theta \in \Theta$. Then, for each compact subset $K$ of $\Theta$,
\begin{eqnarray*}
\sup_{\theta\in K} \vert \lambda_n(\theta) - \lambda(\theta)\vert \stackrel{P}{\to} 0.
\end{eqnarray*}
The function $\lambda(\cdot)$ is necessarily convex on $\Theta$.

\proof

Let $S_n = \{ \omega \in \Omega : \lambda_n^{(\omega)} \text{ is convex}\} $, where $\lambda_n^{(\omega)}$ denotes a realization of the random function $\lambda_n$ associated to an element $\omega$ in the sample space $\Omega$. By assumption, $p_n =  P(S_n) \to 1$ as $n\to \infty$. For any fixed $\epsilon>0$
\begin{eqnarray*}
P\left( \sup_{\theta\in K} \vert \lambda_n(\theta) - \lambda(\theta)\vert  < \epsilon \right) &=&  P\left( \omega\in S_n : \sup_{\theta\in K} \vert \lambda_n^{(\omega)}(\theta) - \lambda(\theta)\vert < \epsilon \right)\\
&+&  P\left( \omega\in S_n^{\complement} : \sup_{\theta\in K} \vert \lambda_n^{(\omega)}(\theta) - \lambda(\theta)\vert < \epsilon \right).
\end{eqnarray*}
Since $P\left( \omega\in S_n^{\complement} : \sup_{\theta\in K} \vert \lambda_n^{(\omega)}(\theta) - \lambda(\theta)\vert < \epsilon \right) \leq P(S_n^{\complement}) = 1-p_n\to 0$, by assumption, the goal is to prove that
\begin{eqnarray*}
P\left( \omega\in S_n : \sup_{\theta\in K} \vert \lambda_n^{(\omega)}(\theta) - \lambda(\theta)\vert < \epsilon \right) \rightarrow 1.
\end{eqnarray*}

Now, consider the sequence of random functions
$$\widetilde{\lambda}_n(\theta) =  \begin{cases}
\lambda_n(\theta), \omega \in S_n,\\
\varphi(\theta), \omega \in S_n^{\complement},
\end{cases}$$
where $\varphi(\cdot)$ is an any finite convex approximation of $\lambda(\cdot)$. By definition, this is a sequence of convex functions and $\widetilde{\lambda}_n(\theta) \stackrel{P}{\to} \lambda(\theta)$. Then, by the convexity lemma in \citep{pollard:1991},
\begin{eqnarray*}
\sup_{\theta\in K} \vert \widetilde{\lambda}_n(\theta) - \lambda(\theta)\vert \stackrel{P}{\to} 0,
\end{eqnarray*}
and $\lambda$ is convex on $\Theta$. Thus, for each $\epsilon>0$
 \begin{eqnarray*}
P\left(\sup_{\theta\in K} \vert \widetilde{\lambda}_n(\theta) - \lambda(\theta)\vert < \epsilon \right) &=&
P\left( \omega\in S_n : \sup_{\theta\in K} \vert \lambda_n^{(\omega)}(\theta) - \lambda(\theta)\vert < \epsilon \right)\\
 &+& P\left( \omega\in S_n^{\complement} : \sup_{\theta\in K} \vert \varphi(\theta) - \lambda(\theta)\vert < \epsilon \right).
\end{eqnarray*}
From the previous statement, we know that $P\left(\sup_{\theta\in K} \vert \widetilde{\lambda}_n(\theta) - \lambda(\theta)\vert < \epsilon \right)  \to 1$ and, since 
$P\left( \omega\in S_n^{\complement} : \sup_{\theta\in K} \vert \varphi(\theta) - \lambda(\theta)\vert < \epsilon \right)  \leq P(S_n^{\complement}) \to 0$, we obtain
$$P\left( \omega\in S_n : \sup_{\theta\in K} \vert \lambda_n^{(\omega)}(\theta) - \lambda(\theta)\vert < \epsilon \right) \to 1,$$
as we wished to prove.

\end{lemma}

\subsection*{Proof of Proposition \ref{prop:ConsistencyP}}

The proof is analogous to that of Proposition \ref{prop:Consistency} in Section \ref{app:proof_consistency}.
Let $\widetilde{M}_n(\eta_\gamma) = n^{-1}\log p_B(w \mid \eta_\gamma)$ be the average log-likelihood evaluated at $\eta_\gamma$. 
The assumption that data are generated \emph{i.i.d} from $F_0$ \color{black} and that $E_{F_0}(|\widetilde{m}(\eta_\gamma)|) < \infty$ \color{black} give that,
by the law of large numbers, $\widetilde{M}_n(\eta_\gamma) \stackrel{P}{\rightarrow} \widetilde{M}(\eta_\gamma)$, for each $\eta_\gamma\in\Gamma_\gamma$. 


The aim is to first show that the average log-likelihood $\widetilde{M}_n(\eta_\gamma)$ converges to its expected value $\widetilde{M}(\eta_\gamma)$ uniformly in $\eta_\gamma$, and then that this implies $\widehat{\eta}_\gamma \stackrel{P}{\longrightarrow} \eta_\gamma^*$. To see that $\widetilde{M}_n(\eta_\gamma)$ converges to $\widetilde{M}(\eta_\gamma)$, uniformly in $\eta_\gamma$,
we use that, under Conditions C1--C4 and by the results in \cite{albert:1984}, $\widetilde{M}_n(\eta_\gamma)$ is a sequence of concave functions in $\eta_\gamma$, for all samples in the sequence of sets ${\mathcal A}_n=\left({\mathcal C}_{n,\gamma} \cup {\mathcal Q}_{n,\gamma}\right)^\complement$, whose probability converge to 1. Thus, Lemma \ref{le:ExtConvexityLemma} implies that
\begin{eqnarray}\label{SC1P}
\sup_{\eta_\gamma \in K} \left\vert \widetilde{M}_n(\eta_\gamma) - \widetilde{M}(\eta_\gamma)\right\vert \stackrel{P}{\longrightarrow} 0,
\end{eqnarray}
for each compact set $K\subseteq\Theta_\gamma$, and also that $\widetilde{M}(\eta_\gamma)$ is finite and concave in $\Gamma_\gamma$.
\color{black} Since $M(\eta_\gamma)$ has a maximum and the MLE exists by C4, this implies that $\eta_\gamma^*$ and $\widehat{\eta}_\gamma$ must both occur in some compact set $K$ and hence that they are both unique (by concavity). \color{black}

The uniform convergence in \eqref{SC1P} and the uniqueness of the maximum provide the conditions to apply Theorem 5.7 from \cite{vandervaart:1998}, which gives that $\widehat{\eta}_\gamma \stackrel{P}{\longrightarrow} \eta_\gamma^*$.

\subsection*{Proof of Proposition \ref{prop:AsympNormP}}

The idea of the proof is the analogous to that of Proposition \ref{prop:AsympNormP}. Define,
\begin{eqnarray*}
\widetilde{m}_{\eta_\gamma}(\omega_1,x_1,s_1) &\equiv& \widetilde{m}(\eta_\gamma) \\
&=& \omega_1  \log\Phi\left(x_1^{\top}\alpha_\gamma + s_1^{\top} \kappa_\gamma \right) +  (1-\omega_1)
  \log\Phi\left(-x_1^{\top}\alpha_\gamma - s_1^{\top} \kappa_\gamma \right).
\end{eqnarray*}
let ${\mathcal B}_{\eta^*_{\gamma}}\subset \Theta_\gamma$ be a neighbourhood of $\eta_\gamma^*$ and consider $\eta_1,\eta_2 \in {\mathcal B}_{\eta^*_\gamma}$.
We need to show that $\vert \widetilde{m}_{\eta_1}(\omega_1,x_1,s_1) - \widetilde{m}_{\eta_2}(\omega_1,x_1,s_1) \vert$ has finite expectation under $F_0$.
Using the mean value theorem and the Cauchy-Schwarz inequality it follows that, with probability 1,
\begin{eqnarray*}
\vert \widetilde{m}_{\eta_1}(\omega_1,x_1,s_1) - \widetilde{m}_{\eta_2}(\omega_1,x_1,s_1) \vert &=& \vert \nabla \widetilde{m}_{\eta_c}(\omega_1,x_1,s_1)^{\top} (\eta_1-\eta_2)\vert\\
&\leq& \vert \vert \nabla \widetilde{m}_{\eta_c}(\omega_1,x_1,s_1) \vert \vert \cdot \vert \vert \eta_1-\eta_2 \vert \vert \\
&\leq& K(\omega_1,x_1,s_1) \cdot \vert \vert \eta_1-\eta_2 \vert \vert,
\end{eqnarray*}
where $\nabla \widetilde{m}_{\eta_\gamma}(\omega_1,x_1,s_1)$ is the gradient of $\widetilde{m}_{\eta_\gamma}(\omega_1,x_1,s_1)$,
$\eta_c = (1-c)\eta_1+c\eta_2$ for some $c\in(0,1)$
and $K(\omega_1,x_1,s_1)= \sup_{{\mathcal B}_{\eta_\gamma^*}} \vert\vert \nabla \widetilde{m}_{\eta_\gamma}(\omega_1,x_1,s_1) \vert\vert$.
This bound together with assumption C5
imply that
\begin{eqnarray*}
\int K(\omega_1,x_1,s_1)^2 dF_0(\omega_1 ,z_1)  < \infty,
\end{eqnarray*}
where $z_1 = (x_1^{\top},s_1^{\top})^{\top}$.
Therefore, by using the mean value theorem and the Cauchy-Schwarz inequality, it follows that for $\eta_1, \eta_2\in{\mathcal B}_{\eta_\gamma^*}$, with probability 1,
\begin{eqnarray*}
\vert \widetilde{m}_{\eta_1}(\omega_1,x_1,s_1) - \widetilde{m}_{\eta_2}(\omega_1,x_1,s_1) \vert &=& \vert \nabla \widetilde{m}_{\eta_c}(\omega_1,x_1,s_1)^{\top} (\eta_1-\eta_2)\vert\\
&\leq& \vert \vert \nabla \widetilde{m}_{\eta_c}(\omega_1,x_1,s_1) \vert \vert \cdot \vert \vert \eta_1-\eta_2 \vert \vert
\\&\leq& K(\omega_1,x_1,s_1) \cdot \vert \vert \eta_1-\eta_2 \vert \vert,
\end{eqnarray*}
where $\eta_c = (1-c)\eta_1+c\eta_2$, for some $c\in(0,1)$, and $\eta_1,\eta_2 \in {\mathcal B}_{\eta^*_{\gamma}}\subset \Theta_\gamma$.

\color{black}
We now present the entries of  the Hessian of the expected log-likelihood $\widetilde{M}(\eta_\gamma)$.
\color{black}
Let us denote $\rho =  P_{F_0}(\omega_1=1)$.
\begin{eqnarray*}
\widetilde{M}(\eta_{\gamma}) &=& \rho \int \log\Phi\left(x_1^{\top}\alpha_\gamma + s_1^{\top} \kappa_\gamma \right) dF_0(z_1 \mid \omega_1=1)
 \nonumber \\
& + & (1-\rho)\int \log\Phi\left(-x_1^{\top}\alpha_\gamma - s_1^{\top} \kappa_\gamma \right) dF_0(z_1 \mid \omega_1=0).
\end{eqnarray*}
To obtain the gradient of $\widetilde{M}(\eta_{\gamma})$, note that under Conditions C1-C4 we can apply Leibniz's integral rule to differentiate under the integral sign, and hence
\begin{eqnarray*}
\nabla_{(\beta_\gamma,\delta_\gamma)} \widetilde{M}(\eta_{\gamma}) &=& \rho \int r\left(x_1^{\top}\alpha_\gamma + s_1^{\top} \kappa_\gamma \right)  \begin{pmatrix} x_1 \\ s_1 \end{pmatrix}dF_0(z_1 \mid \omega_1=1)  \\
&+&  (1-\rho)    \int r\left(x_1^{\top}\alpha_\gamma + s_1^{\top} \kappa_\gamma \right)  \begin{pmatrix} x_1 \\ s_1 \end{pmatrix}  dF_0(z_1 \mid  \omega_1=0),
\end{eqnarray*}
Similarly, the entries of the Hessian matrix are
\begin{eqnarray*}
\nabla^2_{(\beta_\gamma,\delta_\gamma)} \widetilde{M}(\eta_{\gamma}) &=&
- \rho \int D\left(-x_1^{\top}\beta_\gamma - s_1^{\top} \delta_\gamma \right)  \begin{pmatrix} x_1 \\ s_1 \end{pmatrix} \begin{pmatrix} x_1 \\ s_1 \end{pmatrix}^{\top} dF_0(z_1 \mid \omega_1=1)\\
&+& (1-\rho)    \int  D\left(x_1^{\top}\beta_\gamma + s_1^{\top} \delta_\gamma \right)\begin{pmatrix} x_1 \\ s_1 \end{pmatrix} \begin{pmatrix} x_1 \\ s_1 \end{pmatrix}^{\top} dF_0(z_1 \mid  \omega_1=0).
\end{eqnarray*}
The finiteness of $\widetilde{M}(\eta_{\gamma})$, its gradient and Hessian follows by conditions C4--C6. From Proposition \ref{prop:ConsistencyP}, together with conditions C4--C6, we have that $\widetilde{M}(\eta_{\gamma})$ is concave and, consequently, the Hessian
\begin{eqnarray*}
V_{\eta_{\gamma}} =   \nabla^2_{(\beta_\gamma,\delta_\gamma)} \widetilde{M}(\eta_{\gamma}) .
\end{eqnarray*}
is non-singular at $\eta_{\gamma}^*$. Thus, the asymptotic normality follows by Theorem 5.23 from \cite{vandervaart:1998} together with the consistency results in Proposition \ref{prop:ConsistencyP}.

\subsection*{Proof of Proposition \ref{prop:BFRatesP}}

The proof is analogous to that of Propositions \ref{prop:BFRatesI} and \ref{prop:BFRatesCox}.
For brevity, we provide a sketch.
The decomposition of the Laplace approximation to Bayes factors is analogous to \eqref{eq:BFLap}.
Condition C3 guarantees that the probability of separation of the data converges to zero as the sample size increases,
which in turn guarantees existence and uniqueness of the MLE in a sequence of sets of realizations with probability converging to 1 \citep{albert:1984}.
Further, from Proposition \ref{prop:ConsistencyP} and Proposition \ref{prop:AsympNormP} the MLE is consistent, converges to the KL-optimal value,
and is asymptotically normally-distributed.

These results imply that the behavior of the second and third terms in \eqref{eq:BFLap} remains as in Proposition \ref{prop:BFRatesI}.
They also imply that in the case $\gamma^* \subset \gamma$ the asymptotic expansion of the log-likelihood ratio remains as in Proposition \ref{prop:BFRatesI},
namely
$$2 [ \widetilde{\ell}(\widehat{\eta}_\gamma) -  \widetilde{\ell}(\widehat{\eta}_{\gamma^*}) ]=
 z^\top W_{\gamma \setminus \gamma^*}^{1/2} U^{-1}_{\gamma \setminus \gamma^*} W_{\gamma \setminus \gamma^*}^{1/2} z +  \lo_p(1)$$
and hence $2[ \widetilde{\ell}(\widehat{\eta}_\gamma) -  \widetilde{\ell}(\widehat{\eta}_{\gamma^*})] = \lOp(1)$.
The case $\gamma^* \not\subset \gamma$ follows, as in Proposition \ref{prop:BFRatesI}, from the law of large numbers and uniform convergence,
the latter being guaranteed by the log-likelihood's concavity and Pollard's convexity lemma \citep{pollard:1991} and its extension in Lemma \ref{le:ExtConvexityLemma}.
Finally, the concavity of the expected log-likelihood in a neighbourhood of the KL optimal value, established in the proof of  Proposition \ref{prop:AsympNormP}, guarantees that the fourth term in \eqref{eq:BFLap} remains as in Proposition \ref{prop:BFRatesI}.

\section{Theory for log-concave likelihoods}
\label{sec:theory_logconcave}

We extend our theory to general concave log-likelihood settings,
including generalized linear models under the natural parameter such as logistic or Poisson regression.
The main result is that misspecified Bayes factors are asymptotically equivalent to their Laplace approximation,
the technical basis for the result builds on the asymptotic misspecified MLE normality results of \cite{hjort:2011}.
The consequence is that, under minimal conditions, one obtains Bayes factor rates of the same form 
as in Propositions \ref{prop:BFRatesI}-~\ref{prop:BFRatesCox}.
The validity of Laplace approximations was also studied in \cite{kass:1990} for even more general likelihoods,
where the authors remarkably bounded the relative error,
our contribution is considering misspecification and that our technical conditions are significantly simpler, albeit restricted to log-concave likelihoods.
We briefly outline the conditions and the strategy how to check them.

Let $(y_i,z_i) \sim F_0$ independently for $i=1,\ldots,n$. Let the optimal parameter value be
$$
\eta_\gamma^*= \arg\max_{\eta_\gamma} {\mathbb E}_{F_0} \left[ \log p(y_1,z_1 \mid \eta_\gamma, \gamma) \right],
$$
and $H(\eta_\gamma^*)$ the observed log-likelihood hessian evaluated at $\eta_\gamma=\eta_\gamma^*$.
Let 
  $$
A_n(s)= \log p(y \mid \widehat{\eta}_\gamma, \gamma) - \log p(y \mid \widehat{\eta}_\gamma + s/\sqrt{n}, \gamma)= - \frac{s^{\top} H(\eta^*_\gamma) s}{2 n} + r_n(s/\sqrt{n}),
$$
where $r_n$ is the error in the second-order Taylor log-likelihood expansion at $\eta_\gamma^*$.
Assume that
\begin{enumerate}[leftmargin=*, label={\bf D\arabic*.}]
\item $\widehat{\eta}_\gamma \stackrel{P}{\longrightarrow} \eta_\gamma^*$ under $F_0$, as $n \rightarrow \infty$.
 
\item $H(\eta_\gamma^*)/n \stackrel{P}{\longrightarrow} H_\gamma^*$ under $F_0$, as $n \rightarrow \infty$, for positive-definite $H^*_\gamma$.
 
\item $\pi(\eta_\gamma \mid \gamma)$ is continuous at $\eta^*_\gamma$ and $\sup_{\eta_\gamma } \pi(\eta_\gamma \mid \gamma) \leq C$ for a constant $C$.
\end{enumerate}

Conditions D1-D3 are minimal. If one assumes that $\eta_\gamma$ has bounded support, then D1 holds \citep{hjort:2011} and D2 also holds provided $|H(\eta_\gamma^*)|$ has finite mean, by the continuous mapping theorem and strong law of large numbers.
More generally one may show the asymptotic validity of a Taylor log-likelihood expansion around $\eta_\gamma^*$
and establish asymptotic normality of $\widehat{\eta}_\gamma$, see Theorem 4.1 in \cite{hjort:2011}.
The boundedness prior condition in D3 can be relaxed, but simplifies the proof.

\begin{proposition}
 Assume that $P\left( A_n(s) \text{ is convex in } s\right ) \to 1$, as $n\to \infty$, and Conditions D1-D3. Then, as $n \rightarrow \infty$,
 \begin{equation}
p(y \mid \gamma) \times
\frac{n^{d_\gamma/2} |H_\gamma^*|^{1/2}}{p(y \mid \widehat{\eta}_\gamma, \gamma) \pi(\eta_\gamma^* \mid \gamma) (2\pi)^{d_\gamma/2}} \stackrel{P}{\longrightarrow} 1.
\end{equation}
\label{prop:valid_laplace_concave}
\end{proposition}
As an immediate corollary, for any local prior $B_{\gamma \gamma^*}$ is asymptotically equivalent to
\begin{align}
e^{L(\gamma, \gamma^*)/2} \frac{\pi(\eta_\gamma^* \mid \gamma)}{\pi(\eta_{\gamma^*}^* \mid \gamma^*)}
\left(\frac{2 \pi}{n}\right)^{(d_\gamma - d_{\gamma^*})/2}
\frac{|H^*_{\gamma^*}|^{1/2}}{|H^*_\gamma|^{1/2}},
\nonumber
\end{align}
where $L(\gamma, \gamma^*)$ is the likelihood-ratio test statistic between $\gamma$ and $\gamma^*$.
For non-local priors the term $\pi(\eta_\gamma^* \mid \gamma)$ can be zero, which requires studying the rate at which $\pi(\widehat{\eta}_\gamma \mid \gamma)$ vanishes.
One may then characterize $B_{\gamma \gamma^*}$ via the asymptotic distribution of $L(\gamma, \gamma^*)$, which follows directly after proving the validity of the log-likelihood Taylor expansion, to obtain the Bayes factor rates featuring in Proposition \ref{prop:BFRatesI}, \ref{prop:BFRatesCox} and \ref{prop:BFRatesP}. We refer the reader to their proofs for further details.

To prove Proposition \ref{prop:valid_laplace_concave} we recall Lemma A3 in \cite{hjort:2011},
then state and prove a result on general conditions for the validity of Laplace approximations and finally prove the desired result for log-concave likelihoods.

\subsection{Dominated convergence from pointwise convergence and bounded in probability}

Lemma \ref{lem:convprob_integral} establishes that the integral of a random function that converges pointwise in probability to a limit function,
converges to the integral of the limit, provided the function can be bounded with probability tending to 1.
The proof is a direct consequence of the Dominated Convergence Theorem, see Lemma A3 from \cite{hjort:2011}.

\begin{lemma}
  (Lemma A3 from \cite{hjort:2011}).
  Let $\{G_n(s,\omega_n)\}$ be a sequence of random functions, where $\omega_n$ is the random outcome with probability distribution $F_0$.
  Suppose that $G_n(s,\omega_n) \stackrel{P}{\longrightarrow} G(s)$, for each $s$,
  and that there exists a function $W(s)$ such that $\int W(s) ds < \infty$
  and
  $$
\lim_{n \rightarrow \infty} P_{F_0}( \{ w: G_n(s,\omega_n) \leq W(s) \mbox{ for all } s \} ) = 1.
$$
Then $\int G_n(s, \omega_n) ds \stackrel{P}{\longrightarrow} \int G(s) ds$.
  \label{lem:convprob_integral}
\end{lemma}

\subsection{Asymptotic validity of Laplace approximation}

Recall that we assume $(y_i,z_i) \sim F_0$ independently from $i=1,\ldots,n$.
To ease notation let $\eta=\eta_\gamma$, $L_n(\eta)= p(y \mid \eta_\gamma, \gamma)$, $\pi(\eta)= \pi(\eta_\gamma \mid \gamma)$,
$m_n= \int L_n(\eta) \pi(\eta) d \eta$
and $H_n=H(\eta_\gamma^*)$.
where $\eta \in \Theta \subseteq \mathbb{R}^p$.
Recall that
$$
A_n(s)= \log L_n(\widehat{\eta}) - \log L_n(\widehat{\eta} + s/\sqrt{n})= - \frac{1}{2} s^{\top} \frac{H_n}{n} s + r_n(s/\sqrt{n}),
$$
where $r_n()$ is continuous at 0.

\begin{lemma}
Assume that
\begin{enumerate}[leftmargin=*]
\item Consistency. $\widehat{\eta} \stackrel{P}{\longrightarrow} \eta^*$, as $n \rightarrow \infty$.
{
\item  Pointwise convergence. There exist a positive definite matrix $H^*$ such that, as $n \rightarrow \infty$,
  $$
e^{-A_n(s)} \pi(\widehat{\eta} + s/\sqrt{n}) \stackrel{P}{\longrightarrow} e^{-\frac{1}{2} s^{\top} H^* s} \pi(\eta^*).
$$

}
\item Bounded in probability. There exists a function $W(s)$ such that $\int W(s) ds < \infty$ and
  $$
\lim_{n \rightarrow \infty} P_{F_0}\left( e^{-A_n(s)} \pi(\widehat{\eta} + s/\sqrt{n}) < W(s) \right)= 1.
  $$
\end{enumerate}
Then
$$
\frac{m_n n^{p/2}}{L_n(\widehat{\eta})} \stackrel{P}{\longrightarrow}
\frac{\pi(\eta^*) (2\pi)^{p/2}}{|H^*|^{1/2}},
$$

\label{lem:valid_laplace}
\end{lemma}

\begin{proof}

Note that
\begin{align}
  &\frac{m_n}{L_n(\widehat{\eta})}=  \int \exp \{ \log L_n(\eta) - \log L_n(\widehat{\eta}) \} \pi(\eta) d \eta=
  \nonumber \\
  & \frac{1}{n^{p/2}}  \int \exp \{ - [ \log L_n(\widehat{\eta} + s/\sqrt{n}) - \log L_n(\widehat{\eta} ) ] \} \pi(\widehat{\eta}+s/\sqrt{n}) d s
\nonumber \\
  & = \frac{1}{n^{p/2}} \int e^{-A_n(s)} \pi(\widehat{\eta}_n+s/\sqrt{n}) d s,
    \nonumber
\end{align}
where we used the change of variables $s= \sqrt{n}(\eta - \widehat{\eta})$, which has Jacobian equal to $n^{-p/2}$.
Given the pointwise convergence and boundedness in probability assumptions, we can apply Lemma \ref{lem:convprob_integral} to obtain
$$
\int e^{-A_n(s)} \pi(\widehat{\eta} +s/\sqrt{n}) d s \stackrel{P}{\longrightarrow} \pi(\eta^*) \int e^{-\frac{1}{2} s^{\top}H^*s} ds=
\pi(\eta^*) \frac{(2\pi)^{p/2}}{|H^*|^{1/2}},
$$
as we wished to prove.
\end{proof}

\subsection{Proof of Proposition \ref{prop:valid_laplace_concave}}

The proof is an adaptation of the proof of Theorem 4.2 in \cite{hjort:2011}, using Lemmas \ref{lem:convprob_integral} and \ref{lem:valid_laplace}.
From the proof of Lemma \ref{lem:valid_laplace},
\begin{align}
  \frac{m_n n^{p/2}}{L_n(\widehat{\eta})}=  \int e^{-A_n(s)} \pi(\widehat{\eta}+s/\sqrt{n}) d s=
  \int e^{\frac{1}{2} s^{\top} (H_n/n) s + r_n(s/\sqrt{n})}  \pi(\widehat{\eta}+s/\sqrt{n}) d s.
    \nonumber
\end{align}
If the conditions in Lemma \ref{lem:valid_laplace} hold, then
$$
\frac{m_n n^{p/2}}{L_n(\widehat{\eta})} \stackrel{P}{\longrightarrow} \int e^{\frac{1}{2} s^{\top}H^*s} ds=
\frac{(2\pi)^{p/2}}{|H^*|^{1/2}} > 0,
$$
proving the result.
Hence we just need to check that Conditions 2-3 in Lemma \ref{lem:valid_laplace} are satisfied.
Our Assumptions 2-3, along with $r_n()$ being continuous at 0, guarantee pointwise convergence
of $e^{-A_n(s)} \pi(\widehat{\eta} + s/\sqrt{n})$ to $e^{-\frac{1}{2} s^{\top} H^* s} \pi(\eta^*)$ for any fixed $s$, as $n \rightarrow \infty$,
satisfying Condition 2 in Lemma \ref{lem:valid_laplace}.
To verify Condition 3 in Lemma \ref{lem:valid_laplace} we need to find an integrable function $W(s)$ that bounds
$e^{-A_n(s)} \pi(\widehat{\eta} + s/\sqrt{n})$ with probability tending to 1.
First, note that $\pi(\widehat{\eta} + s/\sqrt{n}) < C$ by assumption, so it suffices to bound $e^{-A_n(s)}$.
Second, given that $A_n(s)$ is convex with probability tending to $1$, $A_n(s)$ converges uniformly in probability
to $- 0.5 s^{\top}H^*s$ in compact sets (Lemma \ref{le:ExtConvexityLemma} and the Convexity Lemma in \citealp{pollard:1991}). That is,
$$\sup_{s \in K} |A_n(s) - 0.5 s^{\top}H^*s| \stackrel{P}{\longrightarrow} 0,$$
for any compact $K \in S$, where $S \subseteq \eta$ is an open convex set.
Take the set defined by $\vert\vert s \vert\vert \leq 1$ and define, within that set
\begin{align}
  e^{-A_n(s)} \leq e^{- \frac{1}{2} s^{\top}H^*s} e^{\sup_{\vert\vert s \vert\vert  \leq 1} |A_n(s) + \frac{1}{2} s^{\top}H^*s|} \stackrel{P}{\longrightarrow} e^{- \frac{1}{2} s^{\top}H^*s},
\nonumber
\end{align}
therefore for $\vert\vert s \vert\vert \leq 1$ we may define $W(s)= a C e^{-\frac{1}{2} s^{\top} H^* s}$ for any constant $a>1$, which is integrable.

To define $W(s)$ for $\vert\vert s \vert\vert >1$,
from the convexity of $A_n(s)$ it follows that $A_n(s) \geq \gamma_n \vert\vert s \vert\vert $ for $\vert\vert s \vert\vert \geq 1$, where
$\gamma_n= \inf_{ \vert\vert t  \vert\vert=1} A_n(t)$ (see Lemma \ref{lem:bound_convexfun} below).
Further 
$\gamma_n \stackrel{P}{\longrightarrow} \gamma_0= \inf_{ \vert\vert t  \vert\vert=1} 0.5 t^{\top} H^* t < \infty$,
where finiteness follows from $H^*$ being finite,
and $\gamma_n \stackrel{P}{\longrightarrow} \gamma_0$ from Lemma 2 (nearness of argmins) in \cite{hjort:2011}
by noting that $A_n(s)$ is convex and converges uniformly to $0.5s^{\top}H^*s$ in compact sets,
and that 
$$
\inf_{\vert\vert s - s_0 \vert\vert< \delta} -\frac{1}{2} s^{\top}H^*s + \frac{1}{2} s_0^{\top}H^* s_0 - s_0^{\top} s,
$$
where $s_0= \arg\min -0.5 s^{\top}H^*s$
is bounded below by a positive constant for each fixed $\delta$, since $H^*$ is positive definite.
Hence for $\vert\vert s \vert\vert >1$ it holds that
\begin{align}
  e^{-A_n(s)} \leq e^{- \gamma_n \vert\vert s \vert\vert } \stackrel{P}{\longrightarrow} e^{- \gamma_0 \vert\vert s \vert\vert },
\nonumber
\end{align}
where $\int_{\vert\vert s \vert\vert >1} e^{- \gamma_0 \vert\vert s \vert\vert } < \infty$.
Hence for $\vert\vert s \vert\vert >1$ we may define the integrable function $W(s)= a C e^{- \gamma_0 \vert\vert s \vert\vert }$ for any constant $a>1$.
To summarize, $e^{-A_n(s)} \pi(\widehat{\eta} + s/\sqrt{n} \leq W(s)$ with probability tending to 1,
where $\int_{\vert\vert s \vert\vert  \leq 1} W(s)ds + \int_{\vert\vert s \vert\vert >1} W(s) ds < \infty$, satisfying Condition 3 in Lemma \ref{lem:valid_laplace}.
As a final remark, note that the proof applies to any norm $\vert\vert s \vert\vert $ in $\mathbb{R}^p$.

\begin{lemma}
  Let $A_n(s)$ be a convex function such that $A_n(0)=0$, and let $\gamma_n= \inf_{ \vert\vert t  \vert\vert=1} A_n(t)$. Then $A_n(s) \geq \gamma_n \vert\vert s \vert\vert $ for any $\vert\vert s \vert\vert >1$.
  \label{lem:bound_convexfun}
\end{lemma}

\begin{proof}
  Let $s$ be a point outside the unit ball, \textit{i.e.}~$s= l u$ where $ \vert\vert u \vert\vert=1$ and $l>1$. Convexity of $A_n(s)$ implies
  $$
A_n(u) \leq \left (1 - \frac{1}{l} \right) A_n(0) + \frac{1}{l} A_n(s)= \frac{1}{l} A_n(s)
$$
and, since $\vert\vert s \vert\vert =l$,
$
A_n(s) = \vert\vert s \vert\vert  A_n(u) \geq \vert\vert s \vert\vert  \inf_{ \vert\vert t  \vert\vert=1} A_n(t).
$
\end{proof}

\clearpage

\section{Additional Figures and Tables}\label{app:figs_tables}

\subsection{Simulation study}\label{app:simul}

\begin{figure}[h!]
  \begin{center}
    \begin{tabular}{cc}
      Scenario 1 & Scenario 2 \\
      \includegraphics[width=0.5\textwidth,height=0.4\textwidth]{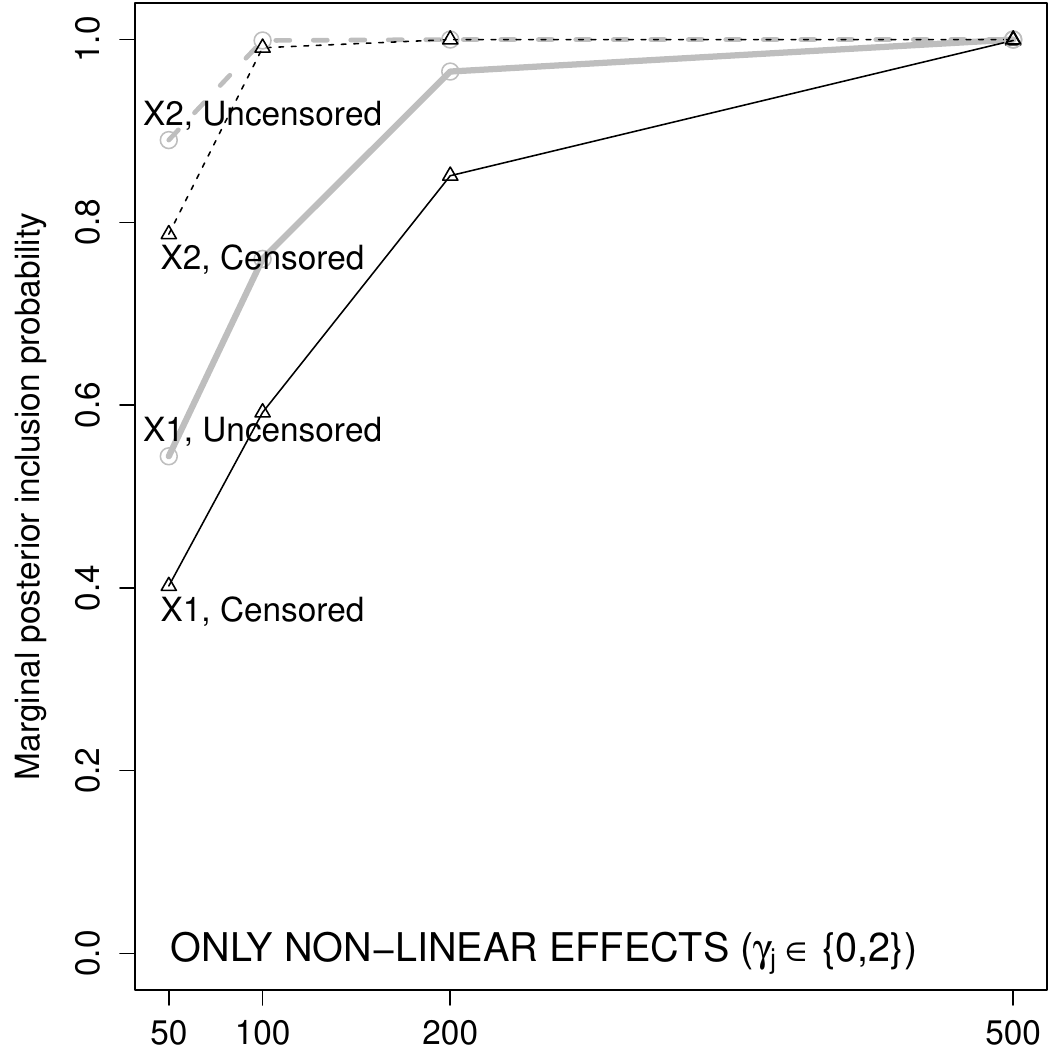} &
      \includegraphics[width=0.5\textwidth,height=0.4\textwidth]{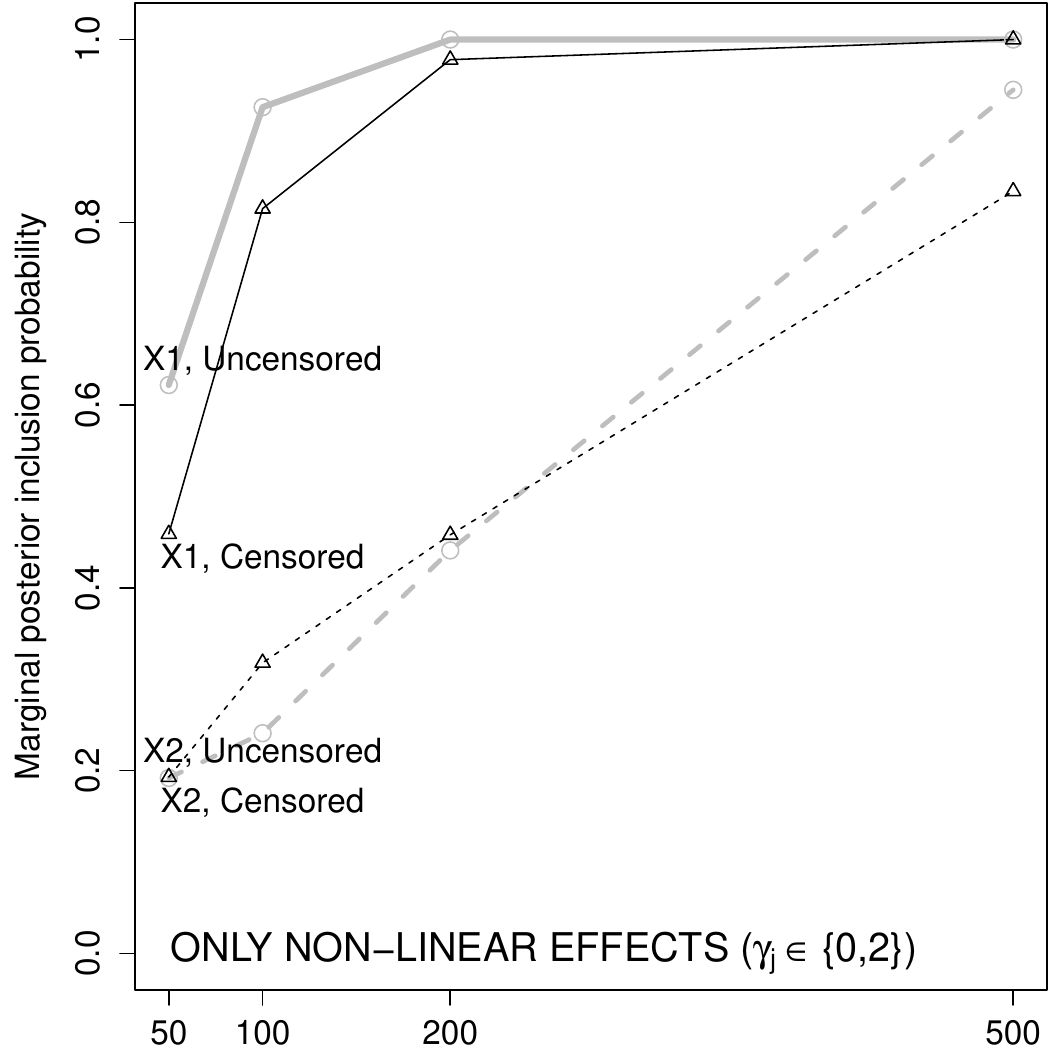} \\
      \includegraphics[width=0.5\textwidth,height=0.4\textwidth]{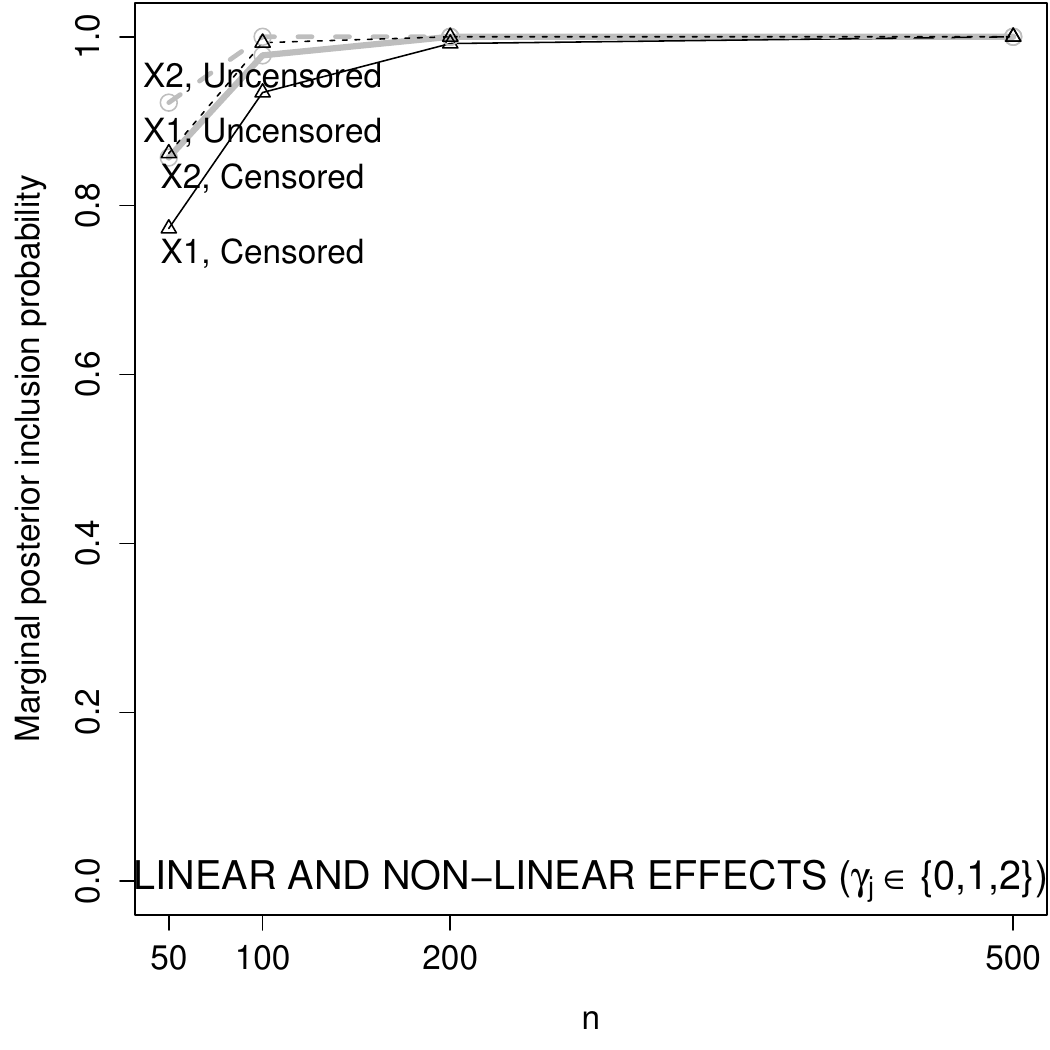} &
      \includegraphics[width=0.5\textwidth,height=0.4\textwidth]{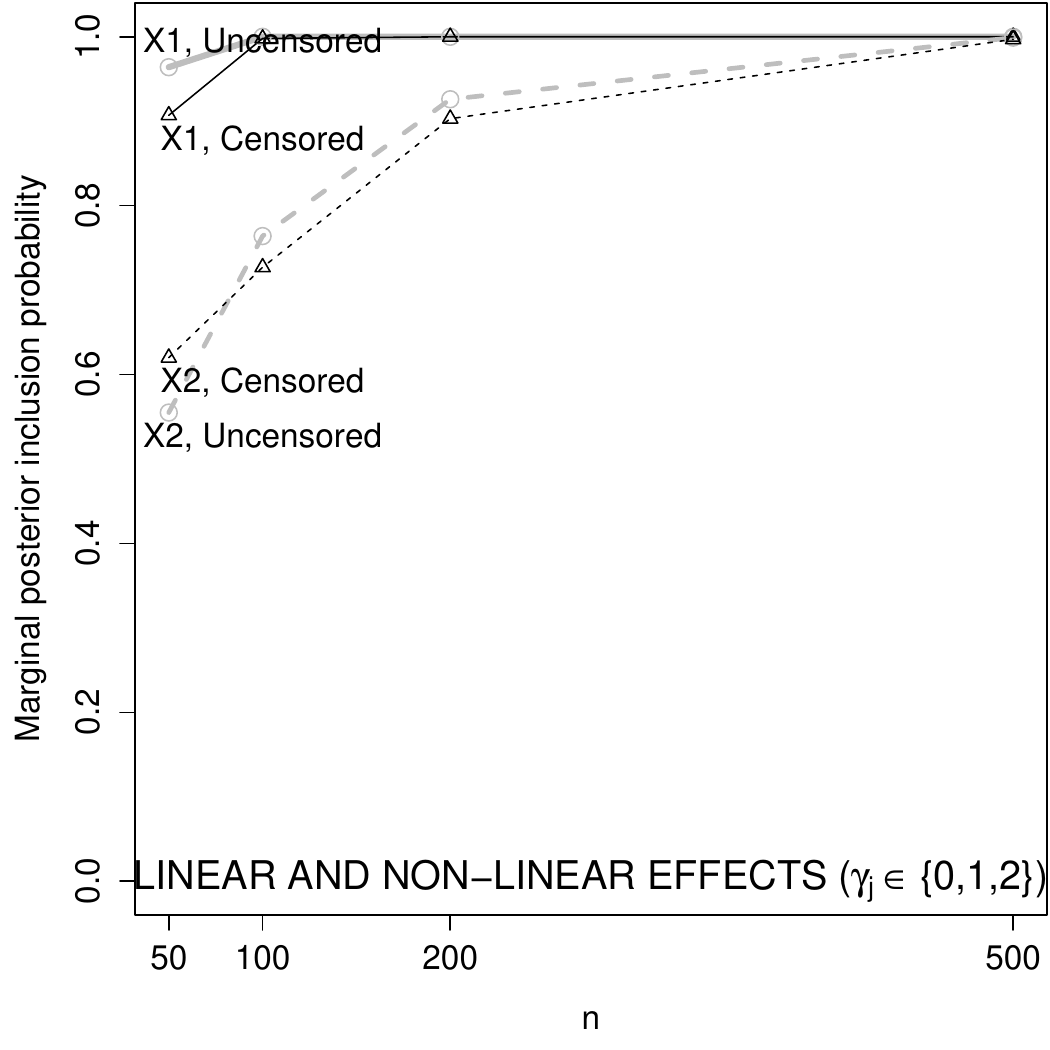} \\
      \includegraphics[width=0.5\textwidth,height=0.4\textwidth]{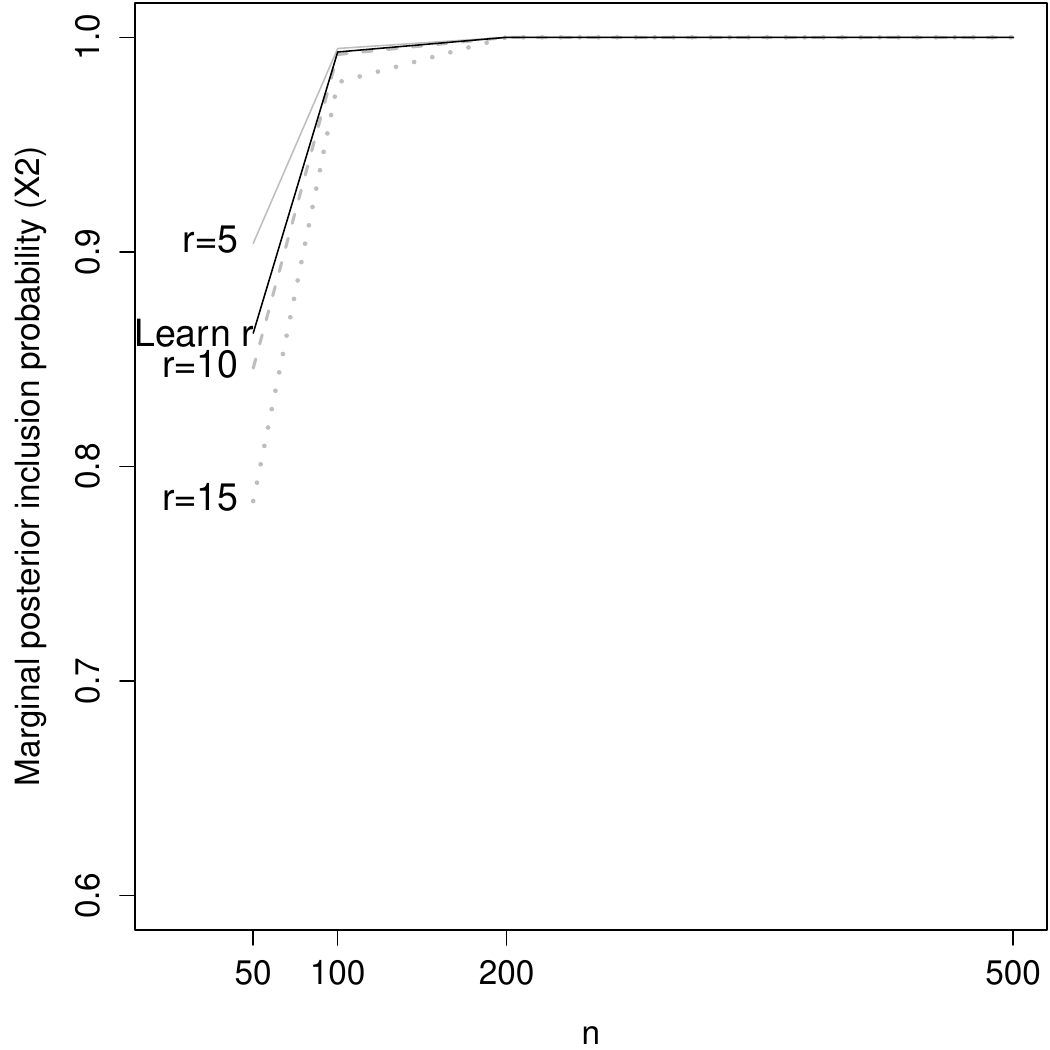} &
      \includegraphics[width=0.5\textwidth,height=0.4\textwidth]{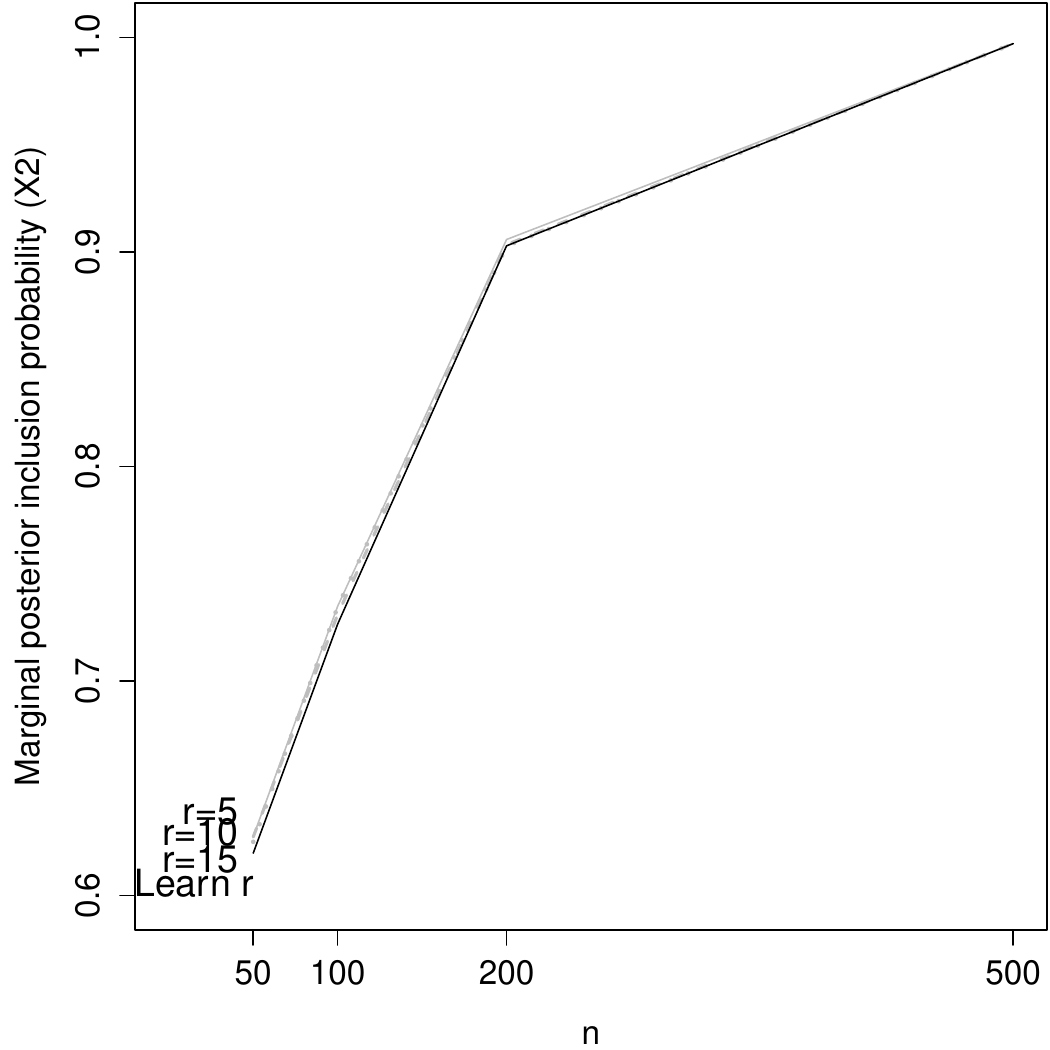} \\
    \end{tabular}
\end{center}
\caption{Scenarios 1-2. $p=2$. Mean posterior inclusion probabilities for AFT-pMOMZ considering only
non-linear (top) or both linear/non-linear (middle) effects. Spline dimension $r \in \{5,10,15\}$ learned from data.
Bottom: comparison of $\pi(\gamma_2 \geq 1 \mid y)$ to fixed $r=5,10,15$}
\label{fig:2vars_margpp}
\end{figure}

\begin{figure}[h!][h]
  \begin{center}
    \begin{tabular}{cc}
      Scenario 3 & Scenario 4 \\
      \includegraphics[width=0.5\textwidth,height=0.45\textwidth]{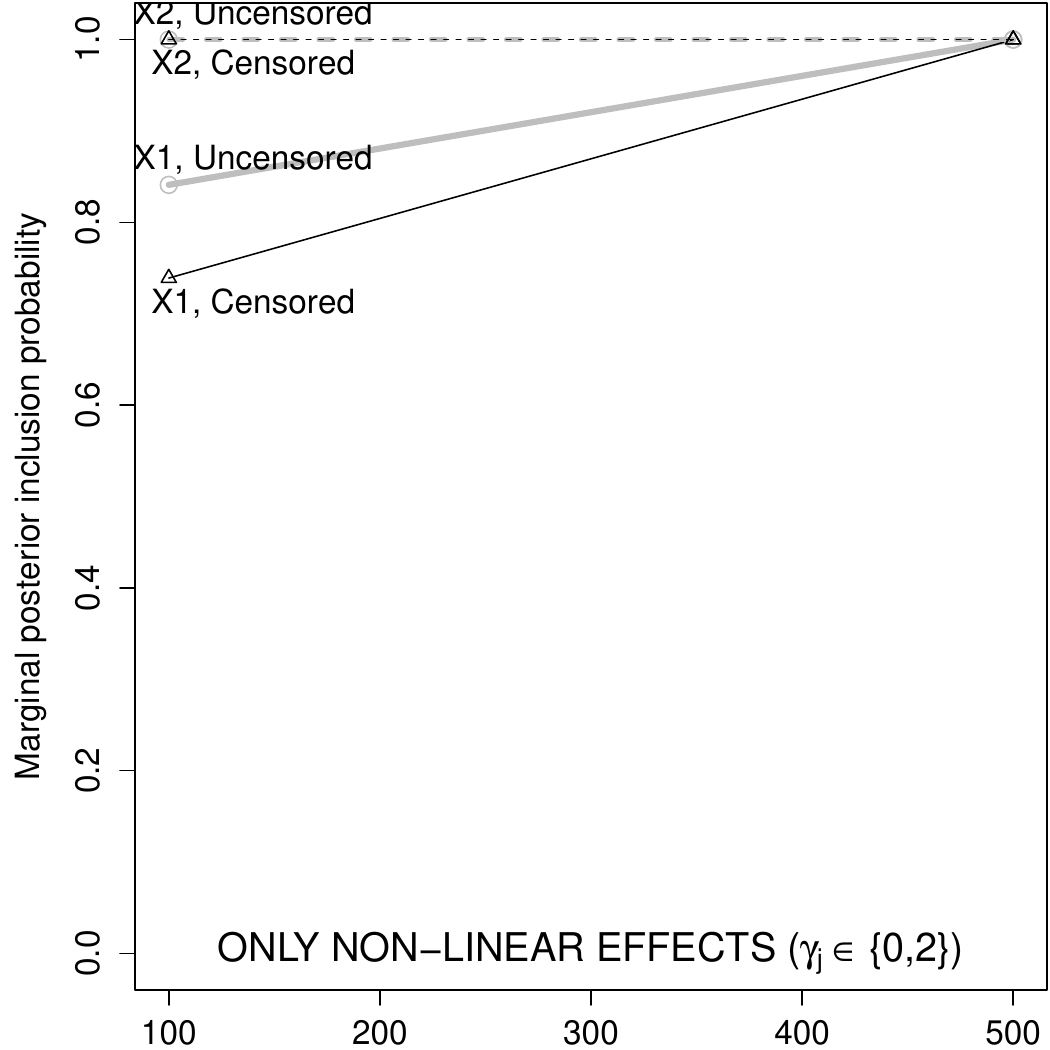} &
      \includegraphics[width=0.5\textwidth,height=0.45\textwidth]{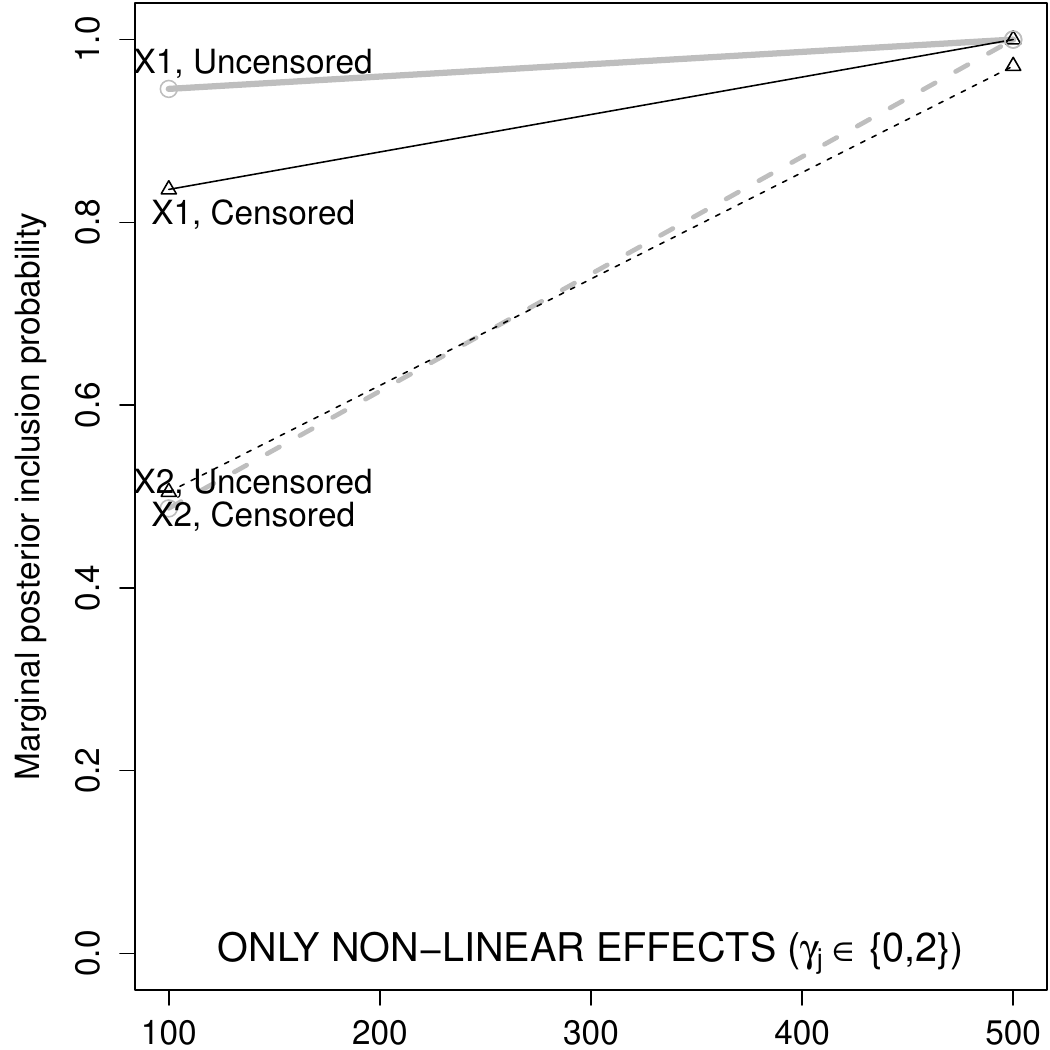} \\
      \includegraphics[width=0.5\textwidth]{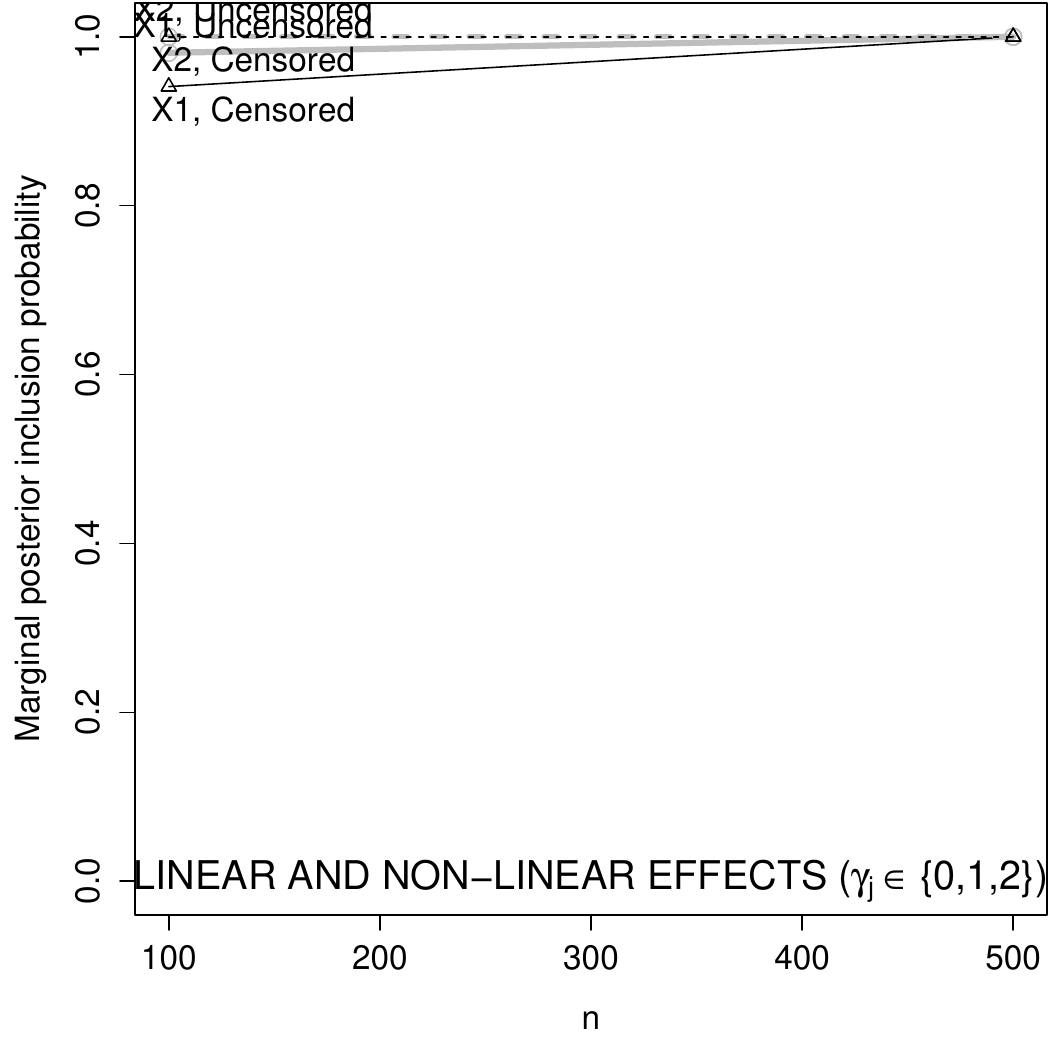} &
      \includegraphics[width=0.5\textwidth]{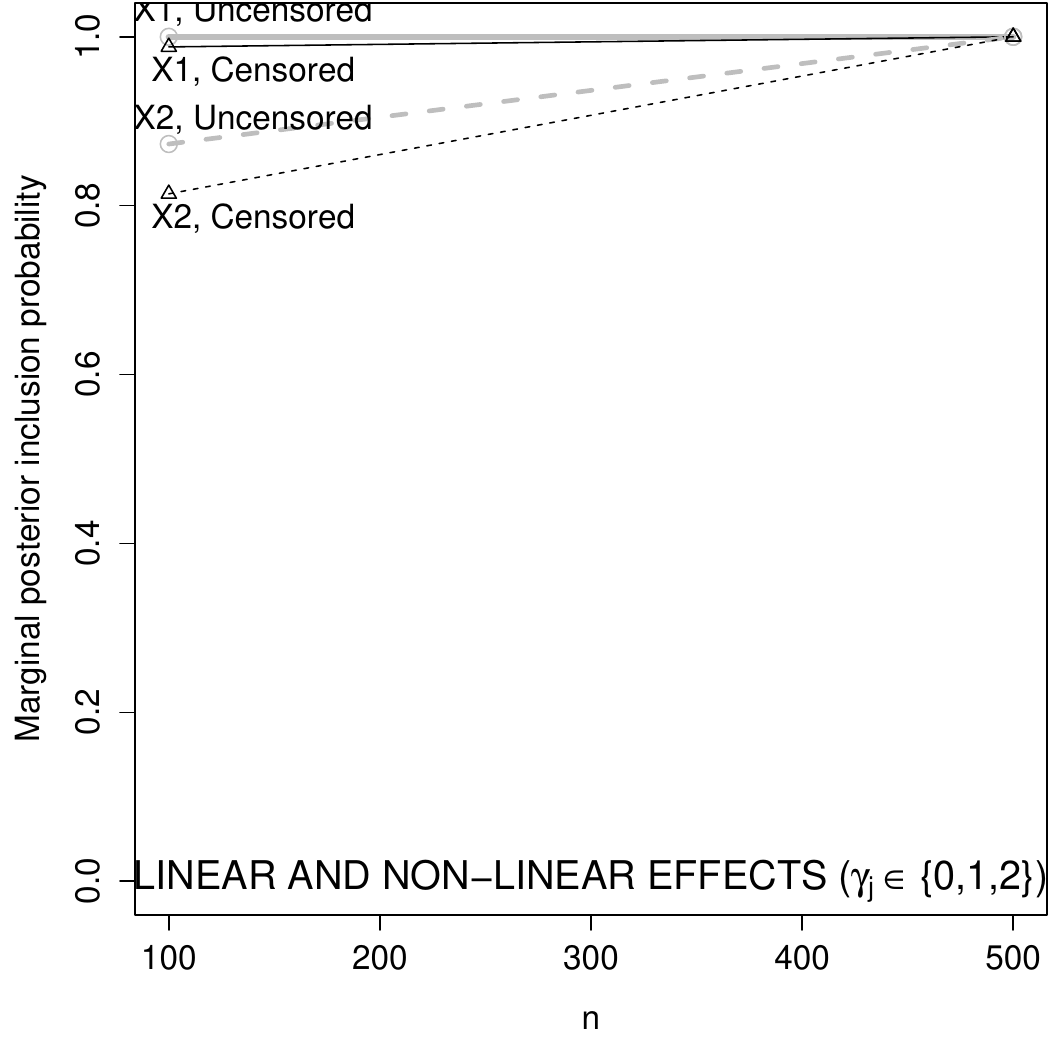}
    \end{tabular}
\end{center}
\caption{Scenarios 3-4. $p=2$. Truly GH $(\mu = 0,\sigma = 0.5)$ structure.
  Average marginal posterior inclusion probabilities under AFT-pMOMZ when considering 
only non-linear (top) or both linear/non-linear (bottom) effects.}
\label{fig:2vars_margpp_sc3}
\end{figure}

\begin{figure}[h!]
  \begin{center}
    \begin{tabular}{cc}
      Scenario 5 & Scenario 6 \\
      \includegraphics[width=0.5\textwidth,height=0.45\textwidth]{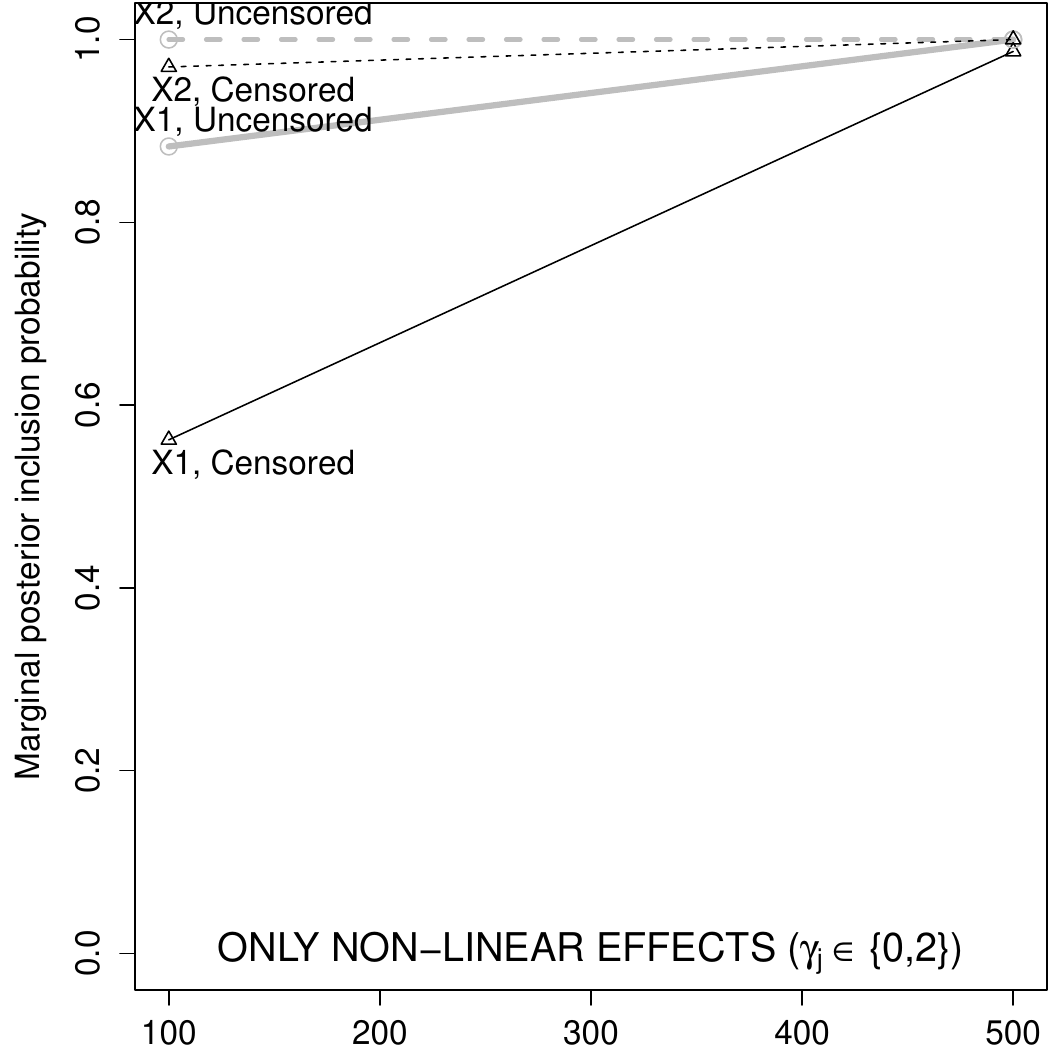} &
      \includegraphics[width=0.5\textwidth,height=0.45\textwidth]{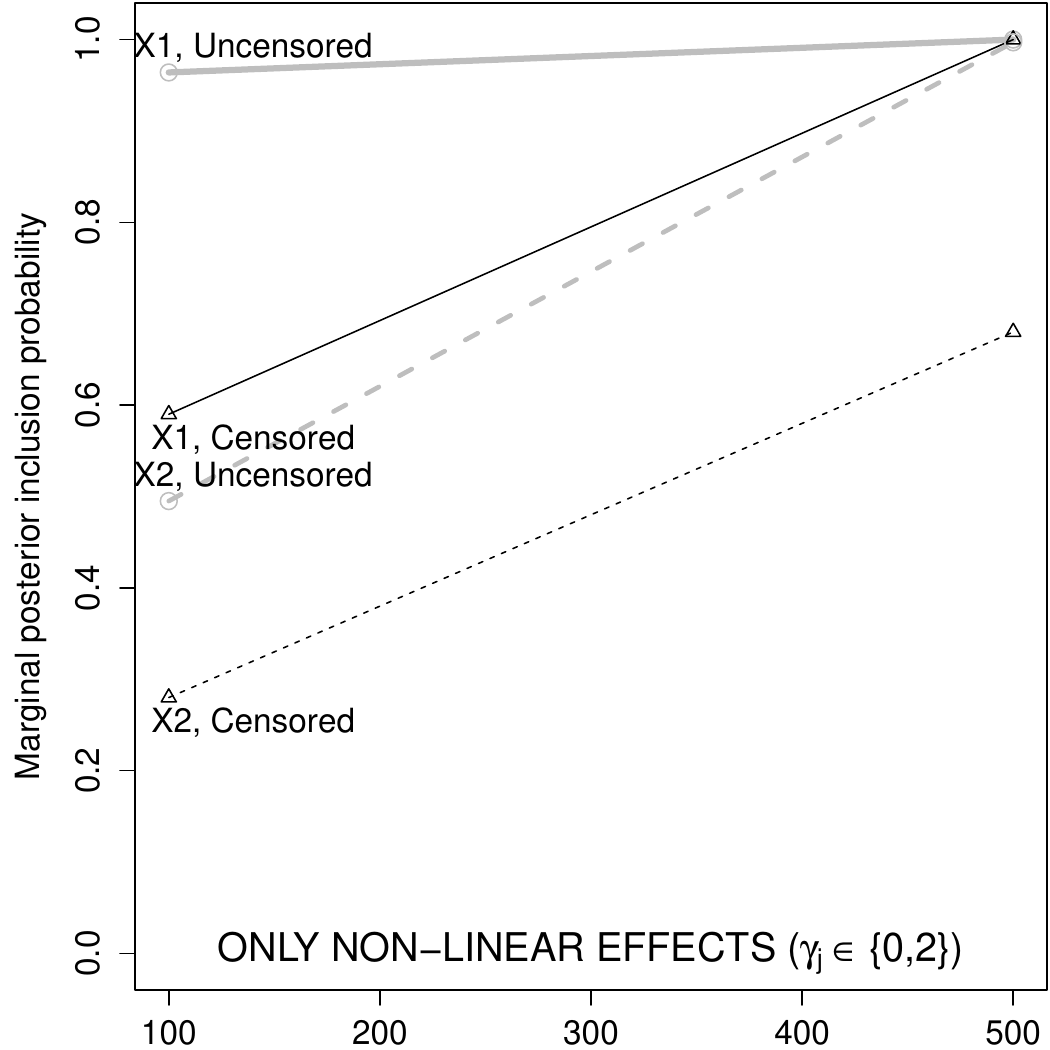} \\
      \includegraphics[width=0.5\textwidth]{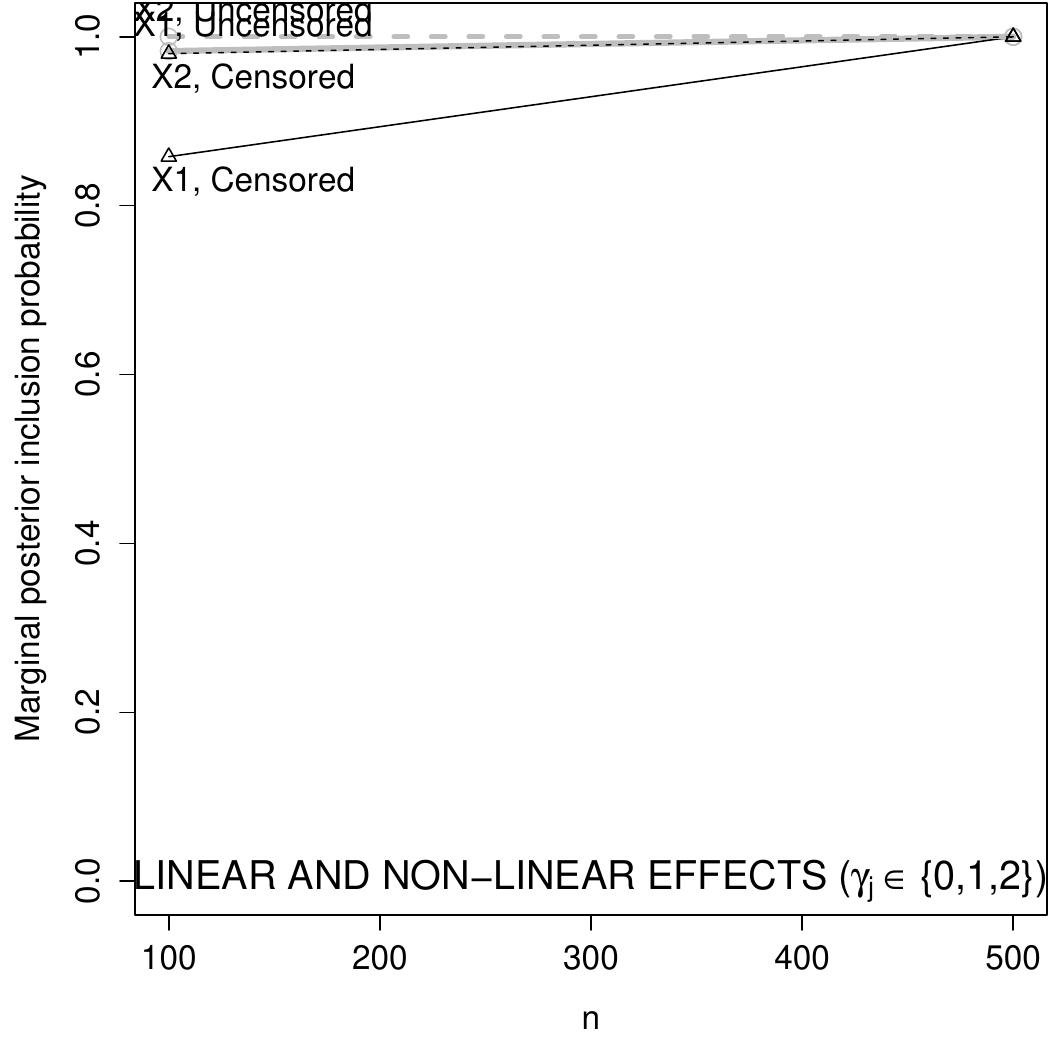} &
      \includegraphics[width=0.5\textwidth]{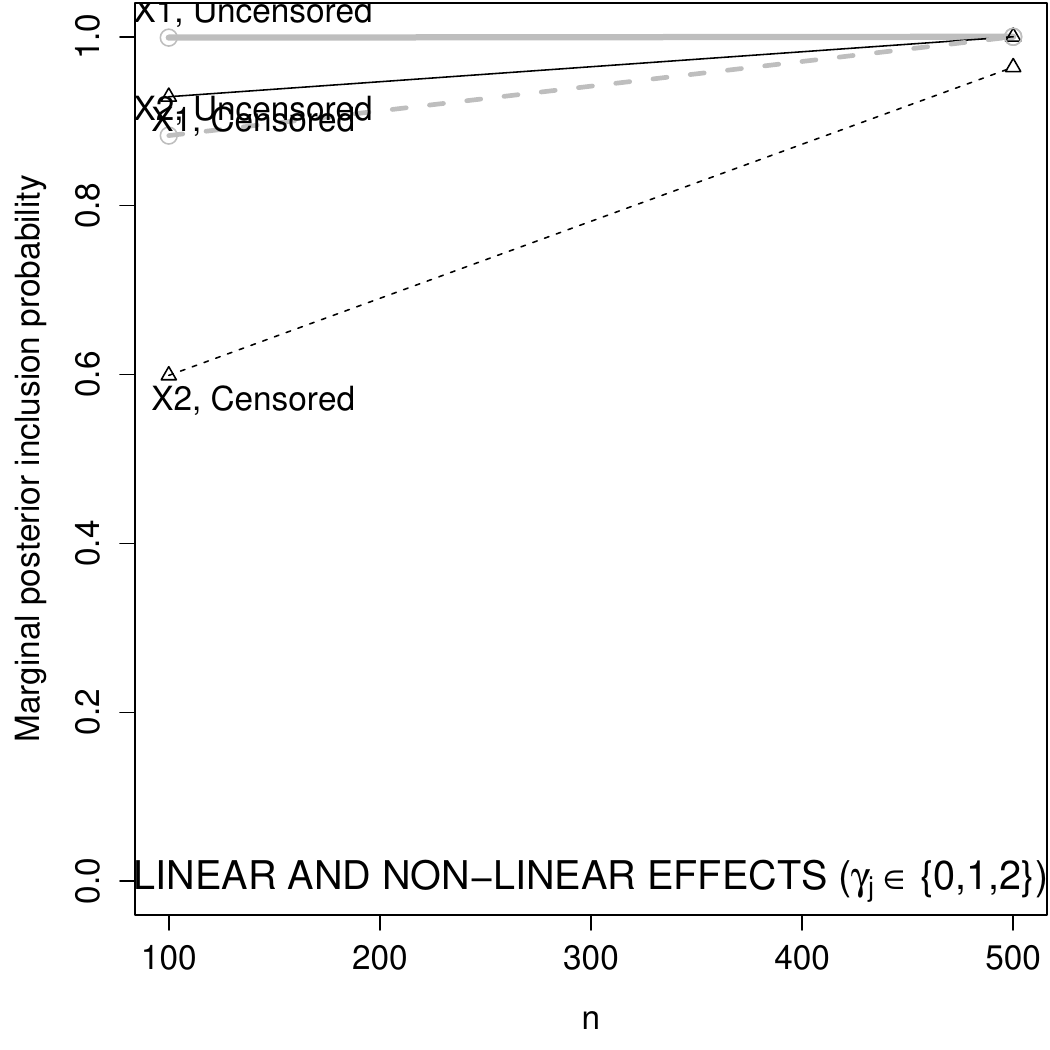}
    \end{tabular}
\end{center}
\caption{Scenarios 5-6. $p=2$. Truly Proportional Hazards $(\mu=0,\sigma = 0.8)$ structure.
  Average marginal posterior inclusion probabilities  under AFT-pMOMZ when considering 
only non-linear (top) or both linear/non-linear (bottom) effects}
\label{fig:2vars_margpp_sc4}
\end{figure}

\begin{figure}[h!][h]
  \begin{center}
    \begin{tabular}{cc}
      Scenario 1 & Scenario 2 \\
      \includegraphics[width=0.5\textwidth,height=0.45\textwidth]{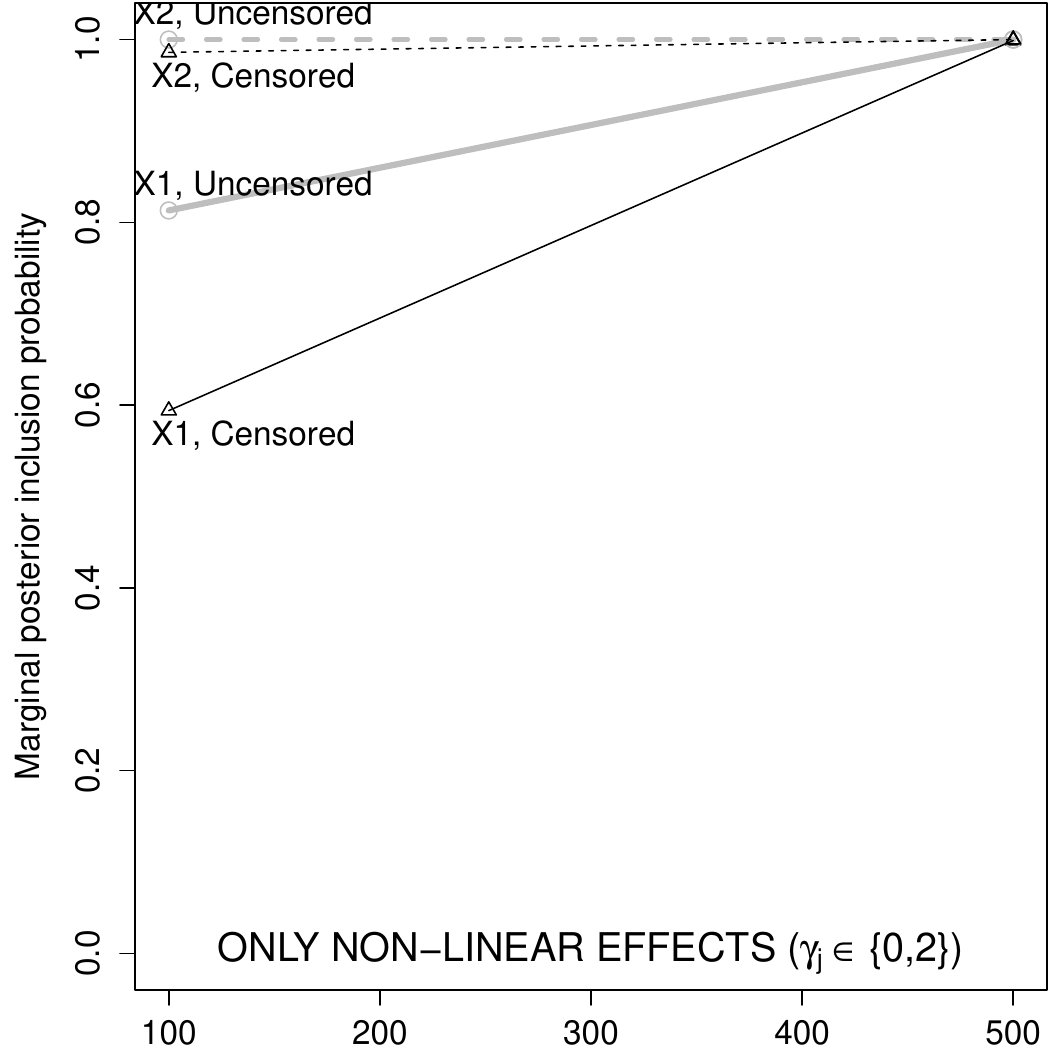} &
      \includegraphics[width=0.5\textwidth,height=0.45\textwidth]{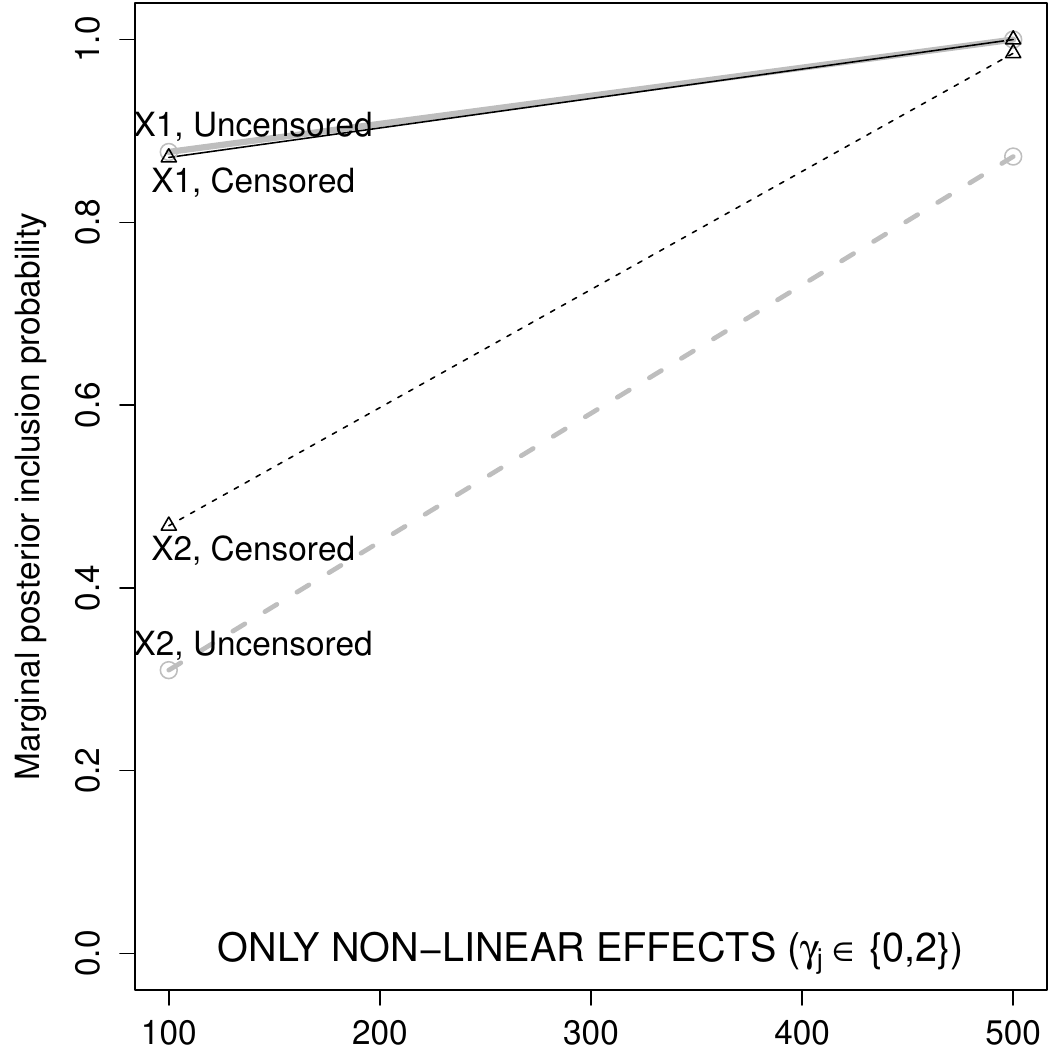} \\
      \includegraphics[width=0.5\textwidth]{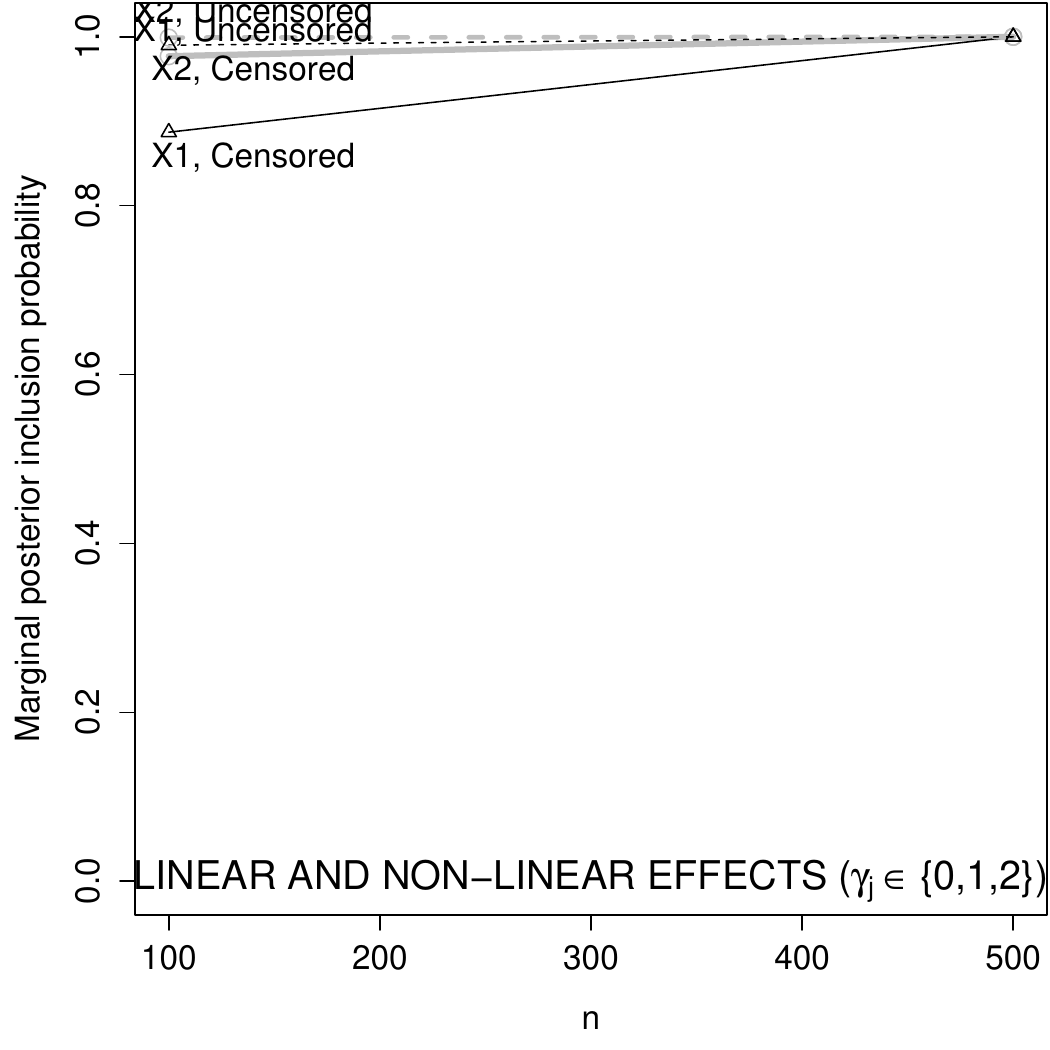} &
      \includegraphics[width=0.5\textwidth]{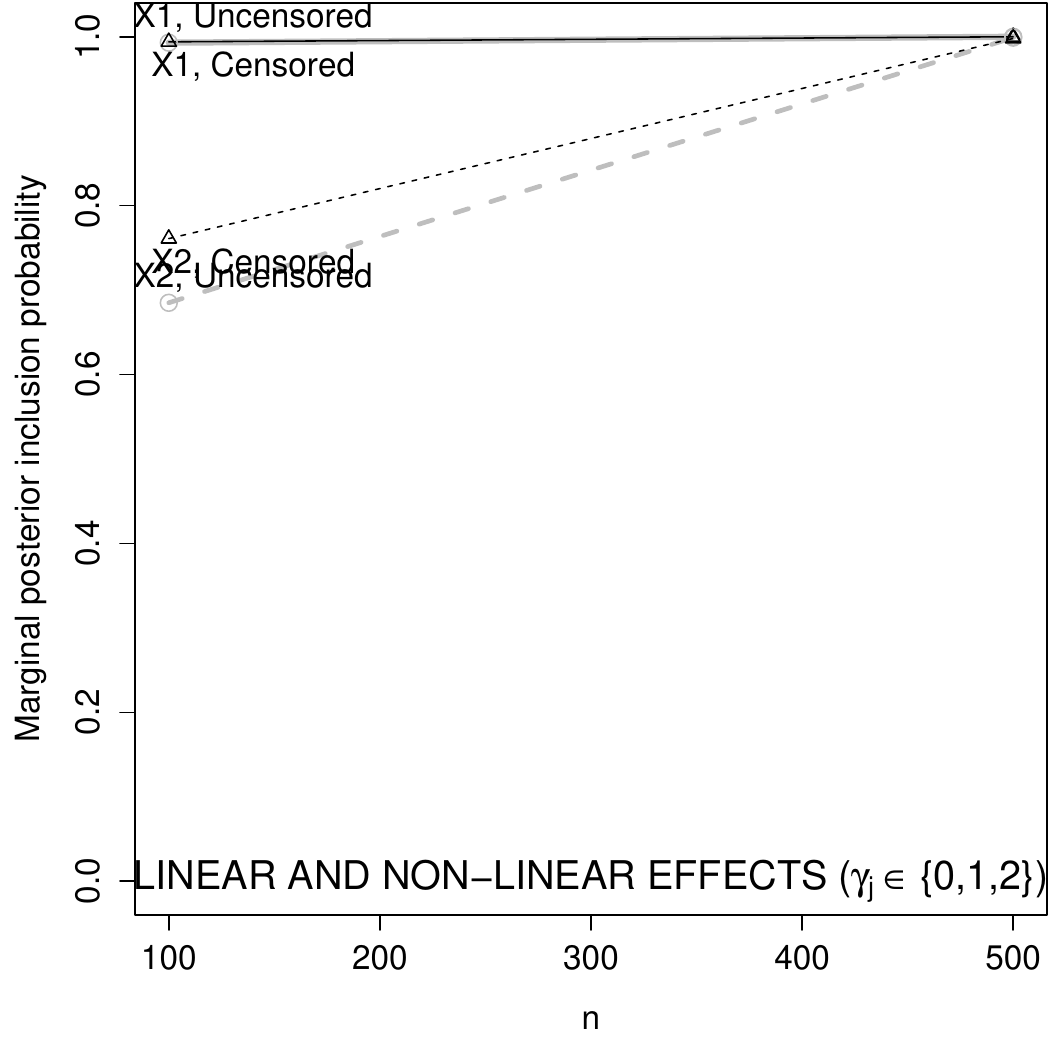}
    \end{tabular}
\end{center}
\caption{Scenarios 1-2. $p=2$. Truly $\epsilon_i \sim \mbox{ALaplace}(0,s=0.1,a=-0.5)$ errors.
  Average marginal posterior inclusion probabilities  under AFT-pMOMZ when considering 
only non-linear (top) or both linear/non-linear (bottom) effects}
\label{fig:2vars_margpp_alapl}
\end{figure}

\begin{figure}[h!][h]
  \begin{center}
    \begin{tabular}{cc}
      Scenario 1 & Scenario 2 \\
      \includegraphics[width=0.5\textwidth,height=0.45\textwidth]{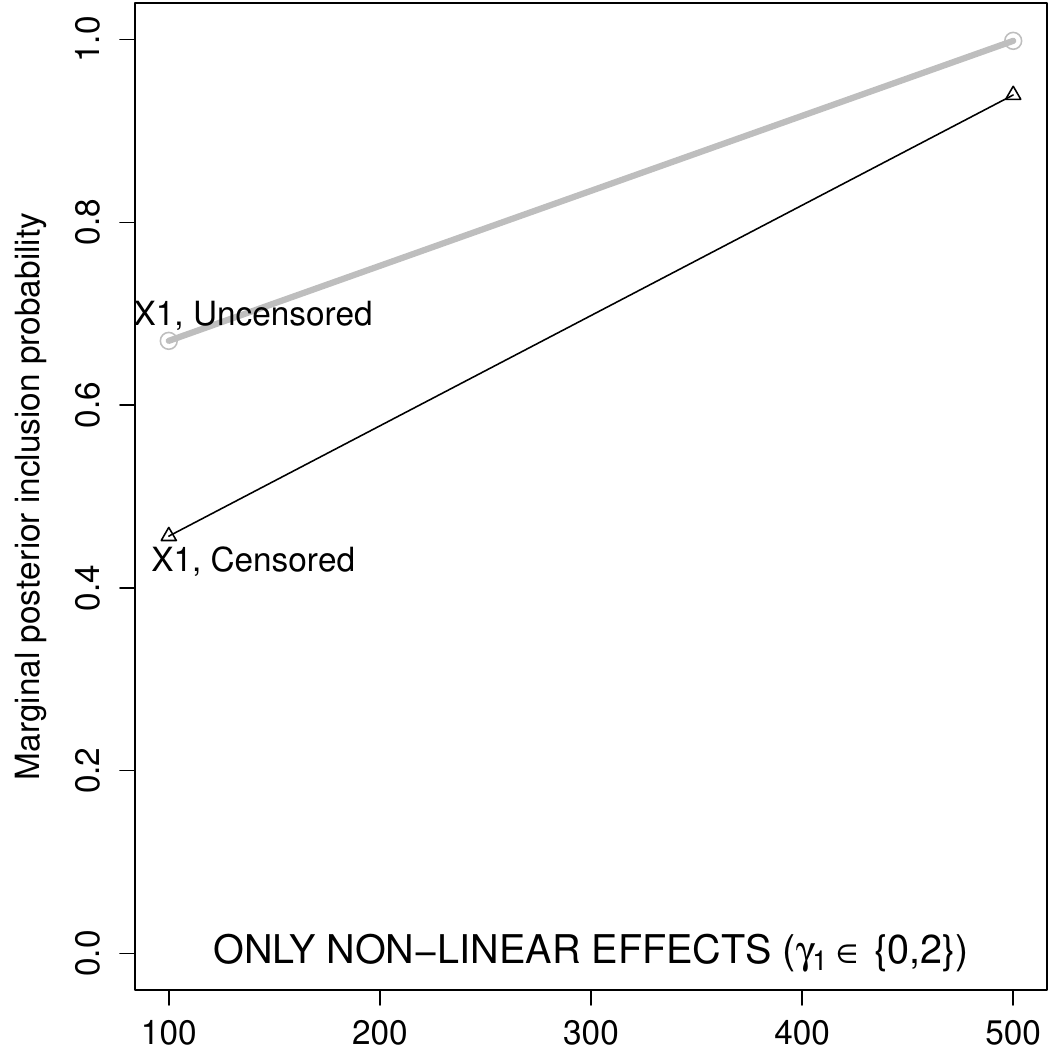} &
      \includegraphics[width=0.5\textwidth,height=0.45\textwidth]{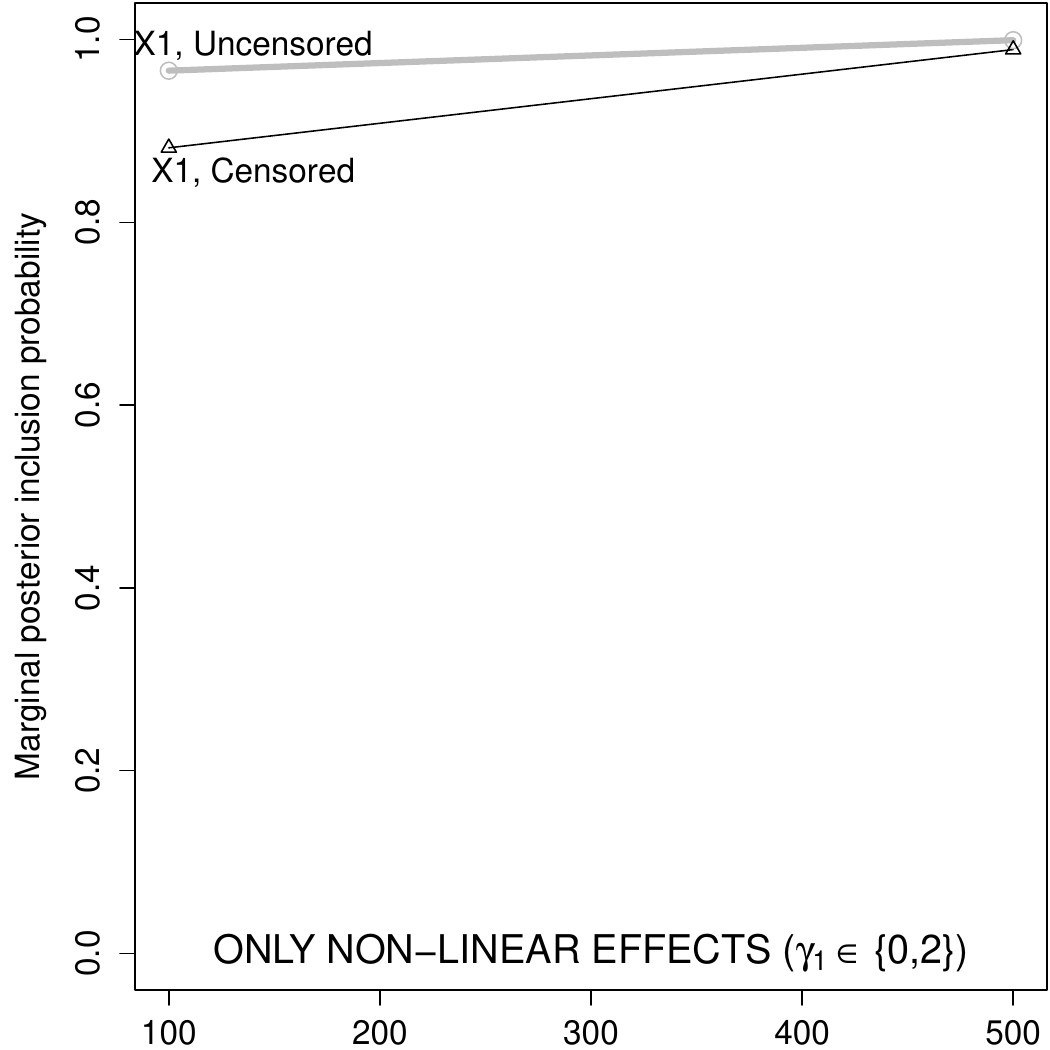} \\
      \includegraphics[width=0.5\textwidth]{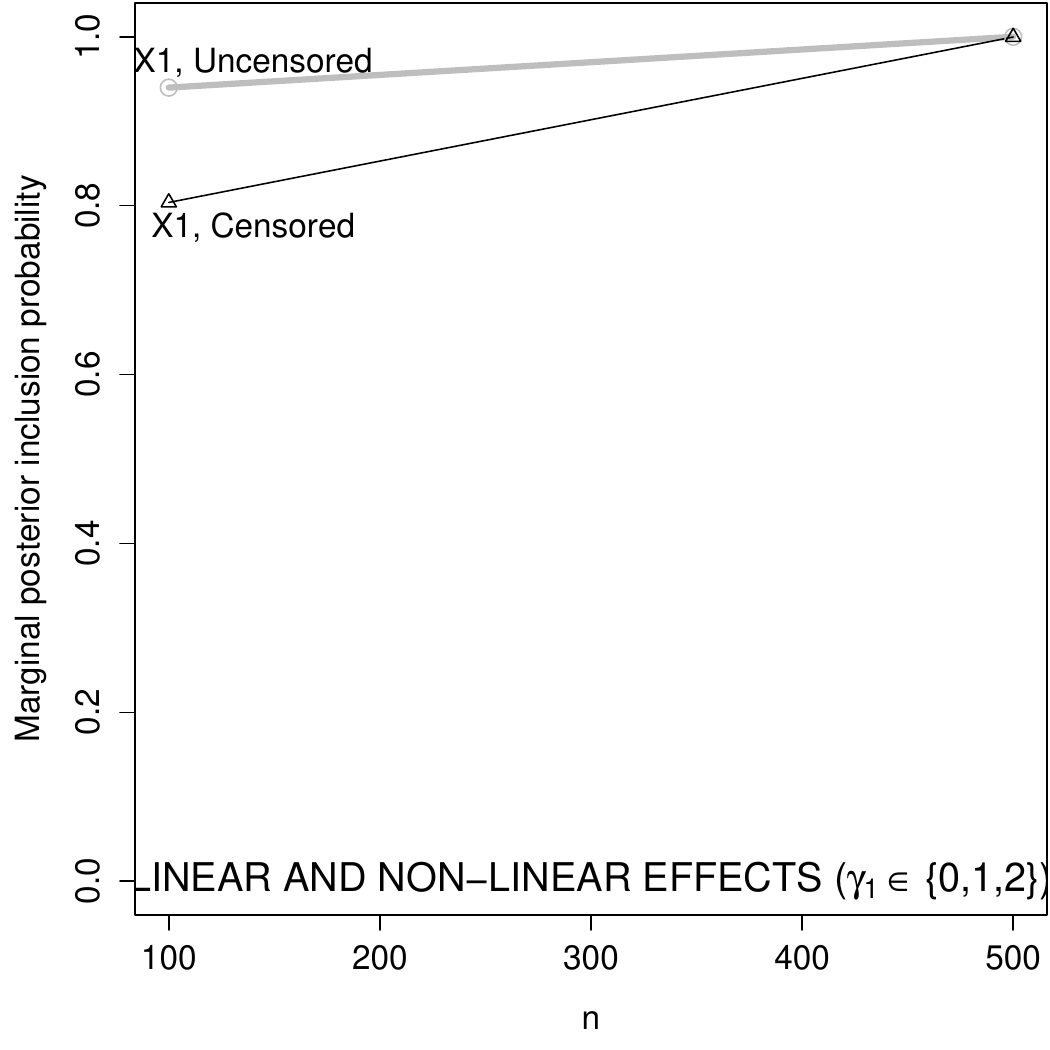} &
      \includegraphics[width=0.5\textwidth]{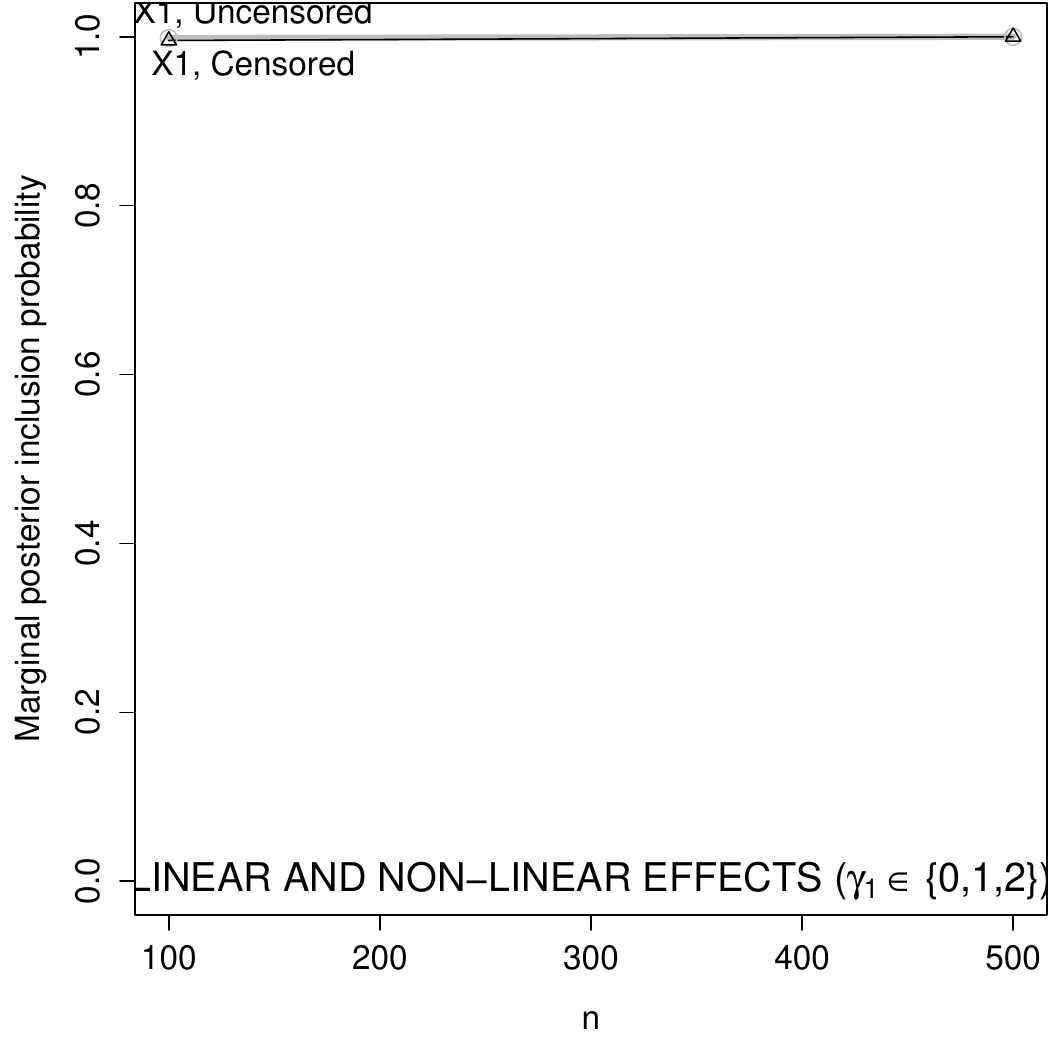}
    \end{tabular}
\end{center}
\caption{Scenarios 1-2. $p=2$ with omitted $x_{i2}$.
Average marginal posterior inclusion probabilities  under AFT-pMOMZ when only considering 
non-linear (top) or both linear and non-linear (bottom) effects}
\label{fig:2vars_margpp_omitX2}
\end{figure}


\begin{figure}[h!]
  \begin{center}
    \begin{tabular}{cc}
      \multicolumn{2}{c}{Scenario 3} \\
      \includegraphics[width=0.4\textwidth]{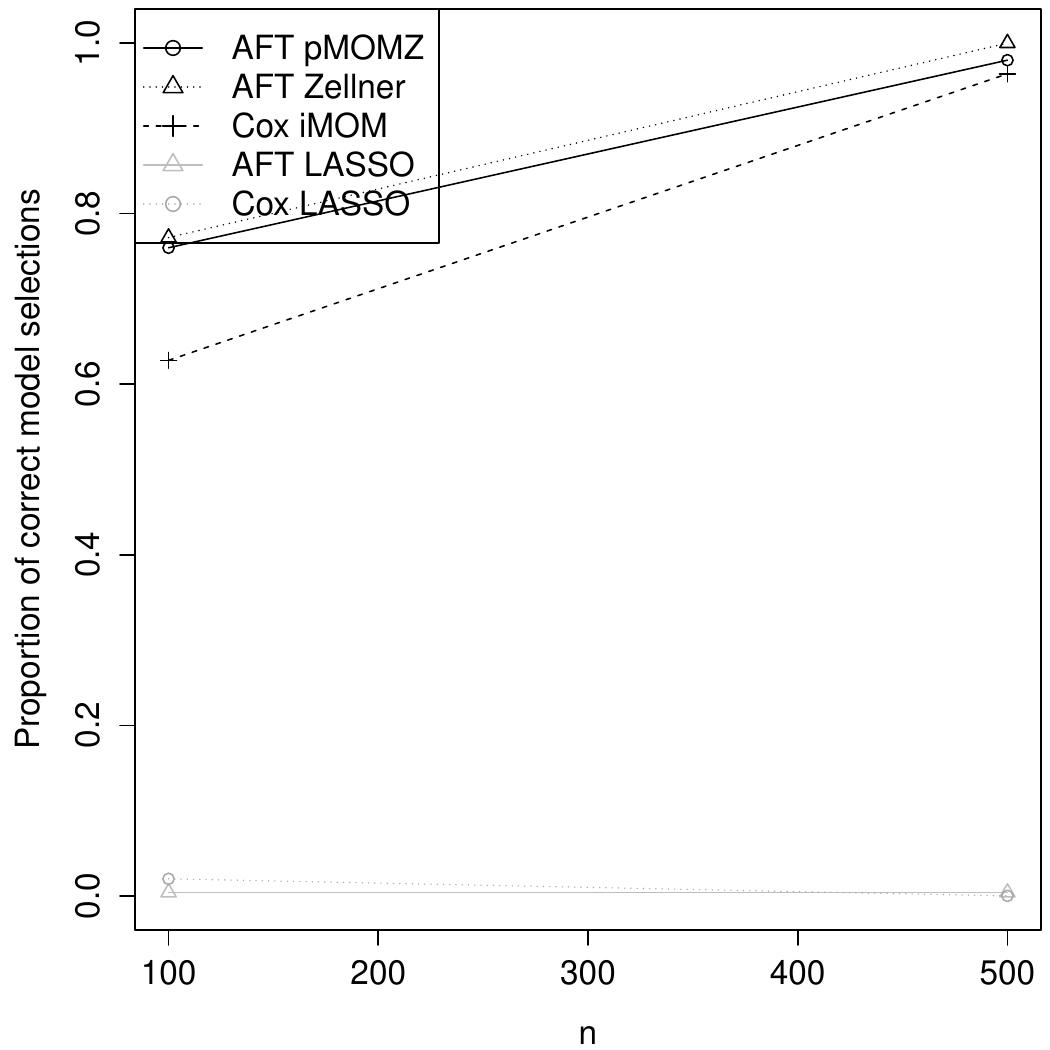} &
      \includegraphics[width=0.4\textwidth]{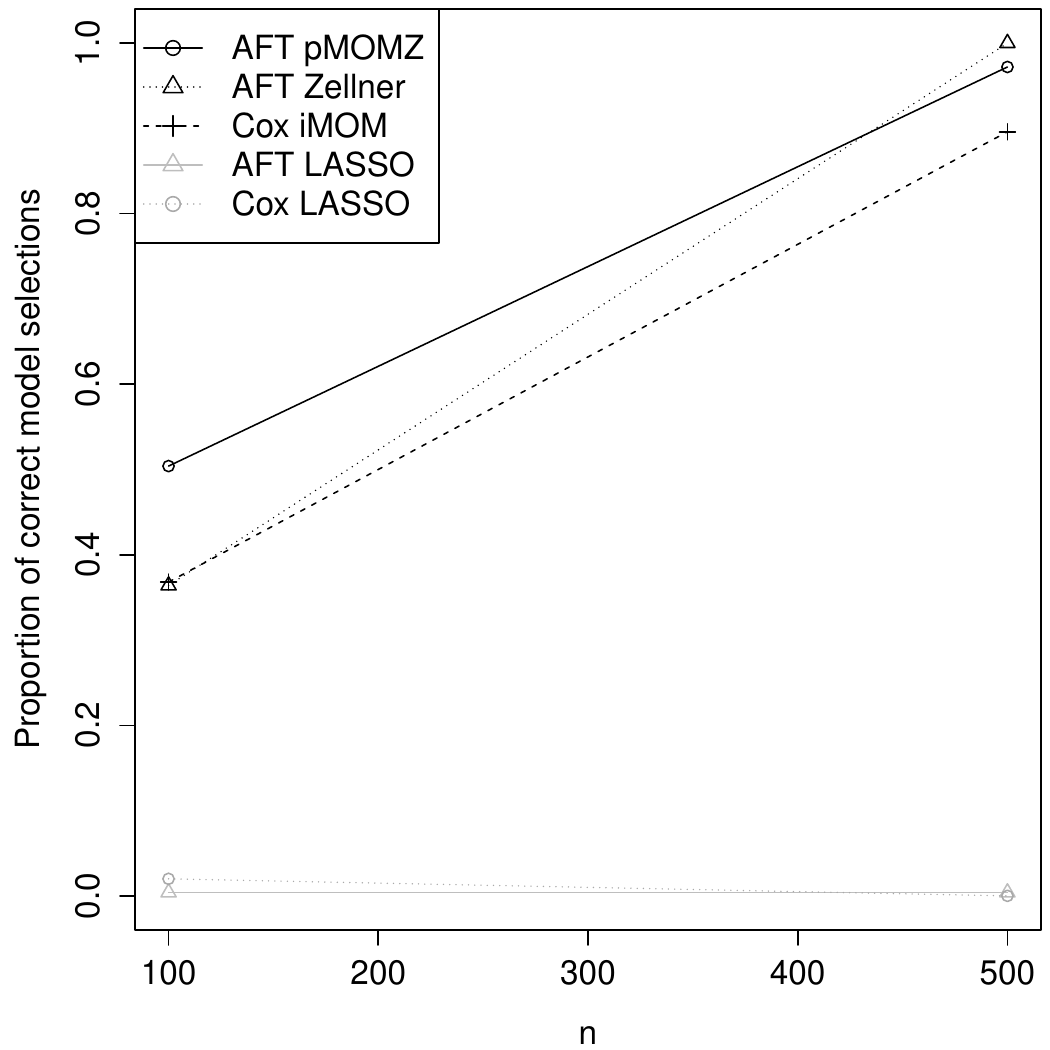} \\
      \multicolumn{2}{c}{Scenario 4} \\
      \includegraphics[width=0.4\textwidth]{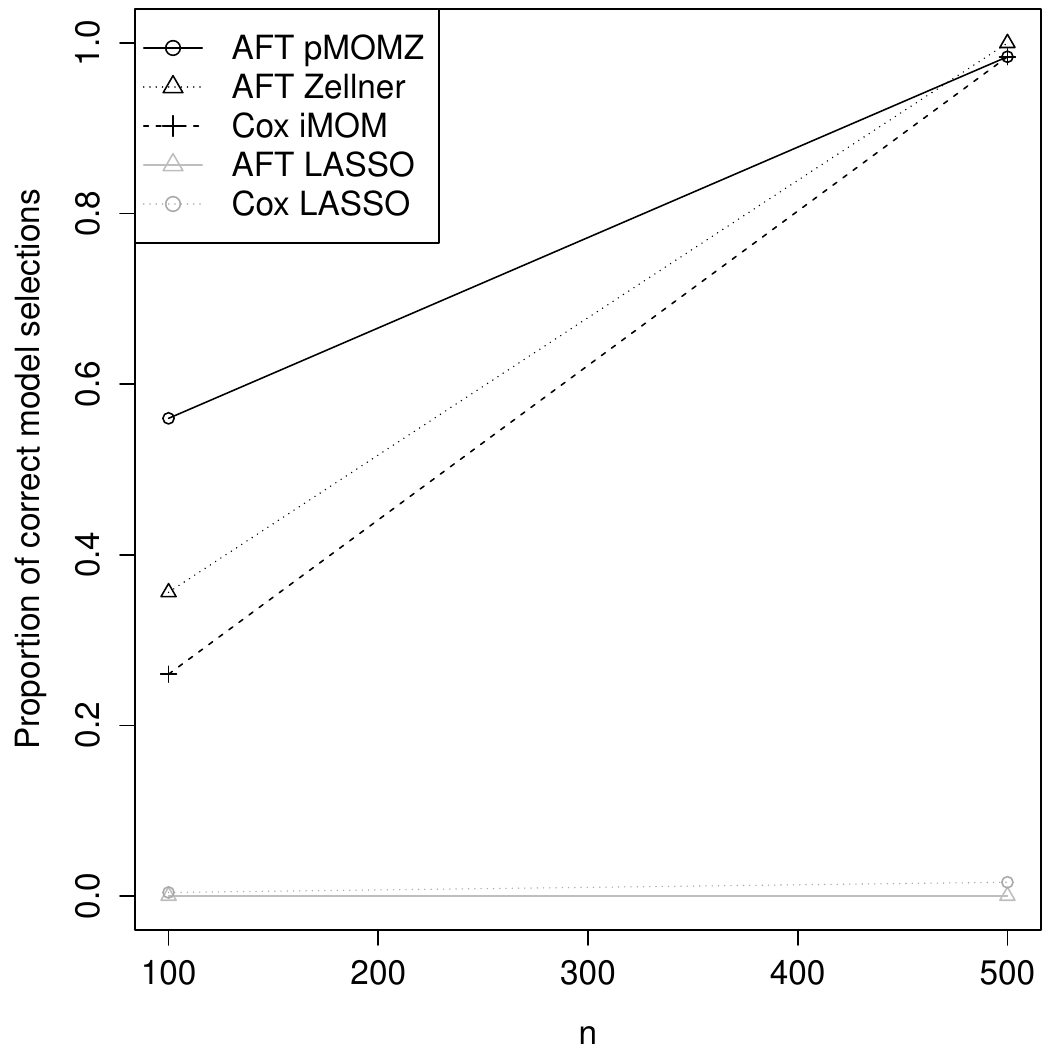} &
      \includegraphics[width=0.4\textwidth]{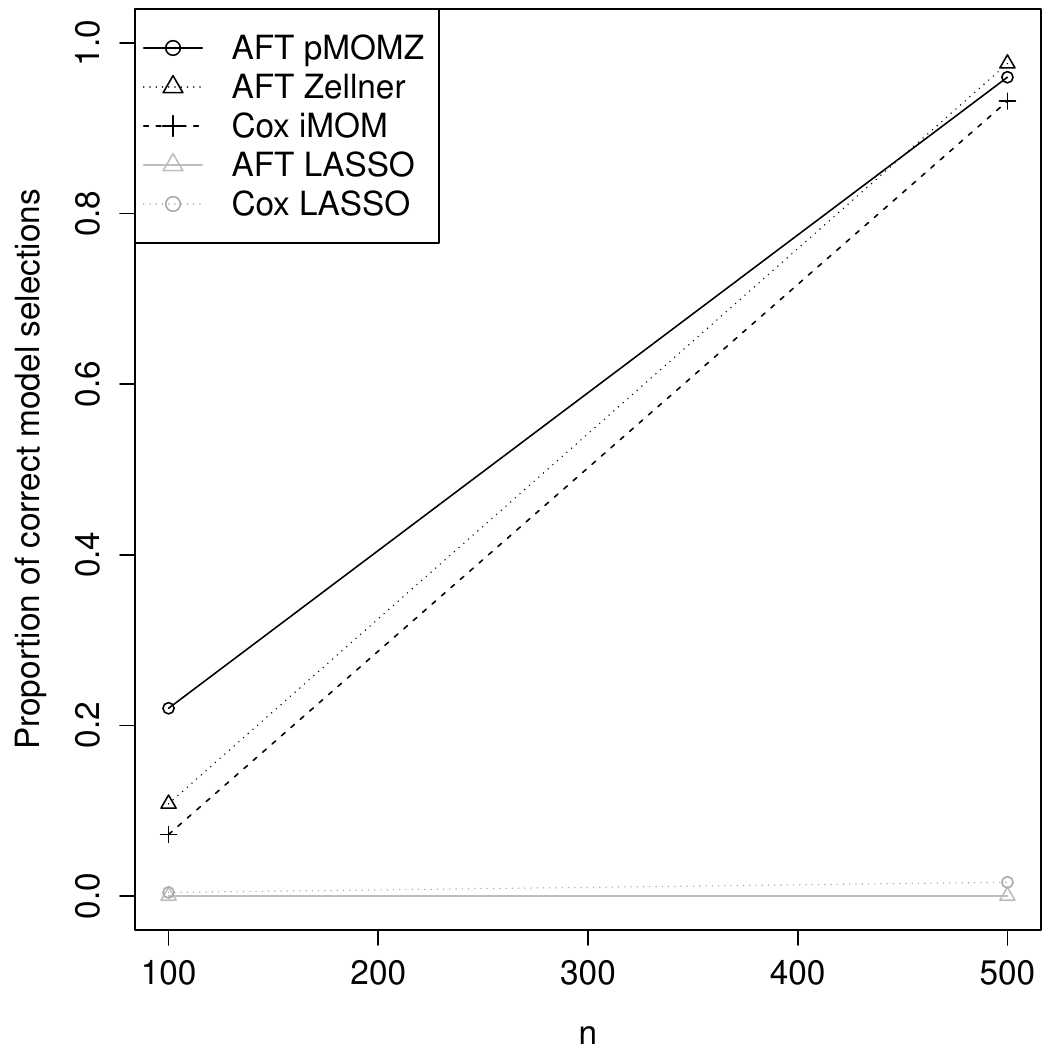}
    \end{tabular}
\end{center}
\caption{Scenarios 3-4, $p=50$. Correct model selection proportion in uncensored (left) and censored (right) data}
\label{fig:pcorrect_scen3}
\end{figure}

\begin{figure}[h!]
  \begin{center}
    \begin{tabular}{cc}
      \multicolumn{2}{c}{Scenario 5} \\
      \includegraphics[width=0.4\textwidth]{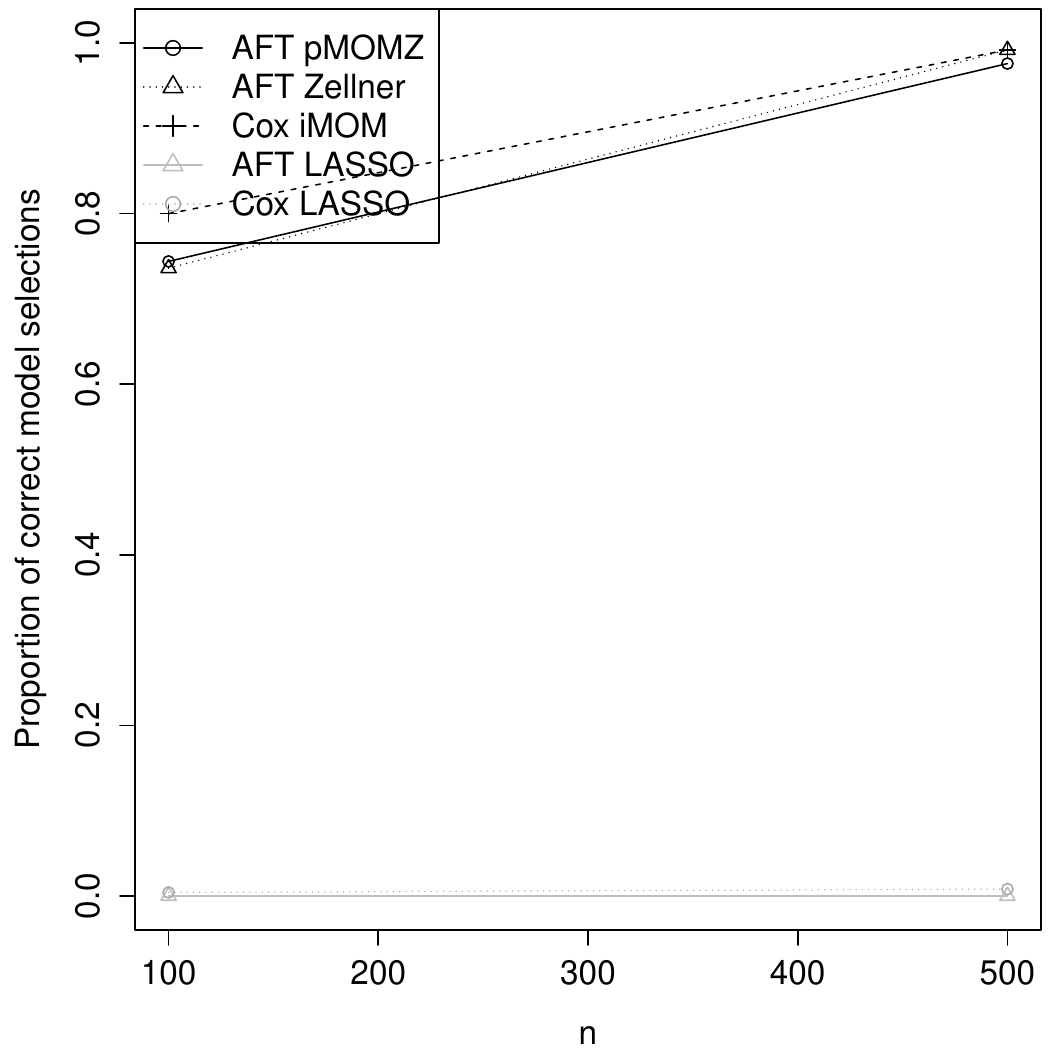} &
      \includegraphics[width=0.4\textwidth]{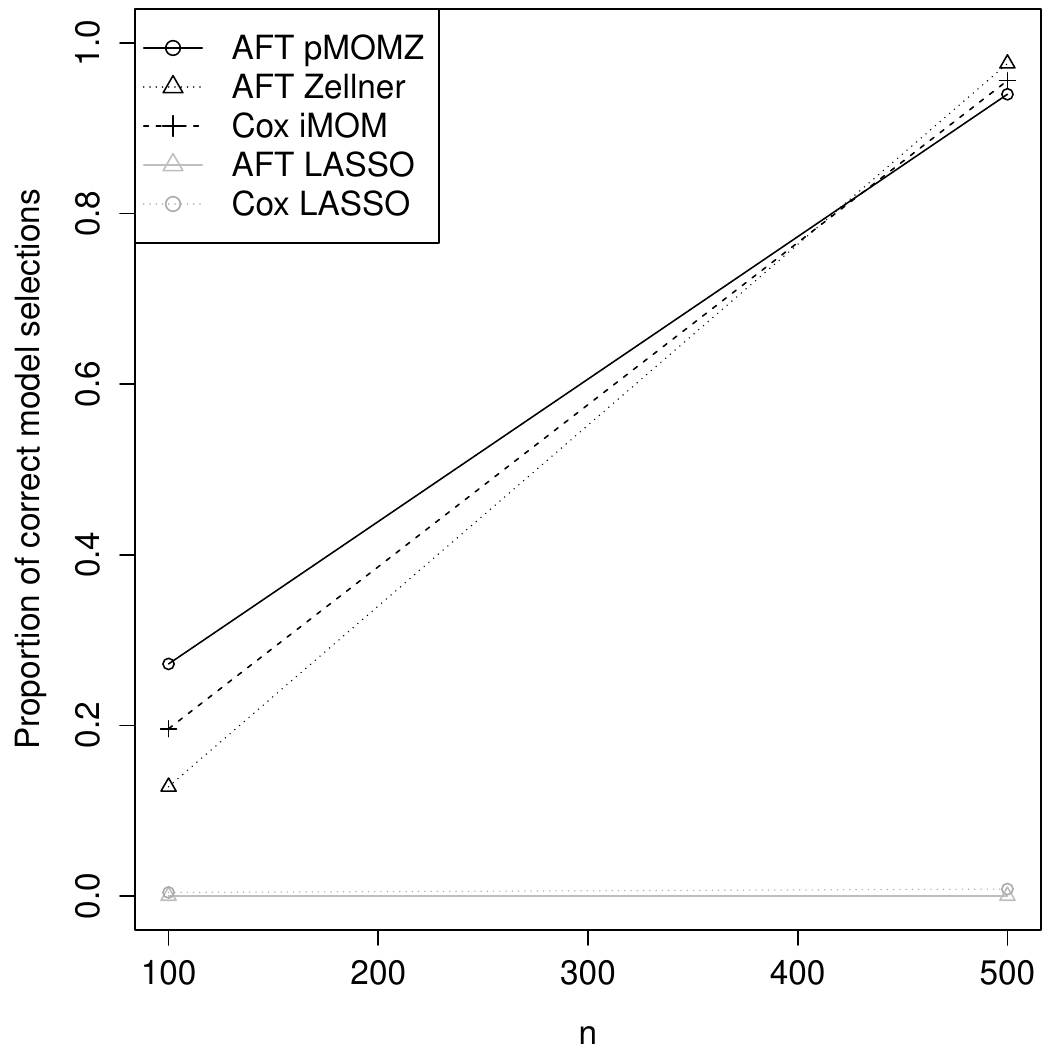} \\
      \multicolumn{2}{c}{Scenario 6} \\
      \includegraphics[width=0.4\textwidth]{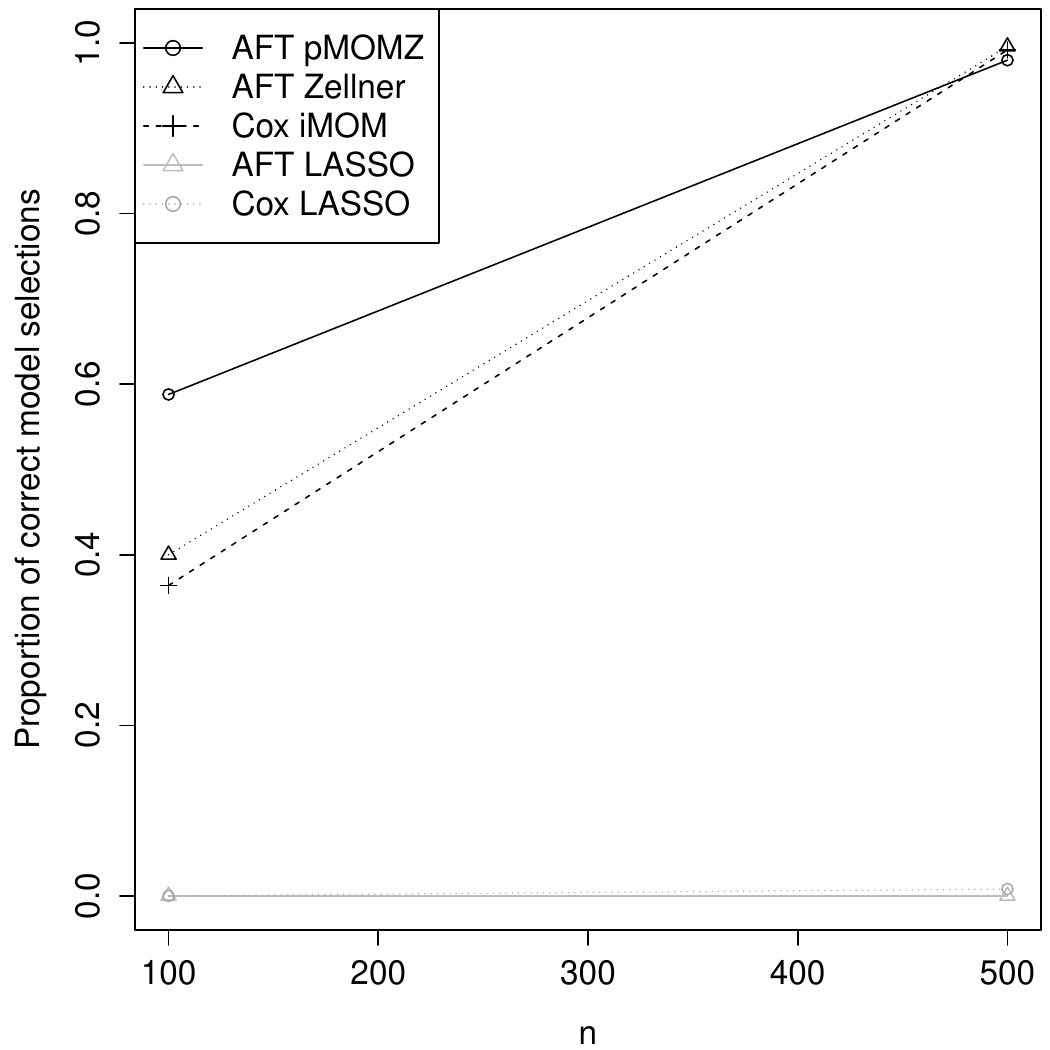} &
      \includegraphics[width=0.4\textwidth]{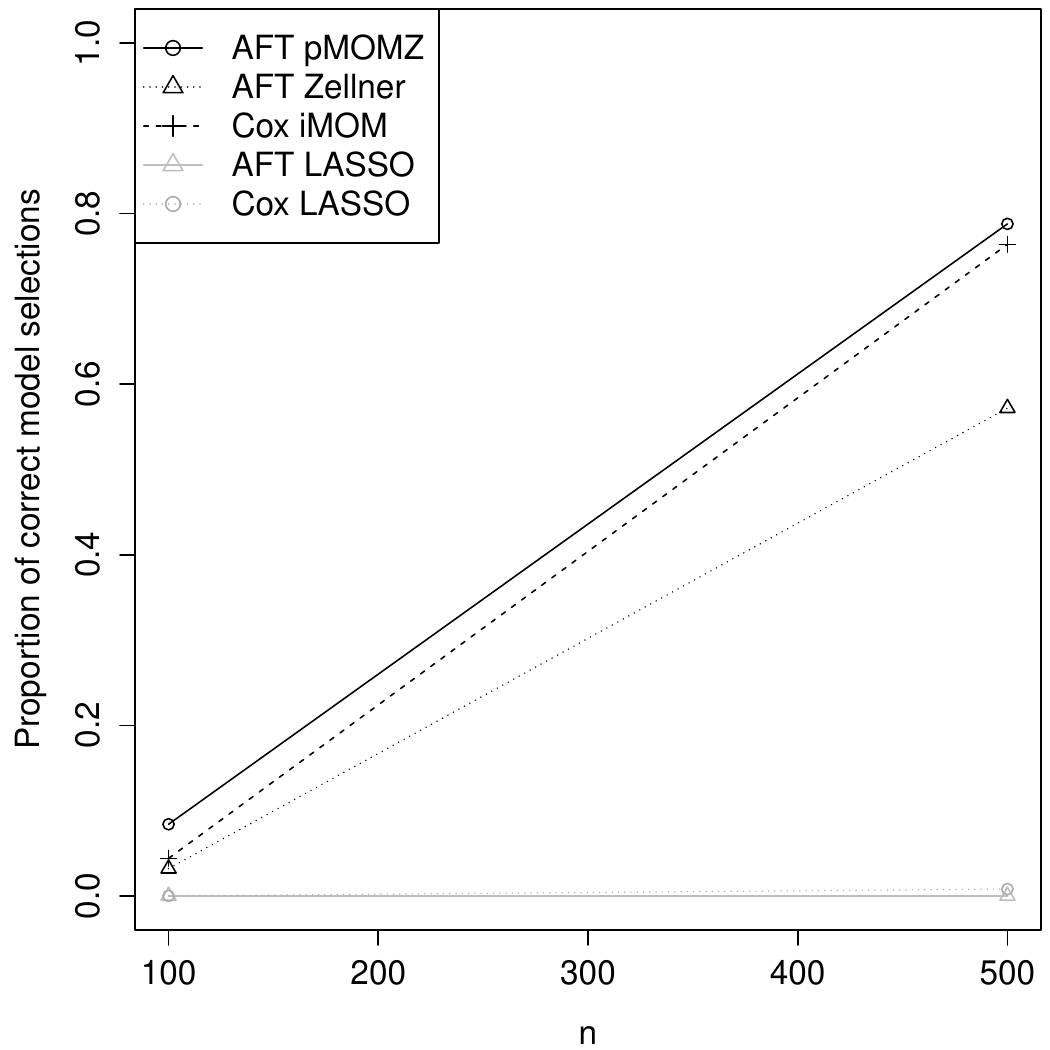}
    \end{tabular}
\end{center}
\caption{Scenarios 5-6, $p=50$. Correct model selection proportion in uncensored (left) and censored (right) data}
\label{fig:pcorrect_scen4}
\end{figure}

\begin{table}
\begin{center}
\begin{tabular}{lrrrrr}  \hline
  Method & $n$ & $\widehat{\gamma}=\gamma^*$ & $\pi(\gamma^* \mid y,\widehat{\gamma}=\gamma^*)$ & Truly active selected & Truly inactive selected \\  \hline
   AFT-pMOMZ   & 100 & 0.380 & 0.331 & 1.34 & 0.37 \\
   AFT-pMOMZ   & 500 & 0.960 & 0.824 & 2.00 & 0.04 \\
   AFT-Zellner & 100 & 0.284 & 0.651 & 1.03 & 0.09 \\
   AFT-Zellner & 500 & 1.000 & 0.957 & 2.00 & 0.00 \\
   Cox-piMOM & 100 & 0.328 & 0.379 & 1.37 & 0.36 \\
   Cox-piMOM & 500 & 0.888 & 0.717 & 2.00 & 0.12 \\
   AFT-LASSO & 100 & 0.004 & - & 0.89 & 5.90 \\
   AFT-LASSO & 500 & 0.008 & - & 1.72 & 15.96 \\
   Cox-LASSO & 100 & 0.028 & -  & 1.78 & 8.24 \\
  Cox-LASSO & 500 & 0.004 & - & 2.00 & 14.06 \\
   \hline
\end{tabular}
\end{center}
\caption{Scenario 1, $p=50$. Proportion of correctly selected models ($\widehat{\gamma}=\gamma^*$) for uncensored and censored data,
  average posterior model probability when the correct model was selected,
number of truly active and truly inactive variables $\widehat{\gamma}$}
\label{tab:pcorrect_scen1}
\end{table}

\begin{table}
\begin{center}
\begin{tabular}{lrrrrr}  \hline
  Method & $n$ & $\widehat{\gamma}=\gamma^*$ & $\pi(\gamma^* \mid y,\widehat{\gamma}=\gamma^*)$ & Truly active selected & Truly inactive selected \\  \hline
  AFT-pMOMZ & 100 & 0.192 & 0.376 & 1.12 & 0.27 \\
  AFT-pMOMZ & 500 & 0.936 & 0.854 & 1.96 & 0.03 \\
  AFT-Zellner & 100 & 0.096 & 0.612 & 0.86 & 0.08 \\
  AFT-Zellner & 500 & 0.872 & 0.925 & 1.87 & 0.00 \\
   Cox-piMOM & 100 & 0.060 & 0.333 & 0.92 & 0.29 \\
   Cox-piMOM & 500 & 0.876 & 0.851 & 1.90 & 0.04 \\
   AFT-LASSO & 100 & 0.000 &  & 0.45 & 5.28 \\
   AFT-LASSO & 500 & 0.000 &  & 1.49 & 13.93 \\
   Cox-LASSO & 100 & 0.012 &  & 1.22 & 5.72 \\
   Cox-LASSO & 500 & 0.012 &  & 1.99 & 10.42 \\
   \hline
\end{tabular}
\end{center}
\caption{Scenario 2, $p=50$. Proportion of correctly selected models ($\widehat{\gamma}=\gamma^*$) for uncensored and censored data,
  average posterior model probability when the correct model was selected,
number of truly active and truly inactive variables $\widehat{\gamma}$}
\label{tab:pcorrect_scen2}
\end{table}

\begin{table}
\begin{center}
\begin{tabular}{lrrrrr}  \hline
  Method & $n$ & $\widehat{\gamma}=\gamma^*$ & $\pi(\gamma^* \mid y,\widehat{\gamma}=\gamma^*)$ & Truly active selected & Truly inactive selected \\  \hline
 AFT-pMOMZ & 100 & 0.504 & 0.357 & 1.54 & 0.26 \\
 AFT-pMOMZ & 500 & 0.972 & 0.829 & 2.00 & 0.03 \\
 AFT-Zellner & 100 & 0.364 & 0.663 & 1.29 & 0.05 \\
 AFT-Zellner & 500 & 1.000 & 0.963 & 2.00 & 0.00 \\
 Cox-piMOM & 100 & 0.368 & 0.398 & 1.45 & 0.31 \\
 Cox-piMOM & 500 & 0.896 & 0.727 & 2.00 & 0.12 \\
 AFT-LASSO & 100 & 0.004 &  & 1.01 & 7.16 \\
 AFT-LASSO & 500 & 0.004 &  & 1.01 & 7.16 \\
 Cox-LASSO & 100 & 0.020 &  & 1.83 & 8.37 \\
 Cox-LASSO & 500 & 0.000 &  & 2.00 & 14.17 \\
   \hline
\end{tabular}
\end{center}
\caption{Scenario 3, $p=50$. Proportion of correctly selected models ($\widehat{\gamma}=\gamma^*$) for uncensored and censored data,
  average posterior model probability when the correct model was selected,
number of truly active and truly inactive variables $\widehat{\gamma}$}
\label{tab:pcorrect_scen3}
\end{table}

\begin{table}
\begin{center}
\begin{tabular}{lrrrrr}  \hline
  Method & $n$ & $\widehat{\gamma}=\gamma^*$ & $\pi(\gamma^* \mid y,\widehat{\gamma}=\gamma^*)$ & Truly active selected & Truly inactive selected \\  \hline
 AFT-pMOMZ & 100 & 0.220 & 0.412 & 1.13 & 0.34 \\
 AFT-pMOMZ & 500 & 0.960 & 0.869 & 2.00 & 0.04 \\
 AFT-Zellner & 100 & 0.108 & 0.632 & 0.77 & 0.09 \\
 AFT-Zellner & 500 & 0.976 & 0.946 & 1.98 & 0.00 \\
 Cox-piMOM & 100 & 0.072 & 0.397 & 0.89 & 0.33 \\
 Cox-piMOM & 500 & 0.932 & 0.853 & 1.97 & 0.04 \\
 AFT-LASSO & 100 & 0.000 &  & 0.38 & 5.35 \\
 AFT-LASSO & 500 & 0.000 &  & 0.38 & 5.35 \\
 Cox-LASSO & 100 & 0.004 &  & 1.28 & 6.12 \\
 Cox-LASSO & 500 & 0.016 &  & 2.00 & 10.31 \\
   \hline
\end{tabular}
\end{center}
\caption{Scenario 4, $p=50$. Proportion of correctly selected models ($\widehat{\gamma}=\gamma^*$) for uncensored and censored data,
  average posterior model probability when the correct model was selected,
number of truly active and truly inactive variables $\widehat{\gamma}$}
\label{tab:pcorrect_scen4}
\end{table}


\begin{table}
\begin{center}
\begin{tabular}{lrrrrr}  \hline
  Method & $n$ & $\widehat{\gamma}=\gamma^*$ & $\pi(\gamma^* \mid y,\widehat{\gamma}=\gamma^*)$ & Truly active selected & Truly inactive selected \\  \hline
 AFT pMOMZ & 100 & 0.272 & 0.341 & 1.03 & 0.30 \\
 AFT pMOMZ & 500 & 0.940 & 0.835 & 1.99 & 0.07 \\
 AFT Zellner & 100 & 0.128 & 0.653 & 0.75 & 0.06 \\
 AFT Zellner & 500 & 0.976 & 0.945 & 1.98 & 0.01 \\
 Cox iMOM & 100 & 0.196 & 0.393 & 1.18 & 0.22 \\
 Cox iMOM & 500 & 0.956 & 0.763 & 2.00 & 0.04 \\
 AFT LASSO & 100 & 0.000 &  & 0.78 & 5.84 \\
 AFT LASSO & 500 & 0.000 &  & 0.78 & 5.84 \\
 Cox LASSO & 100 & 0.004 &  & 1.59 & 7.46 \\
 Cox LASSO & 500 & 0.008 &  & 2.00 & 12.73 \\
   \hline
\end{tabular}
\end{center}
\caption{Scenario 5, $p=50$. Proportion of correctly selected models ($\widehat{\gamma}=\gamma^*$) for uncensored and censored data,
  average posterior model probability when the correct model was selected,
number of truly active and truly inactive variables $\widehat{\gamma}$}
\label{tab:pcorrect_scen5}
\end{table}

\begin{table}
\begin{center}
\begin{tabular}{lrrrrr}  \hline
  Method & $n$ & $\widehat{\gamma}=\gamma^*$ & $\pi(\gamma^* \mid y,\widehat{\gamma}=\gamma^*)$ & Truly active selected & Truly inactive selected \\  \hline
 AFT pMOMZ & 100 & 0.084 & 0.310 & 0.67 & 0.36 \\
 AFT pMOMZ & 500 & 0.788 & 0.808 & 1.81 & 0.03 \\
 AFT Zellner & 100 & 0.032 & 0.449 & 0.36 & 0.12 \\
 AFT Zellner & 500 & 0.572 & 0.851 & 1.57 & 0.00 \\
 Cox iMOM & 100 & 0.044 & 0.360 & 0.66 & 0.46 \\
 Cox iMOM & 500 & 0.764 & 0.830 & 1.78 & 0.03 \\
 AFT LASSO & 100 & 0.000 &  & 0.31 & 4.40 \\
 AFT LASSO & 500 & 0.000 &  & 0.31 & 4.40 \\
 Cox LASSO & 100 & 0.000 &  & 0.98 & 5.33 \\
 Cox LASSO & 500 & 0.008 &  & 1.99 & 9.36 \\
   \hline
\end{tabular}
\end{center}
\caption{Scenario 6, $p=50$. Proportion of correctly selected models ($\widehat{\gamma}=\gamma^*$) for uncensored and censored data,
  average posterior model probability when the correct model was selected,
number of truly active and truly inactive variables $\widehat{\gamma}$}
\label{tab:pcorrect_scen6}
\end{table}

\clearpage
\subsection{TGFB data}\label{app:tgfb}

\begin{figure}[h!]
\begin{center}
\begin{tabular}{cc}
\includegraphics[width=0.5\textwidth]{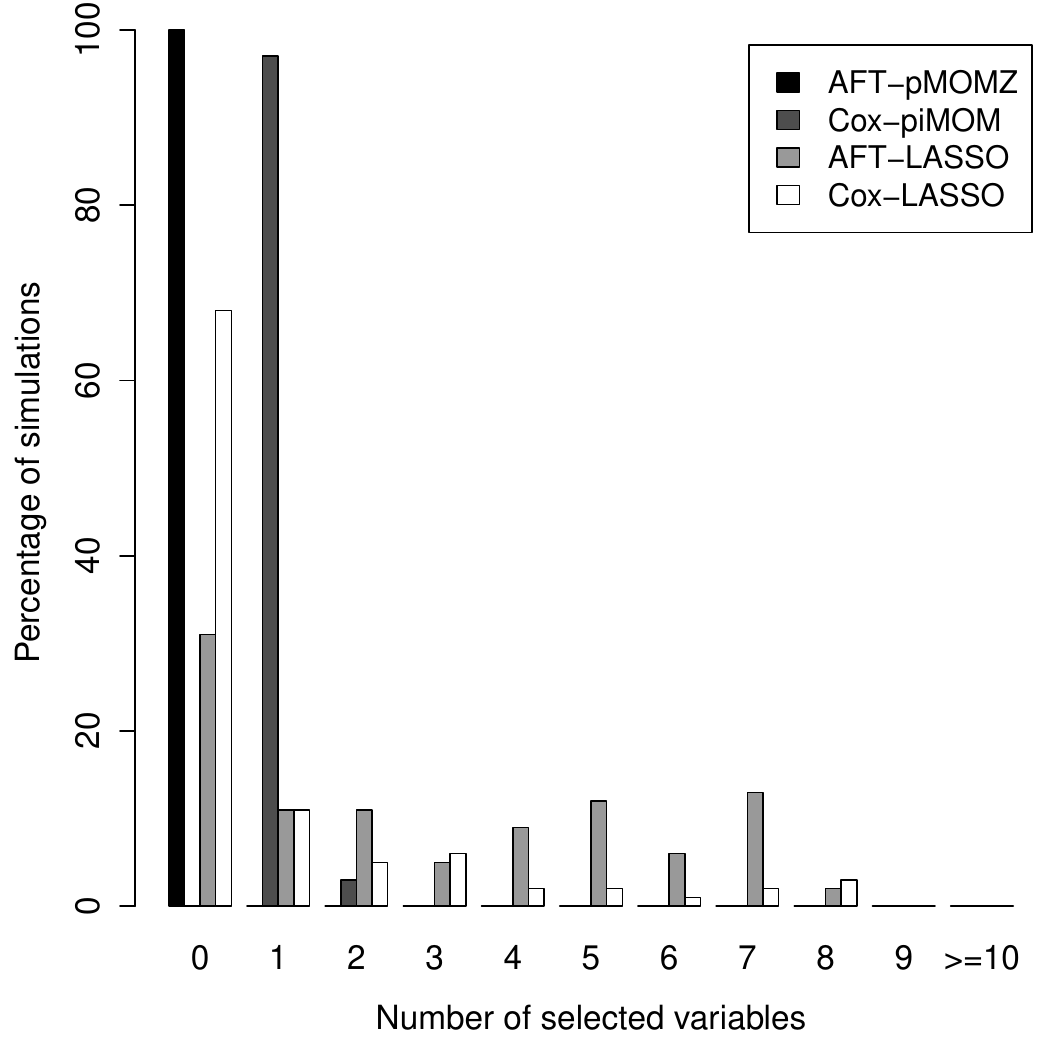} &
\includegraphics[width=0.5\textwidth]{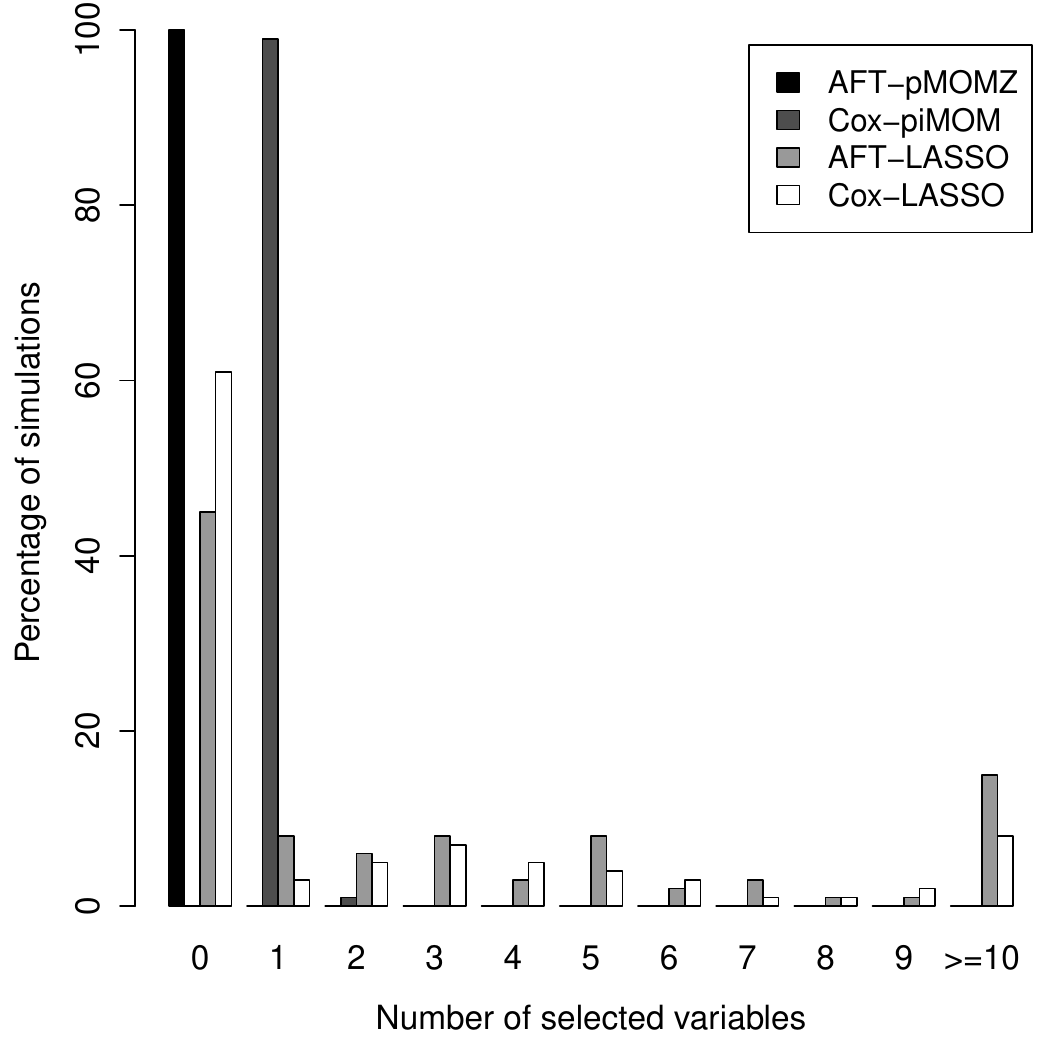}
\end{tabular}
\end{center}
\caption{Number of false positives in permuted colon cancer data (100 permutations) when
the design had 8 columns (left) for stage, linear and non-linear effect of TGFB and 175 columns (right) for stage and linear effect of 173 genes}
\label{fig:falsepos_tgfb}
\end{figure}

\begin{table}
\begin{center}
\begin{tabular}{|l|cc|cc|}\hline
& \multicolumn{2}{c|}{Linear effects} & \multicolumn{2}{c|}{Linear and non-linear} \\
AFT-pMOMZ & 0.64 & 3.9  & 0.66  &  4.9\\
Cox-piMOM & 0.53 & 2.1  & 0.67  &  1.0\\
Cox-LASSO & 0.69 & 13.6 & 0.68  & 11.7\\
AFT-LASSO & 0.51 & 20.8 & 0.45  & 14.9\\
\hline
\end{tabular}
\end{center}
\caption{Leave-one-out cross-validation on colon cancer data.
Concordance index (CI) and average number of selected parameters
when considering only linear effects $\gamma_j \in \{0,1\}$ $(p=175)$
and linear/non-linear effects $\gamma_j \in \{0,1,2\}$ $(p(r+1)=1050)$
}
\label{tab:ci_tgfb}
\end{table}

\begin{figure}[h!]
  \begin{center}
    \begin{tabular}{cc}
      \includegraphics[width=0.48\textwidth]{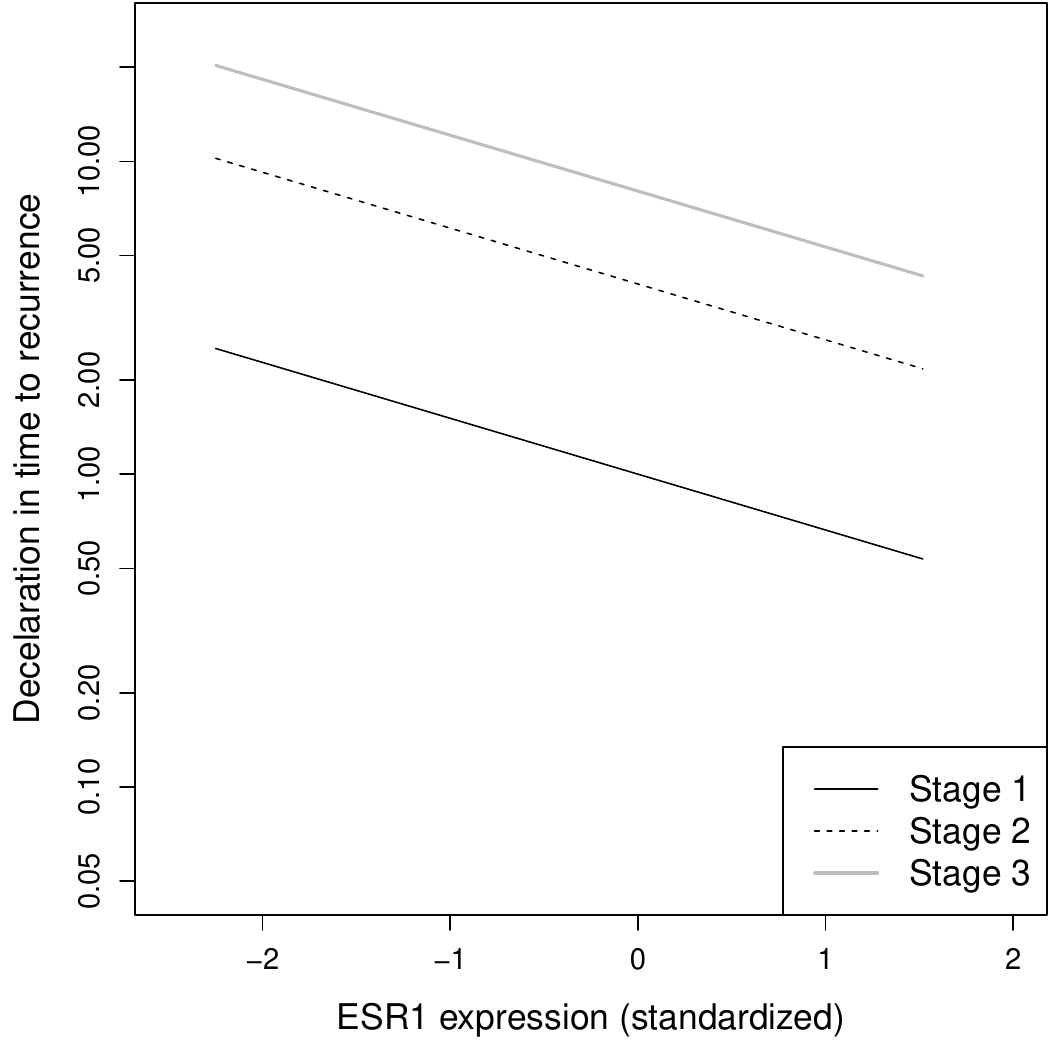} &
      \includegraphics[width=0.48\textwidth]{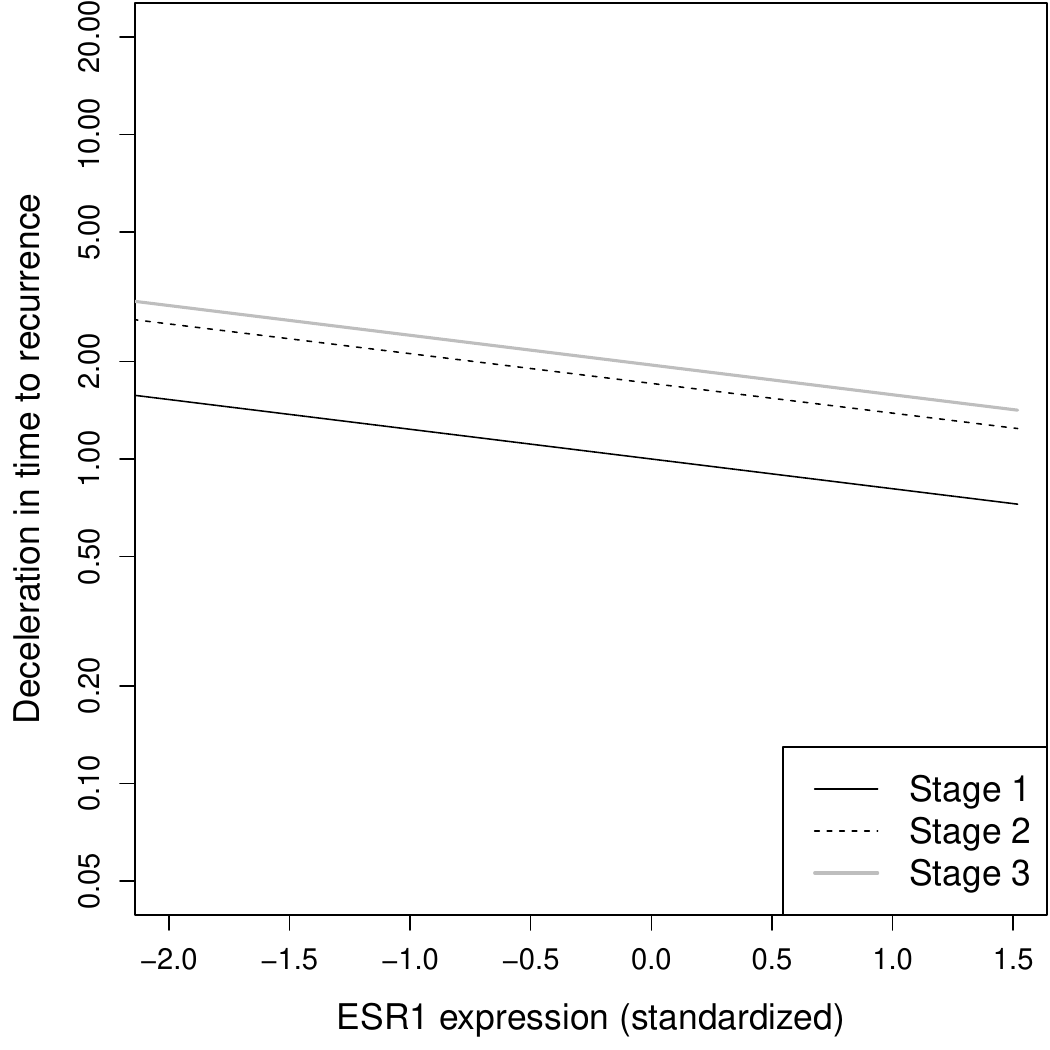}
    \end{tabular}
\end{center}
\caption{Estimated time deceleration due to stage and ESR1 expression (MLE). Model with stage and ESR1 (left) and stage, ESR1 with 8 genes with highest marginal posterior inclusion (right).
ESR1 P-values: 0.003 (left), 0.1 (right)}
\label{fig:esr1}
\end{figure}

\begin{table}[ht]
\small
\begin{tabular}{ccccccc}
\hline
Gene & NUSAP1 & Contig46452\_RC & LINC01520 & NM\_001109 & NM\_003430 & NM\_006398  \\
Prob & 0.965 & 0.411 & 0.306 & 0.286 & 0.258 & 0.217  \\
Gene &  NM\_006727 & GC11M123574 & Contig40158\_RC & NM\_001565 & NM\_004131 & Contig51202\_RC \\
Prob &  0.217 & 0.161 & 0.158 & 0.124 & 0.108 & 0.102 \\
\hline
\end{tabular}
\caption{Genes with marginal posterior probability greater than $10\%$.}
\label{tab:esr1}
\end{table}

Figure \ref{fig:tgfb} shows the estimated time acceleration as a function of tumor stage and the expression of gene TGFB, as estimated via maximum likelihood estimation, for the model with only stage and TGFB (left) and the highest posterior probability model (right).

\begin{figure}[h!]
  \begin{center}
    \begin{tabular}{cc}
      \includegraphics[width=0.48\textwidth]{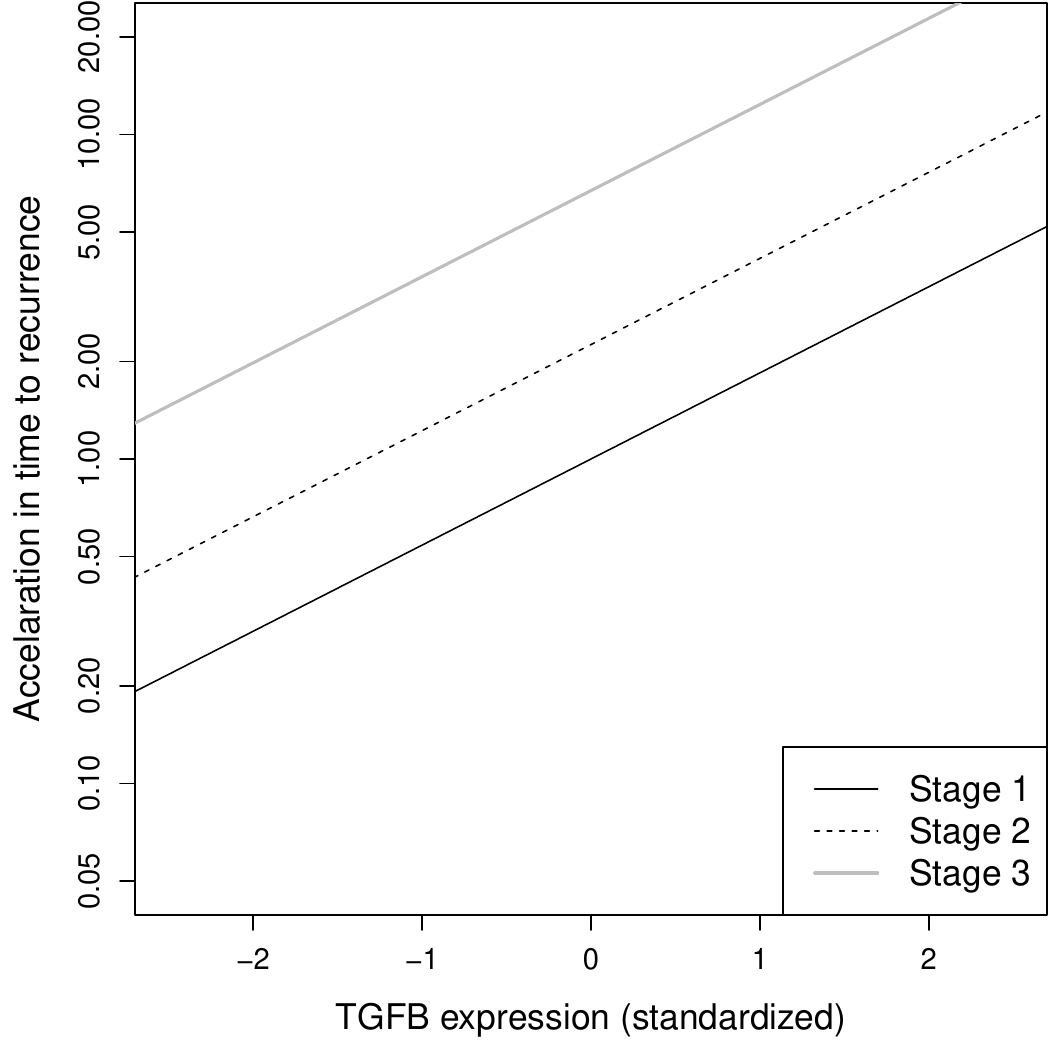} &
      \includegraphics[width=0.48\textwidth]{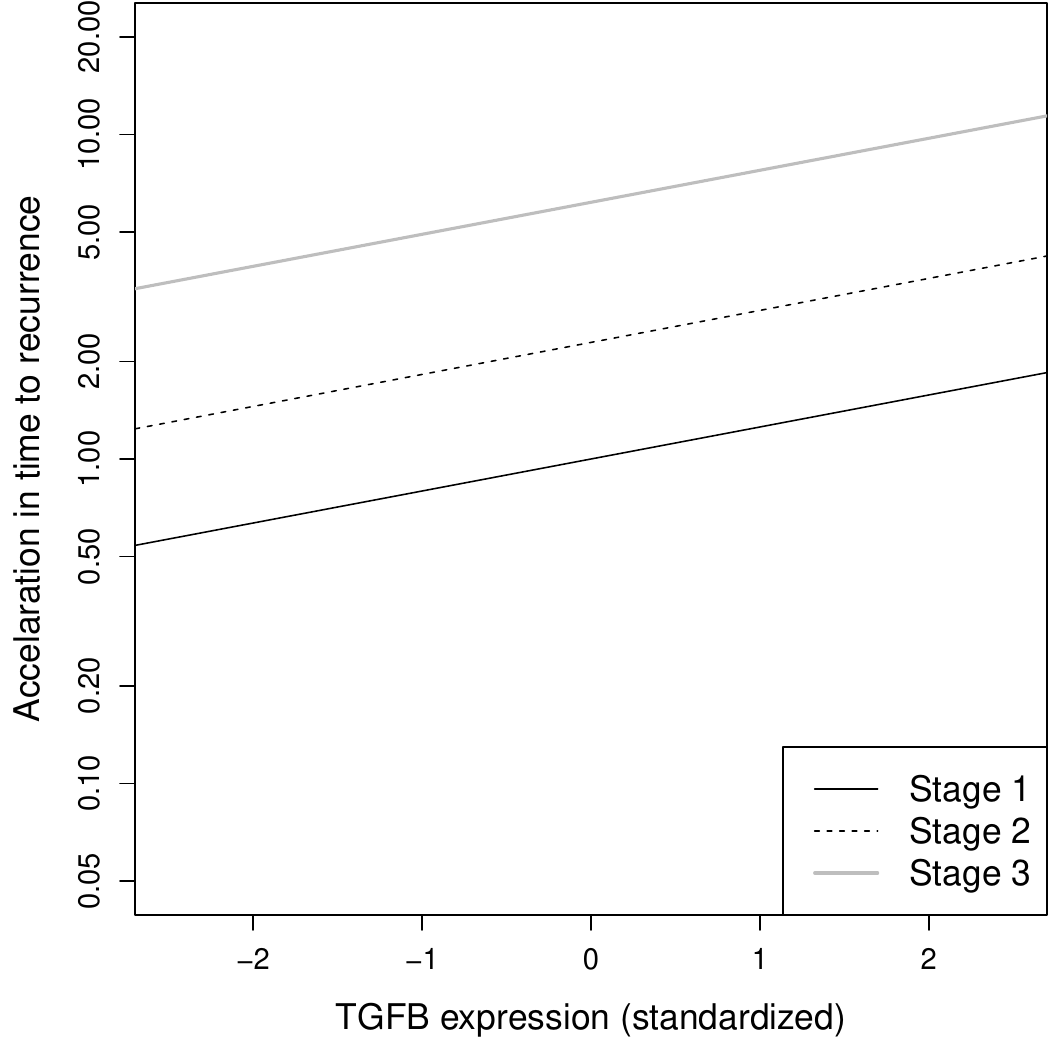}
    \end{tabular}
\end{center}
\caption{Estimated time acceleration due to stage and TGFB expression (maximum likelihood). Model with stage and TGFB (left) and stage, TGFB, FLT1, ESM1, GAS1 (right).
TGFB P-values: 0.001 (left), 0.281 (right)}
\label{fig:tgfb}
\end{figure}

\pagebreak
\subsection{Effect of ESR1 on breast cancer prognosis}\label{app:esr}

The expression of the estrogen receptor alpha gene (ESR1) is associated with therapeutic resistance and metastasis in breast cancer \citep{lei:2019}.
\cite{lum:2013} studied the effect of 1,553 genes and ESR1 on survival (time until recurrence).
We analyzed their $n=272$ patients with recorded survival.
As covariates we used tumor stage (2 dummy indicators), age, chemotherapy (binary), ESR1 and the 1,553 genes, for a total of $p=1,557$.

We first used AFT-pMOMZ to perform selection only on stage and ESR1.
The top model had 0.974 posterior probability and included stage and a linear effect of ESR1, confirming that ESR1 is associated with prognosis.
The posterior marginal inclusion probability for a non-linear effect of ESR1 was only 0.002.
As an additional check, the MLE under the top model gave P-values $<0.001$ for stage and ESR1.
Patients in stages 2 and 3 were estimated to experience recurrence 4.053 and 8.030 times faster (respectively) relative to stage 1,
and a unit standard deviation in ESR1 was associated with a 0.663 deceleration in recurrence time.
Figure \ref{fig:esr1} (left) shows the estimated deceleration.

Next, we extended the exercise to all 1,557 variables.
The top model contained age and the genes NUSAP1, GC11M123574 and Contig46452\_RC.
Table \ref{tab:esr1} shows the 10 genes with largest posterior inclusion probabilities.
Interestingly, the marginal inclusion probability for ESR1 was only 0.008,
\textit{i.e.} after accounting for other genes ESR1 did not show an effect on survival.
For confirmation we fitted via MLE the model with the top 8 genes in Table \ref{tab:esr1}, stage and ESR1.
The P-value for ESR1 was 0.1 and its estimated effect was substantially reduced (Figure \ref{fig:esr1}, right).
As a further check, there appear to be plausible mechanisms how the selected genes may mediate the effect of ESR1 on survival.
NUSAP1 plays a role in spindle microtubule organization that promotes cell proliferation and metastasis
(genecards.org, \cite{stelzer:2016}),
GC11M123574 is positively correlated with carcinoma-promising cell cycle regulator CKS1
and is regulated by cancer-associated transcription factor BCLAF1 \citep{lee_eunkyoung:2011},
and Contig46452\_RC belongs to a set of 100 reporter genes related to ESR1 selected for their prognostic power on breast cancer progression
\citep{vantveer:2002}.

%
%
%
%

\clearpage

\bibliographystyle{unsrtnat}
\bibliography{references}

\end{document}